\documentclass[11pt]{article}
\usepackage[margin=1in]{geometry}
\usepackage{amssymb,amsfonts,amsmath,amsthm,amscd,dsfont,mathrsfs,bbold}
\usepackage{blkarray}
\usepackage{graphicx,float,psfrag,epsfig,color}
	\usepackage{microtype}
	\usepackage[pdftex,pagebackref=true,colorlinks]{hyperref}
	\hypersetup{linkcolor=[rgb]{.7,0,0}}
	\hypersetup{citecolor=[rgb]{0,.7,0}}
	\hypersetup{urlcolor=[rgb]{.7,0,.7}}

\newcommand\independent{\protect\mathpalette{\protect\independent}{\perp}} 
\def\independent#1#2{\mathrel{\rlap{$#1#2$}\mkern2mu{#1#2}}}

\newcommand{\supp}{\mathrm{supp}} 
 
\newcommand{\F}{\mathbb{F}}

\newcommand{\im}{\mathrm{Im}}

\newcommand{\e}{\varepsilon}

\newcommand{\rank}{\mathrm{rank}}

\newcommand{\I}{\mathcal{I}}

\newcommand{\V}{\mathcal{V}}

\newcommand{\wt}{\text{wt}}
\newcommand{\A}{\mathcal{A}}

\newcommand{\poly}{\text{poly}}
\newcommand{\1}{\mathbb{1}}

\theoremstyle{plain}
\newtheorem{thm}{Theorem}[section]
\theoremstyle{plain}
\newtheorem{claim}[thm]{Claim}
\theoremstyle{plain}
\newtheorem{defin}[thm]{Definition}
\theoremstyle{plain}
\theoremstyle{plain}
\newtheorem{lemma}[thm]{Lemma}
\theoremstyle{plain}
\newtheorem{corol}[thm]{Corollary}
\theoremstyle{plain}
\theoremstyle{remark}
\newtheorem{remark}{Remark}%[section]
\theoremstyle{discussion}
\theoremstyle{plain}
\theoremstyle{plain}
\newtheorem{question}[thm]{Question}

\begin{document}

\title{Reed-Muller codes for random erasures and errors}
\author{Emmanuel Abbe\thanks{Program in Applied and Computational Mathematics, and Department of Electrical Engineering, Princeton University, Princeton, USA, \texttt{eabbe@princeton.edu}} \and 
Amir Shpilka\thanks{Department of Computer Science,Tel-Aviv University, Tel-Aviv, Israel,
\texttt{shpilka@post.tau.ac.il}.  The research leading to these results has received funding
from the European Community's Seventh Framework Programme (FP7/2007-2013) under grant agreement number 257575, and from the Israel Science Foundation (grant number 339/10).} 
\and Avi Wigderson\thanks{Institute for Advanced Study, Princeton, USA, \texttt{avi@ias.edu}. This research was partially supported by NSF grant CCF-1412958.}}

%\author{
%%Emmanuel Abbe\thanks{School of Communication and Computer Sciences, Ecole Polytechnique F\'ed\'erale de Lausanne, Switzerland. Email: emmanuel.abbe@epfl.ch}
%}

%%\author{Emmanuel Abbe\thanks{Department of Computer and Communication Sciences, EPFL, Switzerland. Email: emmanuel.abbe@epfl.ch} \;\;\; and\;\;\;
%%Andrea Montanari\thanks{Department of Electrical Engineering,
%%and Department of Statistics,
%%Stanford University. Email: montanari@stanford.edu}}

\date{}

\maketitle

\begin{abstract}
This paper studies the parameters for which Reed-Muller (RM) codes over $GF(2)$ can correct random erasures and random errors with high probability, and in particular when can they achieve capacity for these two classical channels. Necessarily, the paper also studies properties of evaluations of multi-variate $GF(2)$ polynomials on random sets of inputs.

For erasures, we prove that RM codes achieve capacity both for very high rate and very low rate regimes. For errors, we prove that RM codes achieve capacity for very low rate regimes, and for very high rates, we show that they can uniquely decode at about square root of the number of errors at capacity.

The proofs of these four results are based on different techniques, which we find interesting in their own right. In particular, we study the following questions about $E(m,r)$, the matrix whose rows are truth tables of all monomials of degree $\leq r$ in $m$ variables. What is the most (resp. least) number of random columns in $E(m,r)$ that define a submatrix having full column rank (resp. full row rank) with high probability? We obtain tight bounds for very small (resp. very large) degrees $r$, which we use to show that RM codes achieve capacity for erasures in these regimes.

Our decoding from random errors follows from the following novel reduction. For every linear code $C$ of sufficiently high rate we construct a new code $C'$, also of very high rate, such that for every subset $S$ of coordinates, if $C$ can recover from erasures in $S$, then $C'$ can recover from errors in $S$. Specializing this to RM codes and using our results for erasures imply our result on unique decoding of RM codes at high rate.

Finally, two of our capacity achieving results require tight bounds on the weight distribution of RM codes. We obtain such bounds extending the recent~\cite{KLP} bounds from constant degree to linear degree polynomials.

\end{abstract}
\thispagestyle{empty}
\newpage

\tableofcontents

\thispagestyle{empty}
\newpage
% ------------------------------------------------------------

%\Anote{Another Introduction}
\pagenumbering{arabic}
\section{Introduction}

\subsection{Overview}

We start by giving a high level description of the background and motivation for the problems we study, and of our results.

Reed-Muller (RM) codes were introduced in 1954, first by Muller \cite{muller} and shortly after by Reed \cite{reed}, who also provided a decoding algorithm.  They are among the oldest and simplest codes to construct; the codewords are  the evaluation vectors of all multivariate polynomials of a given degree bound. More precisely, in an $RM(m,r)$ code over a finite field $\F$, 
a message is interpreted as the coefficients of a multivariate polynomial $f$ of degree at most $r$ over $\F$, and its encoding is  simply the vector of evaluations $f(a)$ for all possible assignments $a\in \F^m$ to the variables. Thus, RM codes are linear codes. They have been extensively studied in coding theory, and yet some of their most basic coding-theoretic parameters remain a mystery to date. Specifically, fixing the {\em rate} of an RM code, while it is easy to compute its tolerance to errors and erasures in the worst-case (or adversarial) model, it has proved extremely difficult to estimate this tolerance for even the simplest models of random errors and erasures. The questions regarding erasures can be interpreted from a learning theory perspective, about interpolating low degree polynomials from lossy or noisy evaluations. The questions regarding errors relate sparse recovery from random Boolean errors. This paper makes some progress on these basic questions.

Reed-Muller codes (over both large and small finite fields) have been extremely influential in the theory of computation, playing a central role in some important developments in several areas. In cryptography, they have been used e.g. in secret sharing schemes~\cite{Shamir79}, instance hiding constructions~\cite{BF90} and private information retrieval (see the survey \cite{Gasarch04}). In the theory of randomness, they have been used in the constructions of many pseudo-random generators and randomness extractors, e.g.\cite{bogdanov-viola}. These in turn were used for hardness amplification, program testing and eventually in various interactive and probabilistic proof systems, e.g. the celebrated results NEXP=MIP \cite{BFL90}, IP=PSPACE \cite{Sha92},  NP=PCP \cite{ALMSS98}. In circuit lower bounds for some low complexity classes one argues that every circuit in the class is close to a codeword, so any function far from the code cannot be computed by such circuits (e.g. \cite{Razborov}. In distributed computing they were used to design fault-tolerant information dispersal algorithms for networks \cite{Rabin1989efficient}. The hardness of approximation of many optimization problems is greatly improved by the ``short code'' \cite{barak2012making}, which uses  the optimal testing result of \cite{BKSSZ}. And the list goes on. Needless to say, the properties used in these works are properties of low-degree polynomials (such interpolation, linearity, partial derivatives, self-reducibility, heredity under various restrictions to variables, etc.), and in some of these cases, specific coding-theoretic perspective such as distance, unique-decoding, list-decoding, local testing and decoding etc. play important roles. Finally, polynomials are basic objects to understand computationally from many perspectives (e.g. testing identities, factoring, learning, etc.), and this study interacts well with the study of coding theoretic questions regarding RM codes.

To discuss the coding-theoretic questions we focus on, and give appropriate perspective, we need some more notation. First, we will restrict attention to binary codes, the most basic case where $\F = \F_2$, the field of two elements\footnote{This seems also the most difficult case for these questions, and we expect our techniques to generalize to larger finite fields.}. To reliably transmit $k$-bit messages we encode each by an $n$-bit codeword via a mapping $C: \F_2^k \rightarrow \F_2^n$. We abuse notation and denote by $C$ both the mapping and its image, namely the set of all codewords\footnote{A code is {\em linear} if the mapping $C$ is $\F_2$-linear, or equivalently if the set of codewords $C$ is a linear subspace of $\F_2^n$.} . The {\em rate} of $C$ is given by the ratio $k/n$, capturing the redundancy of the code (the smaller it is, the more redundant it is). A major problem of coding theory is to determine the largest rate for which one can uniquely recover the original message from a {\em corrupted} codeword (naturally, explicit codes with efficient encoding and decoding algorithms are desirable). This of course depends on the nature of corruption, and we shall deal here with the two most basic ones, erasures (bit-losses) and errors (bit-flips). Curiously, the two seminal papers from the late 1940s giving birth to coding theory, by Shannon \cite{shannon48} and Hamming \cite{hamming50} differ in whether one should consider recovery for {\em most} corruptions, of from {\em all} corruptions. In other words, Shannon advocates {\em average-case} analysis whereas Hamming advocates {\em worst-case} analysis.

In Hamming's worst case setting, recovery of the original message must be possible from every corruption of every codeword. In this model there is a single parameter of the code determining recoverability: the distance of the code. The {\em distance} of $C$ is the minimum Hamming distance of any two codewords in $C$ (the {\em relative distance} is simply the distance normalized by the block-length $n$). If the distance is $d$, then we one can uniquely recover from at most $d$ erasures and from $\lfloor (d-1)/2 \rfloor$ errors. This leaves the problem of finding the optimal trade-off between rate and distance, and designing codes which achieve this optimum. While these are still difficult open problems, we know a variety of codes that can simultaneously achieve constant rate and constant relative distance (such codes are often called {\em asymptotically good}). In contrast, Reed-Muller codes fall far short of that. The rate of $RM(m,r)$ is ${m \choose \leq r}/2^m$, while the distance is easily seen to be $2^{m-r}$. Thus making any one of them a positive constant makes the other exponentially small in $n$. In short, from a worst-case perspective, RM codes are pretty bad.

In Shannon's average-case setting (which we study here), a codeword is subjected to a random corruption, from which recovery should be possible {\em with high probability}. This random corruption model is called a {\em channel}, and the best achievable rate is called the {\em capacity} of the channel. The two most basic ones, the Binary Erasure Channel (BEC) and the Binary Symmetric Channel (BSC), have a parameter $p$ (which may depend on $n$), and corrupt a message by independently replacing, with probability $p$, the symbol in each coordinate, with a ``lost" symbol in the BEC($p$) channel, and with the complementary symbol in the BSC($p$) case. Shannon's original paper already contains the optimal trade-off achievable for these (and many other channels). For {\em every} $p$, the capacity of BEC($p$) is $1-p$, and the capacity of BSC($p$) is $1-h(p)$, where $h$ is the binary entropy function.\footnote{$h(p) = -p\log_2(p) - (1-p)\log_2(1-p)$, for $p\in (0,1)$, and $h(0)=h(1)=0$.} While Shannon shows that random codes achieve this optimal behavior,\footnote{The fact that random linear codes are optimal for symmetric channels was shown in \cite{elias}.} explicit and efficiently encodable and decodable codes achieving capacity in both channels\footnote{For the case of the BEC, \cite{luby1} provides the first LDPC codes that are capacity-achieving, and further LDPC ensembles have been recently developed with spatial coupling \cite{spatial2}.} have been obtained \cite{forney-thesis}, among which are the recent {\em Polar Codes} \cite{arikan} that we shall soon discuss.

Do Reed-Muller codes achieve capacity for these natural channels (despite their poor rate-distance trade-off)? The coding theory community seems to believe the answer is positive, and conjectures to that effect were made\footnote{The belief that RM codes achieve capacity is much older, but we did not trace back where it appears first.} in \cite{forney-road,arikan-RM,mondelli-RM}. However, to date, we do not know {\em any} value of $p$ for which RM codes achieve the capacity for erasures or errors! This paper provides the first progress on this conjecture, resolving it for very low rates and very high rates (namely for polynomials of degrees $r$ which are very small or very large compared to the number of variables $m$). 
%These results can be interpreted from learning theory perspective, showing that in these ranges of parameters any degree $r$ polynomial in $m$ variables can be uniquely interpolated with high probability from its values on the minimum possible number of random inputs. Another interpretation of our results regards optimal recovery of random sparse vectors via $E(m,r)$. 
Our results unfortunately fall short of approaching the cases where the corruption rate $p$ is a constant, the most popular regime in coding theory.

%\Enote{I changed ``most interesting" to ``most popular" in the above. This may be further changed. There is a new trend these days in information theory to study performances of codes very close to capacity. In part due to the fact that for some real channels (e.g., fiber optics), we can now push rates very close to capacity. So looking at extremal regimes may not be always less interesting :)}

%\Anote{Explain why average case analysis is currently hot  in TCS - learning, avg case hardness in crypto and proof complexity, planted problems, and understanding the behavior of algs on "typical" inputs.}

The conjecture that RM codes achieve capacity has been experimentally ``confirmed" in simulations \cite{arikan-RM,mondelli-RM}. Moreover, despite being extremely old, new interest in it resurged a few years ago with the advent of polar codes \cite{arikan}. To explain the connection between the two, as well as some of the technical problems arising in proving the results above, consider the following $2^m \times 2^m$ matrix $E_m$ (for ``evaluation''). Index the rows and columns by all possible $m$-bit vectors in $\F_2^m$ in {\em lexicographic order}. Interpret the columns simply as points in $\F_2^m$, and the rows as monomials (where an $m$-bit string correspond to the monomial which is the product of variables in the positions containing  a $1$). Finally, $E_m(x,y)$ is the value of the monomial $x$ on the point $y$ (namely it is 1 if the set of 1's in $x$ is contained in the set of 1's in $y$). Thus, every row of $E_m$ is the truth table of one monomial. It is thus easy to see that the code $R(m,r)$ is simply the span of the top (or ``high weight'')  $k$ rows of $E_m$, with $k= {m \choose \leq r}$; these are the truth tables of all degree $\leq r$ polynomials. In contrast, polar codes of the same rate are spanned by a different set of $k$ rows, so they form a different subspace of polynomials. While the monomials indexing the polar code rows have no explicit description (so far), they can be  computed efficiently for any $k$ in $\poly(n) = 2^{O(m)}$ time. It is somehow intuitively ``better" to prefer higher weight rows to lower weight ones as the basis of the code (as the ``chances of catching an error" seem higher). 
Given the amazing result that polar codes achieve capacity, this intuition seems to suggest that RM codes do so as well. 
In fact, experimental results in \cite{mondelli-RM} suggest that RM codes may outperform polar codes for the BEC and BSC with maximum-likelihood\footnote{ML decoding looks for the most likely codeword. For the BEC, this requires inverting a matrix over $GF(2)$, whereas for the BSC, ML can be approximated by a successive list-decoding algorithm.} decoding. 

%\Enote{Added the last phrase above.}

Denoting by $E(m,r)$ the top submatrix of $E_m$ with $k= {m \choose \leq r}$ rows, one can express some natural problems concerning it which are essential for our results. To obtain some of our results on achieving capacity for the erasure channel, we must understand the following two natural questions regarding $E(m,r)$. First, what is the largest number $s$ so that $s$ random columns of $E(m,r)$ are linearly independent with high probability. Second, what is the smallest number $t$ such that $t$ random columns have full row-rank. Capacity achieving for erasures means that $s=(1-o(1))k$ and $t=(1+o(1))k$, respectively. We prove that this is the case for small values of $r$. The second property gives directly the result for low-rate codes RM($m,r$), and the first implies the result for high-rate codes using a duality property of RM codes. Both results may be viewed from a learning theory perspective, showing that in these ranges of parameters any degree $r$ polynomial in $m$ variables can be uniquely interpolated with high probability from its values on the minimum possible number of random inputs.

%\Amnote{Maybe explain here what does it mean to decode BEC and BSC in terms of solving linear algebra equations. BEC can be solved, BSC says there is a unique sparse solution, not clear how to find it.}
%\Amnote{I now think that this will fit better in a section about RM}

%\Enote{I added the paragraph below to address previous comment by Amir (although the second comment says now it should perhaps come later?). I think that is good to mention these in the intro.}

For errors, further analysis is needed beyond the rank properties discussed above. From the parity-check matrix viewpoint, decoding errors is equivalent to solving (with high probability) an underdetermined system of equations. Recall that a linear code can be expressed as the null space of an $(n-k) \times n$ parity-check matrix $H$. If $Z$ is a random error vector with about (or at most) $s$ one's corrupting a codeword, applying the parity-check matrix to the codeword yields $Y=HZ$, where the ``syndrom'' $Y$ is of lower dimension $n-k$. Decoding random errors means reconstructing $Z$ from $Y$ with high probability, using the fact that $Z$ is sparse (hence the connection with sparse recovery). Note however that this differs from compressed sensing, as $Z$ is random and $HZ$ is over $GF(2)$. It relates to randomness extraction in that a capacity achieving code should produce an output $Y$ of dimension $m \approx n h(s/n)$ containing\footnote{See \cite{ITA-corr} for further discussion on this.} all the entropy of $Z$.  
Compared to the usual notion of randomness extraction, the challenge here is to extract with a very simple map $H$ (seedless and linear), while the source $Z$ is  much more structured, i.e. it has i.i.d.\ components, compared to sources in the traditional extractor settings.

%\Enote{Add potential references for previous discussion}

%\Enote{I changed ``almost nothing'' to ``very little'' in the second phrase below, to make sure we don't offend anyone (even though limited, there were some prior works)}

Another extremely basic statistics of a code is its {\em weight distribution}, namely, approximately how many codewords have a given Hamming weight.  Amazingly enough, very little was known about  the weight distribution of Reed-Muller code until the recent breakthrough paper of \cite{KLP}, who gave nearly tight bounds for constant degree polynomials for both. The results of \cite{KLP} also apply to list-decoding of RM codes, which was previously investigated in \cite{GopalanKZ08}. We need a sharpening of their upper bound for two of our results, which we prove by refining their method. The new bound is nearly tight not only for constant degree polynomials, but actually remains so even for degree $r$ that is linear in $m$. We get a similar improvement for their bound on the list-size for list decoding of RM codes.  

Summarizing, we study some very basic coding-theoretic questions regarding low-degree polynomials over GF(2). We stress two central aspects which remain elusive. First, while proving the first results about 
parameters of RM codes which achieve capacity, the possibly most important range, when error rate is constant, seems completely beyond the reach of our techniques. Second, while our bounds for erasures immediately entails a (trivial) efficient algorithm to actually locate them, there are no known efficient algorithms for correcting random errors in the regimes we prove it is information theoretically possible. We hope that this paper will inspire further work on the subject, and we provide concrete open questions it suggests. We now turn to give more details on the problems, results and past related work.

%%%%%%%%%%%%%%%%%%%%%%%%%%%%%%%%%%%%%
%%%%%%%%%%%%%%%%%%%%%%%%%%%%%%%%%%%%%
%%%%%%%%%%%%%%%%%%%%%%%%%%%%%%%%%%%%%
%%%%%%%%%%%%%%%%%%%%%%%%%%%%%%%%%%%%%

\subsection{Notation and terminology}\label{sec:notation}
Before presenting our results we need to introduce some notations and parameters. 
%Recall that $RM(m,r)$ denotes the Reed-Muller code of blocklength $n=2^m$ and order $r$. This means that the code contains the evaluation on the hypercube $\F_2^m$ of all $m$-variate polynomials of degree at most $r$. 
%We denote by $H(m,r)$ the parity-check matrix of $RM(m,m-r)$. The reason for working with $RM(m,m-r)$ is purely notational; we use the parity-check matrix representation of the code, which has $\sum_{i=0}^{r} {m\choose i}$ rows for $RM(m,m-r)$ (hence, the number of rows behaves like $m^r$ for small $r$). 
The following are used throughout the paper:
\begin{itemize}
\item For nonnegative integers $r \leq m$, $RM(m,r)$ denotes the Reed-Muller code 
whose codewords are the evaluation vectors of all multivariate polynomials of degree at most $r$ on $m$ Boolean variables. 
The maximal degree $r$ is sometimes called the order of the code. The blocklength of the code is $n=2^m$, the dimension $k=k(m,r)=\sum_{i=0}^r {m \choose i}\triangleq {m\choose \leq r}$, and the distance $d=d(m,r)=2^{m-r}$. The code rate is given by $R=k(m,r)/n$.
\item We use $E(m,r)$ to denote the ``evaluation matrix'' of parameters $m,r$, whose rows are indexed by all monomials of degree $\leq r$ on $m$ Boolean variables, and whose columns are indexed by all vectors in $\F_2^m$. For $u\in \F_2^m$, we denote by $u^r$ the column of $E(m,r)$ indexed by $u$, which is a $k$-dimensional vector, and for a subset of columns $U \subseteq \F_2^m$ we denote by $U^r$ the corresponding submatrix of $E(m,r)$.

\item A generator matrix for $RM(m,r)$ is given by $G(m,r)=E(m,r)$, and a parity-check matrix for $RM(m,r)$ is given by $H(m,r)=E(m,m-r-1)$ (see Lemma \ref{duality}). 

%\Enote{I changed above that $E(m,r)$ becomes $H(m,m-r-1)$, if viewed as a parity-check matrix (instead of $H(m,r)$). That is, I kept from now on that $H(m,r)$ is the PCM of the $RM(m,r)$ code, just like $G(m,r)$ is its generator matrix, which is the standard notation. I also adopted parenthesis instead of squared brackets. I tried to be consistent in the rest of the paper, but will double-check these more carefully when reading the whole paper on a printed version.}

\item We associate with a subset $U\subseteq \F_2^m$ its characteristic vector $\1_U \in \{0,1\}^n$.  
We often think of the vector $\1_U$ as denoting either an {\em erasure pattern} or an {\em error pattern}. 
\end{itemize}

Finally, we use the following standard notations. $[n]=\{1,\dots,n\}$. The Hamming weight of $x \in \F_2^n$ is denoted $w(x)=|\{i \in [n] : x_i \neq 0\}|$ and the relative weight is $\wt(x)=w(x)/n$ . We use $B(n,s)=\{x \in \F_2^n : w(x) \leq s\}$ and $\partial B(n,s)=\{x \in \F_2^n : w(x) = \lceil s \rceil \}$. 
We use ${[n] \choose s}$ to denote the set of subsets of $[n]$ of cardinality $s$. 
Hence, for $S \in {[n] \choose s}$, $\1_S \in \partial B(n,s)$.

For a vector $x$ of dimension $n$ and subset $S$ of $n$, we use $x[S]$ to denote the components of $x$ indexed by $S$, and if $X$ is matrix with $n$ columns, we use $X[S]$ to denote the subset of columns indexed by $S$. In particular, $E(m,r)[U]=U^r$.  When we need to be more explicit, for an $a\times b$ matrix $A$ and $I\subseteq [a]$, we denote with $A_{I,\cdot}$ the matrix obtained by keeping only those rows indexed by $I$, and denote similarly $A_{\cdot, J}$ for $J\subseteq [b]$.\\

%\Anote{Throughout below I replaced the confusing erasure-rate and error-rate by erasure probability and error probability, and generally corruption probability}
%\Amnote{added the paragraph below on BEC and BSC. Should we write a more formal definition?}

%\Enote{Error-rate and erasure-rate is used when you don't work with the i.i.d.\ model. 
%I changed the phrase below, saying that we describe things here and define them formally in Section \ref{prelim}}

\noindent
{\bf Channels, capacity and capacity-achieving codes}

We next describe the channels that we will be working with, and provide formal definitions in Section \ref{prelim}. Throughout $p$ will denote the corruption probability per coordinate. The Binary Erasure Channel (BEC) with parameter $p$ acts on vectors $v\in\{0,1\}^n$, by changing every coordinate to $\text{``?''}$ with probability $p$. That is, after a message $v$ is transmitted in the BEC the received message $\hat{v}$ satisfies that for every coordinate $i$ either $\hat{v}_i=v_i$ or $\hat{v}_i=\text{``?''}$ and $\Pr[\hat{v}_i=\text{``?''}]=p$. The Binary Symmetric Channel (BSC) with parameter $p$ is flips the value of each coordinate with probability $p$. That is, after a message $v$ is transmitted in the BSC the received message $\hat{v}$ satisfies  $\Pr[\hat{v}_i\neq v_i]=p$. 

In fact, we will use a small variation on these channels; for corruption probability $p$ we will fix the number of erasures/errors to $s=pn$. We note that by the Chernoff-Hoeffding bound (see e.g., \cite{AlonSpencer}), the probability that more than $pn + \omega(\sqrt{pn})$ erasures/errors occur for independent Bernoulli choices is  $o(1)$, and so we can restrict our attention to talking about a fixed number of erasures/errors. Thus, when we discuss $s$ corruptions, we will take the corruption probability to be $p=s/n$. We refer to Section \ref{prelim} for the details. 

%\Amnote{modified paragraphs below}

We now define the notions of ``capacity-achieving'' for the channels above. We consider $RM(m,r)$ where $r=r(m)$ typically depends on $m$. We say that $RM(m,r)$ can correct random erasures/errors, if it can correct the random erasures/errors with high probability when $n$ tends to infinity. % (see Section \ref{} for the specific models). 
The goal is to recover from the largest  amount of erasures/errors that is information-theoretically achievable.
% with $r$ has large as possible. 
We note that while recovering from erasures, whenever possible, is always possible efficiently (by linear algebra), this need not be the case for recovery from errors. As we focus on the information theoretic limits, we allow maximum-likelihood (ML) decoding rule. Obtaining an efficient algorithm is a major open problem. Note that ML minimizes the error probability for equiprobable messages, hence if ML fails to decode the codewords with high probability, no other algorithms can succeed. 
%For erasures, ML is efficient as it simply amounts to inverting a matrix over $GF(2)$. 

%\Anote{These paragraphs before the tables were too cumbersome and repeated defs like the rate of an RM code. I simplified below.
%Is there  a Latex way of avoiding the annoying way in which p inside parenthesis () becomes $?$? }

%The table below provides the conditions for $RM(m,r)$ to achieve capacity at low and high code-rates, for both erasures and errors. These correspond to translating the results of Shannon (which are typically stated for constant erasure or error probabilities). 

Recall that the capacity of a channel is the largest possible code rate at which we can recover (whp) from corruption probability $p$. This capacity  is given by $1-p$ for BEC erasures, and by $1-h(p)$ for BSC errors. Namely,  Shannon proved that for any code of rate $R$ that allows to correct corruptions of probability $p$, then $R< 1-p$ for the BEC and $R< 1-h(p)$ for the BSC. 

Achieving capacity means that $R$ is close to the upper bound, say within $(1+\e)$ factor of the optimal bounds above. For {\em fixed} corruption probabilities $p$ and rates $R$ in $(0,1)$ this is easy to define (previous paragraph). However as we deal with  very low or very high rates above, defining this needs a bit more care, and is described in the table below, and formally in Section \ref{prelim}. A code of rate $R$ is $\e$-close to achieve capacity if it can correct from a corruption probability $p$ that satisfies the bounds below\footnote{Note that for $R\to 0$, in the BEC we have $p\to 1$, while for the BSC we have $p\to \frac12$. Also, we have stated the bounds thinking of $R$ fixed and putting a requirement on $p$. One can equivalently fix $p$ and require the code to correct a corruption probability $p$ for a rate $R$ that satisfies the bounds in the table.}. It is capacity-achieving if it is $\e$-close to achieve capacity for all $\e >0$.

%For Reed-Muller codes we naturally parametrize their rate via polynomial degree. We first fix $r$ in some regime, for example $r \leq m/100$, and then see how many random erasures/errors we can handle in that regime (when $m$ tends to infinity).
%Recalling that $k(m,r)= {m \choose \leq r}$ and $R=k(m,r)/n$ are respectively the code dimension and code rate of $RM(m,r)$, the low code-rate regime corresponds to the regimes of $m,r$ where $R \to 0$, and the high code-rate regime corresponds to $R \to 1$. 

%Denote with $s(m,r)$  the number of random erasures/errors, and let $p$ be the erasure- or error-probability (depending on the channel), i.e., $p=s(m,r)/n$. We say that $RM(m,r)$ is $\e$-close to achieving capacity if it can correct $s(m,r)$ random erasures/errors whenever $s(m,r)$ satisfies:   

%$RM(m,r)$ is capacity-achieving in an extremal regime if it can correct $s(m,r)$ random erasures/errors whenever $\e>0$ and $s(m,r)$ satisfies:   

%\Enote{I updated the table below. I didn't change the inequalities because Avi changed the above phrase, but I changed the sign of $\e$ in the second row (otherwise there is an issue I think). Also note that we no longer quite follow how we state our results, where we first fix the rate and see how many corruptions we can support (but that should be fine).}
\begin{center}
  \begin{tabular}{ |c |  c | c  | }\hline
  %\backslashbox{Rate}{Channel}
   & BEC & BSC \\
  %\hline &&\\
  % Rate $\backslash$ Channel & BEC & BSC \\
   \hline
  && \\
%Low code-rate   & $\frac{k(m,r)}{n}  \sim (1-\frac{s(m)}{n}) (1- \e)$  & $\frac{k(m,r)}{n}  \sim (1-H(\frac{s(m)}{n}))(1-\e)$ \\
Low code-rate ($R \to 0$)   & $p \geq 1- R (1+\e)$  & $h(p) \geq 1- R(1+\e)$ \\
&& \\ 
   %&(Erd\H{o}s and R\'enyi '60) & (Erd\H{o}s and R\'enyi '60) \\ 
    \hline 
    &&\\
%    High code-rate & $1-\frac{k(m,r)}{n}  \sim \frac{s(m)}{n} (1- \e)$  & $1-\frac{k(m,r)}{n}  \sim H(\frac{s(m)}{n})(1-\e)$ \\ 
    High code-rate ($R \to 1$) & $p \geq (1-R) (1- \e)$  & $h(p)   \geq (1-R)(1-\e)$ \\ 
    &&\\
    \hline
  \end{tabular}
\end{center}

\subsection{Our results}\label{sec:our results}
%\Enote{Should we make this a section? and put related literature just before as a subsection of the intro}

We now state all our results, with approximate parameters, as the exact statements (given in the body of the paper) are somewhat technical. We divide this section to results on decoding from random erasures, then on weight distribution and list decoding, and finally decoding random errors. In brief, we investigate four cases: two regimes for the code rates (high and low rates) and two models (BEC and BSC). Besides for the BSC at high-rate, we obtain a capacity-achieving result for all other three cases. For the low-rate regimes, we obtain results for values of $r$ up to the order of $m$.

%\Enote{If we remove the summary table below, I would suggest putting the above phrase.}
\subsubsection{Random erasures - the BEC channel}

As mentioned earlier, some of the questions we study concerning properties of Reed-Muller codes can be captured by the following basic algebraic-geometric questions about evaluation vectors of low-degree monomials, namely, submatrices of $E(m,r)$. For any parameters $r \in  [m]$ (the degree) and $s\in [n]$ (the size of the corrupted set $U$),  we will study correcting random erasures and errors  patterns of size $s$ in $RM(m,r)$.

%\Amnote{Updated questions to match those in Section 5}

%\begin{question}%\label{question:columns}
\begin{enumerate}
\item What is the largest $s$ for which the submatrix $U^r$ has full column-rank with high probability?
%\end{question}
%\begin{question}%\label{question:rows}
\item What is the smallest $s$ for which the submatrix $U^r$ has full row-rank with high probability?
%\end{question}
\end{enumerate}

%\begin{enumerate}
%\item For a random point-set $U \subseteq \{0,1\}^m$ of size $s$, what is the probability that the corresponding degree $r$ evaluation matrix, $U^r$, has linearly independent {\em columns} (i.e., full column-rank)? This is investigated in Section \ref{first-question}.
%\item For a random point-set $U \subseteq \{0,1\}^m$ of size $s$, what is the probability that the corresponding degree $r$ evaluation matrix, $U^r$, has linearly independent {\em rows} (i.e., full row-rank)? This is investigated in Section \ref{second-question}.
%\end{enumerate}

More generally, we will be interested in characterizing sets $U$ for which these properties hold. We note that for achieving capacity, $s$ should be as close as possible to ${m\choose \leq r}$ (from below for the first question and from above for the second question). In other words, the matrix $U^r$ should be as close to square as possible. Note that this would be achieved for the case where $E(m,r)$ is replaced by a random uniform matrix, so our goal in a sense is to show that $E(m,r)$ behaves like a random matrix with respect to these questions. 

We obtain our decoding results for the BEC by providing answers to these questions for certain ranges of parameters. 
%We begin by discussing capacity achieving parameters for erasures.
Our first theorem concerns Reed-Muller codes of low degree. 
%(EA: we should adapt the terminology with what we set in the intro.)

%\Anote{I think the original parameters were complemented - I complemented back}
\begin{thm}[See Theorem~\ref{thm:random_span_using_KL}]\label{thm:intro:low_deg_BEC}
%There exists a universal constant $0<\eta<1$ such that for $r \leq \eta m$, 
Let $r=o(m)$. Then, If we pick uniformly at random a set $U$ of $(1+o(1))\cdot {m\choose \leq r}$ columns of $E(m,r)$, then with probability $1-o(1)$ the rows of this submatrix are linearly independent, i.e., $U^r$ has full row-rank. 
% i.e., the columns of this submatrix span the entire space $\{0,1\}^{{m\choose \leq r}}$.
\end{thm}

%\Enote{So the above theorem and corollary below need to be changed now?}
%\Amnote{Fixed it. Changed Cor to Thm}

As an immediate corollary we get that Reed-Muller codes of sub-linear degree achieve capacity for the BEC.

%\Amnote{Need to write a corollary in the main body}

\begin{thm}[See Corollary~\ref{corol_low_rate_BEC}]\label{thm:intro:RM_for_BEC_low_degree}
%There exists a universal constant $0<\eta<1$ such that for $r \leq \eta m$, 
For $r=o(m)$, $RM(m,r)$ 
%can correct $2^m-(1-o(1)){m\choose \leq r}$ uniform erasures with probability $1-o(1)$.  
achieves capacity for the {BEC}. 
More precisely, for every
$\delta>0$ and  $\eta = O(1/\log(1/\delta))$  the following holds:  
For every $r \leq\eta m$, $RM(m,r)$ is $\delta$-close to capacity for the {BSC}

\end{thm}
%Moreover,  $RM(m,r)$ are $\delta$-close to capacity when $r=\eta m$, where $\eta=O(1/\log(1/\delta))$.

%\Amnote{In the above I changed ``close to achieve'' to ``achieves''}

We obtain similar results in a broader range of parameters when the code is of high degree rather than low degree (i.e., the code has high rate rather than low rate).

\begin{thm}[See Theorem~\ref{thm:lin_ind_eval_vectors}]\label{thm:intro:high_deg_BEC}
Let $r = O(\sqrt{m/\log m})$. If we pick uniformly at random a set $U$ of $(1-o(1))\cdot{m\choose \leq r}$ columns of $E(m,r)$, then with  probability $1-o(1)$ they are linearly independent, i.e., the submatrix $U^r$ has full column rank.
\end{thm}

Due to the duality between linear independent set of columns in $E(m,r)$ and spanning sets in the generating matrix of $RM(m,m-r-1)$ (see Lemma~\ref{lem:equiv2}) we get as  corollary that Reed-Muller codes with the appropriate parameters achieve capacity for the BEC.

%\Enote{changed $RM(m,m-r)$ into $RM(m,m-r-1)$ above. Also, changed the theorem below to be about $RM(m,r)$ instead of $RM(m,m-r)$}

%\Amnote{maybe put BSC, BEC in bold face? perhaps best to use a macro and then decide}

%\Anote{We also need to decide consistently if they are italic or not. Also note that they are not formally defined. Indeed, they are defined by random corruptions at some probability $p$ per coordinate, whereas we study corruptions of $k$ random coordinates for fixed $k$. WE should make the argument that results for our model imply the same result with essentially the same parameters for BEC,BSC by concentration. }
%\Amnote{I added something in Section 1.2 to address Aviâs comments. What I wrote is pretty informal though. B.t.w., I donât think that we italicise BEC/BSC. This is done automatic in the theorem environment.}

\begin{thm}[See Corollary~\ref{corol:RM_for_BEC}]\label{thm:intro:RM_for_BEC_high_degree}
For $m-r=O(\sqrt{m/\log m})$, $RM(m,r)$ is capacity-achieving on the BEC.
%Let $r = O(\sqrt{m/\log m})$. Then, $RM(m,m-r)$ achieves capacity for the {BEC}. 
%\Amnote{Maybe remove the I.e. and only keep it in the formal theorem everywhere.}
%I.e., for $1-p=(1+o(1)){m\choose \leq r}/n$, $RM(m,r)$ can correct $pn$ uniform erasures with probability $1-o(1)$.  
\end{thm}
%\Enote{We could also say ``Let $m-r=O(\sqrt{m/\log m})$, then $RM(m,r)$ is capacity-achieving for the BEC.'' That way $r$ is always the same parameter}

\subsubsection{Weight distribution and list decoding}

Before moving to our results on random errors, we take a detour to discuss our results on the weight distribution of Reed-Muller codes as well as their list decoding properties. These are naturally important by themselves, and, furthermore, tight weight distribution bounds turns out to be crucial for achieving capacity for the BEC in Theorem~\ref{thm:intro:low_deg_BEC} above, as well as for achieving capacity for the BSC in Theorem~\ref{thm:intro-low-BSC} below. Our bound extends an important recent 
%\Amnote{Isn't the word ``seminal'' too strong here?}  
result of Kaufman, Lovett and Porat on the weight-distribution of Reed-Muller codes \cite{KLP}, using a simple variant of their technique. Kaufman et al. gave a bound that was tight for $r=O(1)$, but degrades as $r$ grows. Our improvement extends this result to degrees $r=O(m)$. Denote with $W_{m,r}(\alpha)$ the number of codewords of $RM(m,r)$ that have at most $\alpha$ fraction of nonzero coordinates. 

\begin{thm}[See Theorem~\ref{thm:wt-dist}]\label{thm:intro:wt-dist}
Let $1\leq \ell \leq r-1 < m/4 $ and $0 < \e \leq 1/2$. 
Then, 
$$W_{m,r}((1-\e)2^{-\ell}) \leq  (1/\e)^{O\left(\ell^4 {m-\ell \choose \leq r-\ell} \right)}.$$
\end{thm}

%\Amnote{Moved the result about list-decoding here.}
As in the paper of \cite{KLP}, almost the exact same proof as our proof of Theorem~\ref{thm:intro:wt-dist} yields a bound for list-decoding of Reed-Muller codes, for which we get similar improvements. 
Following \cite{KLP} we denote:
$$L_{m,r}(\alpha) = \max_{g:\F_2^m\to \F_2}\left| \{ f \in RM(m,r) \mid \wt(f-g)\leq \alpha\}\right|.$$
That is, $L_{m,r}(\alpha)$ denotes the maximal number of code words of $RM(m,r)$ in a hamming ball of radius $\alpha 2^m$. The bound concerns $\alpha$ of the form $(1-\e)2^{-\ell}$ for $1\leq \ell \leq r-1$, and our main contribution is making the first factor in the exponent depend on $\ell$ (rather than on $r$ in~\cite{KLP}).

\begin{thm}\label{thm:list-dec}
Let $1\leq \ell \leq r-1$ and $0 < \e \leq 1/2$. %Set $\alpha=(1-\e)2^{-\ell}$. 
Then, 
if $r \leq m/4$ then
$$L_{m,r}((1-\e)2^{-\ell}) \leq (1/\e)^{O\left(\ell^4 {m-\ell \choose \leq r-\ell} \right)}.$$
%where $c$ is the absolute constant from Lemma~\ref{lem:KLP}.
\end{thm}

%\Amnote{Should we say the following:}
%As the proof is very similar to the proof of Theorem~\ref{thm:intro:wt-dist} we omit it.

\subsubsection{Random errors - the BSC channel}

We now return to discuss decoding from random errors.
Our next result shows that Reed-Muller codes achieve capacity also for the case of random errors at the low rate regime. The proof of this result  relies on Theorem~\ref{thm:intro:wt-dist}.

\begin{thm}[See Theorem~\ref{thm:errors_from_weights}]\label{thm:intro-low-BSC}
For $r=o(m)$, $RM(m,r)$ achieves capacity for the BSC. More precisely, for every
$\delta>0$ and  $\eta = O(1/\log(1/\delta))$  the following holds:  
For every $r \leq\eta m$, $RM(m,r)$ is $\delta$-close to capacity for the {BSC}. 
%I.e., for and any $p$ satisfying  
%\begin{align}
%1-H(p)=(1+\delta)R, \quad \text{where } R=\frac{{m\choose \leq r}}{n},
%\end{align}
%$RM(m,r)$ can correct $pn$ random errors with probability at least 
%$1-o(1)$.
\end{thm}
%\Anote{achieving capacity, and "close" to achieving capacity is not defined}

%Thus, in this parameter regime of Theorem~\ref{thm:intro-low-BSC}, $RM(m,r)$ achieves capacity for the BSC.

To obtain results about the behavior of high-rate Reed-Muller codes with respect to random errors we use a novel connection between robustness to errors and robustness to erasures in related Reed-Muller codes. 

%\Amnote{Write a theorem in Section~\ref{sec:deg-r}}

\begin{thm}[See Theorem~\ref{thm:erasures_to_errors}]\label{thm:intro:erasures_to_errors}
%[See Section~\ref{sec:deg-r}]\label{thm:intro:erasures_to_errors}
If a set of columns $U$ are linearly independent in $E(m,r)$ (namely, $RM(m,m-r-1)$ can correct the {\em erasure} pattern $\1_U$), then the {\em error} pattern $\1_U$ can be corrected (i.e., it is uniquely decodable) in $RM(m,m-(2r+2))$.
\end{thm}

%\Amnote{I changed $RM(m,m-(2r+1))$ to $RM(m,m-(2r+2)$. Please verify}
%\Enote{Amir: I agree with $m-(2r+2)$ since we need $E(m,2r+1)$ for the parity-check matrix}
%

Using Theorem~\ref{thm:intro:high_deg_BEC} this gives a new result on correcting random errors in Reed-Muller codes. 

\begin{thm}[See Theorem~\ref{thm:main_for_BSC}]\label{thm:intro:main_for_BSC}
For $r= O(\sqrt{m/\log m})$, $RM(m,m-(2r+2))$ can correct a random error pattern of weight $(1-o(1))\cdot{m\choose \leq r}$ with probability larger than $1-o(1)$.
\end{thm}
%\Enote{Perhaps compare the above statement with the one I wrote it in the summary table.}

While this result falls short of showing that Reed-Muller codes achieve capacity for the BSC in this parameter range, it does show that they can cope with many more errors than suggested by their minimum distance. Recall that the minimum distance of $R(m,m-(2r+2))$ is $2^{2r+2}$. Achieving capacity for this code means that it should be able to correct roughly ${m \choose 2r}$ random errors. Instead we show that it  can handle roughly ${m \choose \leq r}$ random errors, which is approximately the square root of the number of errors at capacity. \\

%\Enote{Shouldn't it be ${m \choose 2r}$ instead of ${m \choose 2r+1}$ above?}

The proof of Theorem~\ref{thm:intro:erasures_to_errors} reveals a more general phenomenon, that of reducing error correction to erasure correction. We prove that for any linear code $C$, of very high rate, there is another linear code $C'$ of related high rate, so that if $C$ can correct the {\em erasure} pattern $\1_U$ then $C'$ can correct the {\em error} pattern $\1_U$. Furthermore $C'$ is very simply defined from $C$. The decline in quality of $C'$ relative to $C$ is best explained in terms of the co-dimension (namely the number of linear constraints on the code, or equivalently the number of rows of its parity-check matrix). We prove that the co-dimension of $C'$ is roughly the cube of the co-dimension of $C$. We now state this general theorem.

%that the connection between correcting from {\em erasures} (for some linear code) to correcting from {\em errors} (in a related linear code) does not only hold for RM codes, but actually extends to {\em any} linear code! We show the following. 
For a matrix $H$ we denote by $H^r$ the corresponding matrix that contains the evaluations of {\em all} columns of $H$ by all degree $\leq r$ monomials (in an analogous way to the definition of $U^r$ from $U$).

\begin{thm}[See Theorem~\ref{thm:general_erasures_tensored_to_errors}]\label{thm:intro:erasures_tensored_to_errors}
If a set of columns  $U$ is linearly independent in a parity check matrix $H$, then the code that has $H^3$ as a parity check matrix can correct the error pattern $\1_U$.
\end{thm}

Note that applying this result as is to $E(m,r)$ would give a weaker statement than Theorem~\ref{thm:intro:erasures_to_errors}, in which $E(m,2r+1)$ would be replaced by $E(m,3r)$. We conclude by showing that
this result is tight, namely replacing $3$ by $2$ in the theorem above fails, even for RM codes. 

\begin{thm}[See Section~\ref{sec:counterex}]\label{thm:intro:counterex}
There are subsets of columns $U$ that are linearly independent in $E(m,1)$, but such that the patterns $\1_U$ are not uniquely decodable in $E(m,2)$.
\end{thm}

\subsection{Proof techniques}\label{sec:techniques}

Although the statements of Theorems~\ref{thm:intro:high_deg_BEC} and \ref{thm:intro:low_deg_BEC} sound very similar, 
%albeit for a ``complementary'' range of parameters, 
their proofs are very different. We first explain the ides behind the proofs of these two theorems and then give details for the proofs of Theorems~\ref{thm:intro:wt-dist}, \ref{thm:intro-low-BSC}, \ref{thm:intro:erasures_to_errors} and 
\ref{thm:intro:erasures_tensored_to_errors}.

%\Amnote{I reordered the paragraph according to the theorem number}

\paragraph{Proof of Theorem~\ref{thm:intro:high_deg_BEC}}

%\Anote{Here I thought that the explanation was too long and technical for the intro, and I tried to simplify it. The way it was written, with refs to specific lemmas,  would suit very well a "plan of the proof" before the actual one in Section 7. I commented out the original parts I removed.}

The proof of Theorem~\ref{thm:intro:high_deg_BEC} 
%is more involved and requires new ideas. It focuses 
relies
on estimating the size of varieties (sets of common zeros) of linear subspaces of degree $r$ polynomials. Here is a high level sketch. 

Recall that we have to show that if we pick a random set of points $U\subset \F_2^m$, of size $(1-o(1))\cdot {m\choose \leq r}$, and with each point associate its degree-$r$ evaluation vector, then with high probability these vectors are linearly independent. While proving this is simple when considered over large fields, it is quite non-trivial  over very small fields.  We are able to prove that this property holds for degrees $r$ up to (roughly) $\sqrt{m/\log m}$. It is a very interesting question to extend this to larger degrees as well.

To prove that a random set $U$ of appropriate size gives rise to linearly independent evaluation vectors we consider the question of what it takes for a new point to generate an evaluation vector that is linearly independent of all previously generated vectors. As we prove, this boils down to understanding what is the probability that a random point is a common zero of all degree $r$ polynomials, in a certain linear space of polynomials defined by the previously picked points. If this set of common zeros is small, then the success probability (i.e., the probability that a new point will yield an independent evaluation vector) is high, and we can iterate this argument.
%(Lemmas~\ref{lem: new independent point} and \ref{lem: common zero in U^r}). More accurately, if we already picked $k$ evaluation vectors then we need to understand how many common zeros do polynomials in a linear space of dimension $m-k$ have (Lemma~\ref{lem: number of common zeroes}). Indeed, every point which is not a common zero gives rise to a vector that is linearly independent of the already chosen $k$ evaluation vectors. 

To bound the number of common zeros we yet again move to a dual question. Notice that if a set of $K$ linearly independent polynomials of degree $r$ vanishes on a set of points $V$, then there are at most ${m\choose \leq r}-|K|$ linearly independent degree $r$ polynomials that are defined over $V$. In view of this, the way to prove that a given set of polynomials does not have too many common zeros is to show that any large set of points (in our case, the set of common zeros) has many linearly independent degree $r$ polynomials that are defined over it. We give two different proofs of this fact. The first uses a hashing argument;  if $V$ is large then after some linear transformation it supports many different degree $r$ {\em monomials}. The second relies on a somewhat tighter bound that was obtained by Wei \cite{Wei}, who studied the generalized Hamming weight of Reed-Muller codes. While Wei's result gives slightly tighter bounds compared to the hashing argument, we find the latter argument more transparent.
% and so we give it in the appendix. 
  
%\Anote{Is it "generalized Hamming dist?}
  
% (roughly) contains a large ``ball'' around some point. This implies that there are many linearly independent polynomials of degree $r$ that are defined over it (Claim~\ref{cla: A contains ball}). 

%We now explain how everything fits together. Denote by $S$ the set of $k$ points that we already picked and let $\I(S)$ be the set of degree-$r$ polynomials vanishing on $S$ and $\V(\I(S))$ the set of common zeroes of all polynomials in $\I(S)$. We would like to show then, that unless $S$ is very large, $\V(\I(S))$ is small (and therefore a random point will not be in  $\V(\I(S))$, giving rise to an evaluation vector that is linearly independent of the evaluations $S^r$). Now, if $\V(\I(S))$ is not small then, as we claimed, there are many linearly independent degree $r$ polynomials that are defined over it. Hence, $|\I(S)|$ cannot be large. As $|\I(S)| = {m\choose \leq r}-|S|$ this implies that $S$ is large, which is what we wanted to show.  

%\Anote{Again I commented out some technical details which I felt are too much for an intro}

\paragraph{Proof of Theorem~\ref{thm:intro:low_deg_BEC}}
To prove Theorem~\ref{thm:intro:low_deg_BEC} we first observe that a set of columns $U$ (in $E(m,r)$) spans the entire row-space if and only if there is no linear combination of the rows of $E(m,r)$ that is supported on the complementary set $U^c=\F_2^m \setminus U$. As linear combinations of rows correspond to ``truth-tables'' of degree $r$ polynomials, this boils down to proving that, with high probability, no nonzero degree $r$ polynomial vanishes on all points in $U$. For each such polynomial, if we know its weight (the number of nonzero values it takes), this is a simple calculation, and the hope is to use a union bound over all polynomials. To this end, we can partition the codewords to dyadic intervals according to their weights, carry out this calculation and union bound the codewords in each interval and then combine the results. For this plan to work we need a good estimate of the number of codewords in each dyadic interval, which is what Theorem~\ref{thm:intro:wt-dist} gives.

%\Amnote{I moved the discussion of Theorem~\ref{thm:intro:wt-dist} to a different paragraph and made it more technical as I thought that it conveys no information in its current version. It is also ok if we leave it to the main body of the paper}

\paragraph{Proof of Theorem~\ref{thm:intro:wt-dist}}
As mentioned earlier, this theorem improves upon a beautiful result of Kaufman, Lovett and Porat \cite{KLP}. Our proof is closely related to their proof. Roughly, what they show is that any small weight codeword, i.e., a degree $r$ polynomial with very few non-zero values, can be well approximated by a ``few'' partial derivatives. Namely, there is a function that when applied to a few  lower degree polynomials, agrees with the original polynomial on most of the inputs. Here ``few'' depends on the degree $r$, the weight and (crucially for us) the quality of the approximation.   
Kaufman et al. then pick an approximation quality parameter that guarantees that the approximating function can be close to at most {\em one} polynomial of degree $r$. Then, counting the number of possible approximating functions they obtain their bound. The cost is that such a small approximation parameter blows the number of ``few'' derivatives that are required. We diverge from their proof in that we choose a much larger approximation quality parameter, but rather allow each approximating function to be close to {\em many} degree $r$ polynomials. The point is that, by the triangle inequality, all these polynomials are close to each other, and so subtracting one of them from any other still yield polynomials of very small weight. Thus, we can use induction on weight to bound their number, obtaining a better bound on the number of polynomials of a given weight. 

%
%In a beautiful paper, Kaufman, Lovett and Porat \cite{KLP} (Theorem~\ref{thm:KLP}) were the first to give a nearly-tight bound on the number of codewords (i.e. truth-tables of degree $r$ polynomials) of each weight, which makes our plan possible. However, their result deteriorates too fast as a function the degree. To obtain our result we needed a tight estimate, which rather deteriorates as a function of the (minus logarithm of the) relative weight\footnote{Which is at most the degree, but can be much smaller!}. Our proof of the more general bound follows the same outline as the proof of \cite{KLP} but deviates at two points as we explain later in the proof of Theorem~\ref{thm:wt-dist}. 
%
%Combining our improvement of \cite{KLP} with the outline above we get the result.
%
%%For each interval we use a simple probabilistic argument to claim that we hit a nonzero of each polynomial. E.g., if the number of polynomials that have at least a fraction of $\alpha$ nonzeros (and at most $2\alpha$ fraction of nonzeros) is $A(\alpha)$ then w.h.p., if we sample roughly $\frac{1}{\alpha}\cdot \log A(\alpha)$ many evaluation points then we are likely to hit a nonzero for every such polynomial. Finally, a union bound argument gives the result. 

\paragraph{Proof of Theorem~\ref{thm:intro-low-BSC}}
Here we use a coarse upper-bound on the error probability, as for the proof that a random linear code achieves capacity, and show that the argument still holds for RM codes. To prove that a random code can, w.h.p., uniquely decode an error pattern $\1_U$ of weight $w$ we basically wish to show that for no other error pattern $\1_V$, of weight $w$, the vector $\1_U\oplus \1_V$ is a code word (as then both error patterns will have the same syndrome). Stated differently, we want to count how many  different ways are there to represent a codeword as a sum of two vectors of weight at most $w$. This counting depends very much on the weight of the codeword that we wish to split. In random linear codes weights are very concentrated around $n/2$, which makes a union bound easy. Reed-Muller codes however have many more codewords of smaller weights, and the argument depends precisely on how many. Once again Theorem~\ref{thm:intro:wt-dist} comes to the rescue and enables us to make this delicate calculation for each (relevant) dyadic interval of weights. Here too our improvement of \cite{KLP} is essential. 

\paragraph{Proofs of Theorems~\ref{thm:intro:erasures_to_errors},~\ref{thm:intro:main_for_BSC} and \ref{thm:intro:erasures_tensored_to_errors}}

%\Enote{all the $H()$ below should be $E()$ if we go with the standard notation (and we can mention that $E(m,r)=H(m-r-1)$), or we can change $H(m,k)$ into $H(m,m-k-1)$ everywhere - any preferences?}
%
%\Amnote{I changed some of the $H$ to $E$ but tried keeping both as we also talk about syndromes}

Consider an erasure pattern $\1_U$ such that the corresponding set of degree-$r$ evaluation vectors, $U^r$,  is linearly independent. Namely the columns indexed by $U$ in $E(m,r)$ are linearly independent. We would like to prove that $\1_U$ is uniquely decodable from its syndrome  under $H(m,m-2r-2)=E(m,2r+1)$. We  actually prove that if $\1_V$ is another erasure pattern,  which has the same syndrome under $H(m,m-2r-2)$, then $U=V$. The proof may be viewed as a reconstruction (albeit inefficient) of the set $U$ from its syndrome. Here is a high level description of our argument that different (linearly independent) sets of erasure patterns give rise to different syndromes.

We first prove this property for the case $r=1$ (details below). This immediately implies Theorem~\ref{thm:intro:erasures_tensored_to_errors} as every parity check matrix of any linear code is a submatrix of $E(m,1)$ for some $m$. This is a general reduction from the problem of recovering from errors to that of recovering from erasures (in a related code).
As a special case, it also implies that for any $r$, $H(m,m-3r-1)$ uniquely decodes any error pattern $\1_U$ such that the columns indexed by elements of $U$ in $E(m,r)=H(m,m-r-1)$ are linearly independent. We then slightly refine the argument for larger degree $r$ to  replace  $H(m,m-3r-1)$ above by $H(m,m-2r-2)$, which gives Theorem~\ref{thm:intro:erasures_to_errors}.

For the case $r=1$,
%\Amnote{Seems there was a repetition here so I removed my text}
% we first prove that if $\1_U$ and $\1_V$ have the same syndrome, when the columns corresponding to the points in $U$ are linearly independent, then the span of the columns in $V$ is the same as that of the columns in $U$. For this property it is actually enough to only consider $H(m,2)$ and not $H(m,3)$ (and in general $H(m,2r)$ and not $H(m,2r+1)$). We then show that the additional information that we get from $H(m,3)$ is enough to conclude that $U=V$.
the proof divides to two logical steps. In the first part we prove that the columns of $V$ must span the same space as the columns of $U$. This requires only the submatrix $E(m,2)$, namely at pairs of coordinates in each point (degree $2$ monomials). In the second part we use this property to actually identify each vector of $U$ inside $V$. This already requires looking at the full matrix $E(m,3)$, namely at triples of coordinates.
%Very roughly, to show that $U$ and $V$ have the same span let us assume that $U$ is the set of $m$ unit vectors $e_1,\ldots,e_m$. By considering the row corresponding to the monomial $x_1$ we see that there must be an odd number of columns in $V$ that have $1$ in their first coordinate. We now claim that those vectors must add to $e_1$. Indeed, if this is not the case and their sum has a nonzero entry, say at the second coordinate, then by considering the row corresponding to the monomial $x_1 x_2$ we will get that the syndrome of $V$ is different that that of $U$. In this fashion we show that $U$ and $V$ have the same span. To show that $e_1$ actually belongs to $V$ we look at triples of coordinates and therefore need the code $H(m,3)$. See Section~\ref{sec: deg 3} for more details. 

It is interesting that going to triples of coordinate is essential for $r=1$ (and so this result are tight).
We prove that even if the columns of $U$ are linearly independent, then there can be a different set $V$ that has the same syndrome in $E(m,2)$. 
%Thus, it is essential to go to $H(m,3)$ to get unique decoding in this case.
This result is given in Section~\ref{sec:counterex}.
We do not know what is the right bound for general $r$.

%%%%%%%%%%%%%%%%%%%%%%%%%%%%
%%%%%%%%%%%%%%%%%%%%%%%%%%%%
%%%%%%%%%%%%%%%%%%%%%%%%%%%%
%%%%%%%%%%%%%%%%%%%%%%%%%%%%

\subsection{Related literature} 
%\Enote{I changed this section, see if ok}
\noindent
{\bf Recovery from random corruptions}

Besides the conjectures mentioned in the introduction that RM codes achieve capacity, results fall short of that for all but very spacial cases.
We are not familiar of works correcting random erasures.
Several papers have considered the quality of  RM codes for correcting random errors when using specific algorithms, focusing mainly on efficient algorithms.
In \cite{krich}, the majority logic algorithm \cite{reed} is shown to succeed in recovering all but a vanishing fraction of error patterns of weight up to $d \log(d)/4$, where $d=2^{m-r}$ is the code distance, requiring however a positive rate $R>0$. This was later improved to weights up to $d \log(d)/2$ in \cite{dumer2}.  We note that for a fixed rate $0<R<1$, this is roughly $\sqrt{n}$, whereas to achieve capacity one should correct $\Omega(n)$ erasures/errors.

\sloppy A line of work by Dumer \cite{dumer1,dumer2,dumer3} based on recursive algorithms (that exploits the recursive structure of RM codes), obtains results mainly for low-rate regimes. In \cite{dumer1}, it is shown that for a fixed order $r$, i.e., for $k(m,r) =\Theta(m^r)$, an algorithm of complexity $O(n \log (n))$ can correct most error patterns of weight up to $n(1/2 - \e)$ given that $\e$ exceeds $n^{-1/2^r}$. In \cite{dumer2}, this is improved to errors of weight up to $n/2(1-(4m/d)^{1/2^r})$, requiring that $r/\log(m) \to 0$. 
Further, \cite{dumer3} shows that most error patterns of weight up to $n/2(1-(4m/d)^{1/2^r})$ can be recovered in the regime where $\log(m)/(m-r) \to 0$. 

%Note that correcting $s=n/2(1-(4m/d)^{1/2^r})$ random errors with a %scaling-capacity achieving
%\Amnote{changed from scaling capacity to capacity. Also changed $\asymp$ to $\Theta$.}
%capacity achieving code would require $k(m,r)/n =\Theta\left( 1-H(1/2(1-(4m/d)^{1/2^r}))\right) = \Theta\left((4m2^r/n)^{1/2^{r-1}}\right)$, whereas \cite{dumer1} reaches $k(m,r)/n =\Theta( m^r/n)$ which is much smaller. In particular, previous results are not 
%%scaling-capacity achieving. 
%capacity achieving. 

%\Amnote{When discussing Sidelâs result, maybe just say that it is not capacity-achieving because of the constant $2^r$. In any case I donât think that we need the scaling-capacity terminology}

Note that while the previous results rely on efficient decoding algorithms, they are far from being capacity-achieving. 
Concerning maximum-likelihood decoding, known algorithms are of exponential complexity in the blocklength, besides for special cases such as $r=1$ or $r=m-2$ \cite{litsyn}. In \cite{sidel}, it is shown that RM codes of fixed order $r$ can decode most error patterns of weight up to $n/2(1-\sqrt{c(2^r-1)m^r/ n r!}$, where $c> \ln(4)$. However, this does not provide a capacity-achieving result, which would require decoding most error patterns of weight approaching $n/2(1-\sqrt{\ln(4)m^r /n r!})$, i.e., \cite{sidel} has an extra $\sqrt{2^r-1}$ factor. 
%Note that this is scaling-capacity achieving and not far from being sharp-capacity achieving for fixed order $r$: the code rate is $k(m,r) \sim m^r / r!$, to be sharp-capacity achieving, one should be able to support an error rate $s/n$ such that $(1-H(s/n)-\e) \sim k(m,r)/n$ for all $\e>0$, or equivalently, $s/n \sim 1-H^{-1}(k(m,r)/n)-\e  \sim 1/2(1-\sqrt{\ln(4)k(m,r)/n})-\e \sim  1/2(1-\sqrt{\ln(4)m^r /n r!}) -\e$. Hence \cite{sidel} is off by a constant factor $2^{r}-1$ to be sharp-capacity achieving. 

For the special case of $r=1,2$ (i.e., the generator matrix has only vectors of weights $n$, $n/2$ and $n/4$), \cite{hell} shows that RM codes are capacity-achieving. For $r \geq 3$, the problem is left open.  \\

\noindent
{\bf Weight enumeration}

%The distance of $RM(m,r)$ is $2^{m-r}$,  
%This implies that $RM(m,r)$ can correct $2^{m-r}-1$ erasures and $2^{m-r-1}-1$ errors with probability 1. 
%Hence $RM(m,m-r)$ requires $r=\log(s)+2$ to correct $s$ errors with probability one (assuming $s$ is a power of 2). 
%In the case of $r=m/2$, the code has hence rate $1/2$ and a distance of roughly $\sqrt{n}$. 
%While this is a vanishing relative distance, it is already a non-trivial set of parameters. 
%which does not say much about the performance of the code in probabilistic settings. 
The weight enumerator (how many codewords are of any given weight) of $RM(m,2)$ was characterized in \cite{sloane-RM}. For $RM(m,3)$, a complete characterization of the weight enumerator is still missing.
The number of codewords of minimal weight is known for any $r$, and corresponds to the number of $(m-r)$-flats in the affine geometry $AG(m,2)$ \cite{sloane-book}. In \cite{kasami1}, the weight enumerator of RM codes is characterized for codewords of weight up to twice the minimum distance, later improved to 2.5 the minimum distance in \cite{kasami2}. 
%The codewords of minimal weights in $RM(m,r)$ can be characterized for all $m$ and $r$, they correspond to the affine subspaces of dimension $m-r$, and there are exactly $(...)$ of them \cite{}. 
%This implies in particular that for $s$ uniform errors ($s$ a power of 2), taking $r=\log(s)+1$ allows to correct these errors with high probability. 

For long, \cite{kasami2} remained the largest range for which the weight enumerator was characterized, until \cite{KLP} managed to breakthrough the $2.5$ barrier and obtained bounds for all distances in the regime of small $r$. The results of \cite{KLP} is given in Theorem \ref{thm:KLP}. %This represents to date the best bound for the weight enumerator of RM codes. \Anote{Aren't our results better?}

\subsection{Organization}

%The rest of the paper appears in the appendix following the references.
The paper is organized as follows. We first discuss the model of random erasures and errors (Section~\ref{prelim}) and then give a quick introduction to Reed-Muller codes (Section~\ref{RM-background}). In section~\ref{sec:wt-dist} we prove Theorem~\ref{thm:intro:wt-dist} on the weight distribution of RM codes. In Section~\ref{sec:sampling} we give answers to the two questions on sub matrices of $E(m,r)$, when $r$ is small.  In Section~\ref{sec:RM-erasures} we use the result obtained thus far to obtain our results for the BEC (Theorems~\ref{thm:intro:low_deg_BEC} and \ref{thm:intro:high_deg_BEC}). In Section~\ref{sec:BSC-proofs} we give our results for the BSC. 
%{\em 
Finally, in Section~\ref{sec:open}, we discuss some intriguing future directions and open problems which our work raises.
%}.

%\section*{Acknowledgements}
%We thank Venkatesan Guruswami for bringing \cite{Wei} to our attention. 
%Venkat corresponded with us during a Dagstuhl meeting, where we first presented our result. 

%\bibliographystyle{amsalpha}
%\bibliography{RMcode}

%%%%%%%%%%%%%%%%%%%%%%%%%%%%%%%
%%%%%%%%%%%%%%%%%%%%%%%%%%%%%%%
%%%%%%%%%%%%%%%%%%%%%%%%%%%%%%%

%\vspace{0.5in}
%\noindent{\Large \bf APPENDIX}

\section{Preliminaries}\label{prelim}
In this section, we review basic concepts about linear codes and their capability of correcting random corruptions, as well as Reed-Muller codes. 
\subsection{Basic coding definitions}
Recall that for a binary linear code $C \subseteq \F_2^n$ of blocklength $n$, if $k$ denotes the dimension of a code, i.e., $k=\log_2|C|$, a (non-redundant) generator matrix $G$ has dimension $k \times n$, a (non-redundant) parity-check matrix $H$ has dimension $(n-k)\times n$, and $C=\im(G)=\ker(H)$.

In the worst-case model, the distance of the code determines exactly how many erasures and errors can be corrected, with the following equivalent statements for the generator and parity-check matrices:
\begin{itemize}
\item $C$ has distance $d$,
\item $C$ allows to correct $d-1$ erasures,
\item $C$ allows to correct $\lfloor (d-1)/2 \rfloor$ errors,
\item any $d-1$ columns of $H$ are linearly independent,
%\item the sum of any $t$ columns of $H$, with $0< t \leq d$, is non-zero,
\item any $n-d+1$ rows of $G$ have full span.
\end{itemize}

%\begin{defin}
%A linear code with parity-check matrix $H$ allows to correct $s$ errors if $H$ is injective on $B(n,s)$, i.e., if 
%\begin{align}
%Hx \neq Hx', \quad \forall x,x' \in  B(n,s), x \neq x',
%\end{align}
%or equivalently, if 
%\begin{align}
%(\mathrm{Ker}(H) \setminus \{0\}) \cap  B(n,2s) = \emptyset.
%\end{align}
%Note that if the error vectors belong to $S(n,s)$ instead of $B(n,s)$, then the kernel of $H$ must not intersect $B_e(n,2s)=\{x\in B(n,2s) : w(x) \text{ is even}\}$.
%\end{defin}
\noindent
Two fundamental problems in worst-case coding theory is to determine the largest dimension of a code that has distance at least $d$, for a fixed $d$, and to construct explicit codes achieving the optimal dimension. None of these questions are solved in general, nor in the asymptotic regime of $n$ tending to infinity with $d=\alpha n$, and $\alpha \in (0,1/2)$. A random linear code achieves a dimension of $n(1-h(\alpha))+o(n)$, the Gilbert-Varshamov bound, but this bound has not been proved to be tight nor has it been improved asymptotically since 1957. Further, no explicit construction is known to achieve this bound. 

In this paper, we are interested in random erasures and errors, and in correcting them ``with high probability.'' This changes the requirements on the generator and parity-check matrices. For erasures, it simple requires linear independence of the columns ``with high probability'' (see Section \ref{prob-erasure})),while the requirement is more subtle for errors. Yet, in the probabilistic setting, random codes can be proved to achieve the optimal tradeoffs (e.g., code rate vs.\ erasure rate, or code rate vs.\ error rate) that are known for both erasures and errors (as special cases of Shannon's theorem). Explicit constructions of codes that achieve the optimal tradeoffs are also known, e.g., polar codes \cite{arikan}, and this paper investigates RM codes as such candidates. We now provide the formal models and results.

We mainly work in this paper with ``uniform'' models for erasures and errors, but we sometimes interpret the results for ``i.i.d.'' models (namely, the BEC and BSC channels). To formally connect these, we define first a unified probabilistic model. 
We start with erasures. We restrict ourselves to linear codes, although the definitions extend naturally to non-linear codes.    
\begin{defin}
A sequence of linear codes $\{C_n\}_{n \geq 1}$ of blocklength $n$ allows to correct random erasures from a sequence of erasure distributions $\{\mu_n\}_{n \geq 1}$ on $\F_2^n$, if for $Z \sim \mu_n$ and $S_Z=\{i \in [n] : Z_i=1\}$ (the erasure pattern),
\begin{align*}
\Pr\{ \exists x,y \in C_n : x \neq y, x[S_Z^c]=y[S_Z^c] \} \to 0, \quad \text{as } n \to \infty, 
\end{align*}
i.e., the probability of drawing erasures at locations that can confuse different codewords is vanishing. 
\end{defin}
Notice that $x[S_Z^c]=y[S_Z^c]$ if and only if we cannot correct erasures on coordinates $S_Z$ for neither $x$ nor $y$.
We now present a unified model for errors. 
\begin{defin}
A sequence of linear codes of length $n$ and parity-check matrix $\{H_n\}_{n \geq 1}$ allows to correct random errors from a sequence of error distributions $\{\mu_n\}_{n \geq 1}$ on $\F_2^n$ if for $Z\sim \mu_n$, 
\begin{align}
\Pr \{ \exists z' \in \F_2^n:  z' \neq Z, H_n z'=H_n Z,  \mu_n(z') \geq \mu_n(Z) \} \to 0, \quad \text{as } n \to \infty, \label{ML-general} 
\end{align}
i.e., the probability of drawing an error pattern $Z$ for which there exists another error pattern $z'$ that has the same syndrome as $Z$ and is more likely than $Z$ is vanishing. 
\end{defin}
Note that \eqref{ML-general} is the requirement that the error probability of the maximum likelihood (ML) decoder vanishes. Since ML minimizes the error probability for equiprobable codewords, \eqref{ML-general} is necessary for any decoder to succeed with high probability.

\begin{remark}
We next drop the term ``sequence of'' and the subscripts $n$, and simply say that a code $C$ of blocklength $n$ allows to correct random erasures/errors in specified models. The parameters introduced may also depend on $n$ without being explicitly mentioned.  
\end{remark}
We now introduce the uniform and i.i.d.\ models. 
\begin{defin}\text{}\\
(i) A linear code of blocklength $n$ allows to correct $s=s_n$ random erasures (resp.\ errors) if it can correct them from the uniform erasure (resp.\ error) distribution $U_{s}$, i.e., the uniform probability distribution on $ \partial B(n,s)=\{z \in \F_2^n : w(z)=\lceil s \rceil \}$.\\
(ii) A linear code of blocklength $n$ allows to correct erasures (resp.\ errors) for the $BEC(p)$ channel (resp.\ $BSC(p)$ channel), where $p=p_n$, if it can correct the distribution $B_{p}$, where $B_{p}$ is the i.i.d.\ distribution\footnote{This means the product distribution with identical marginals.} on $\F_2^n$ with Bernoulli$(p)$ marginal. 
\end{defin}

Note that for $\mu_n=U_{s}$, i.e., the uniform distribution over $ \partial B(n,s)$, the above definition reduces\footnote{We define ${n\choose s}$ as ${n \choose \lceil s \rceil }$ for a non-integer $s$.}  to 
\begin{align*}
\frac{ |\{z \in \partial B(n,s):  \exists z' \text{ s.t. } z \neq z',H z=H z' \}|}{{n\choose s}} \to 0, \quad \text{as } n \to \infty, 
\end{align*}
i.e., the fraction of bad error patterns, which have non-unique syndrome, is vanishing.  

The following Lemma\footnote{The statements are relevant for $s_n$ or $np_n$ that are $\omega(1)$.} 
 follows from standard probabilistic arguments.  
\begin{lemma}\label{equiv-models}\text{}\\
(i) If a linear code can correct $s=s_n$ random erasures (resp.\ errors), then it can correct erasures (resp.\ errors) from the $BEC\left((s - \omega(\sqrt{s}))/n\right)$ channel (resp.\ $BSC\left((s - \omega(\sqrt{s}))/n\right)$ channel).  \\
(ii) If a linear code can correct erasures (resp.\ errors) from the $BEC(p)$ channel (resp.\ $BSC(p)$ channel), then it can correct $np - \omega(\sqrt{np})$ random erasures (resp.\ errors).  
%(ii) If a linear code $C_n$ can correct erasures (resp.\ errors) from the $BEC(p_n)$ channel (resp.\ $BSC(p_n)$ channel), then it can correct $np_n$ random erasures (resp.\ errors).  
\end{lemma}
We now define the notions of capacity-achieving. Since in the rest of the paper typically considers codes at a given rate, and investigate how many corruptions they can correct, the definitions are stated accordingly. Note that what follows is simply a restatement of Shannon's theorems for erasures and errors, namely that a code $C$ of rate $R=(\log_2 (|C|))/n$ correcting a corruption probability $p$ must satisfy $R<1-p$ for erasures and $R<1-H(p)$ for errors. However, since we consider code rates that tend to 0 and 1, the requirements are broken down in various cases to prevent meaningless statements. 
\begin{defin}
A code is capacity-achieving (or achieves capacity) if it is $\e$-close to capacity for all $\e>0$.  
We now define the notion of $\e$-close to capacity in the four configurations: 
\begin{itemize}
\item A linear code $C$ of rate $R=o(1)$ is $\e$-close to capacity-achieving for erasures or for the BEC if it can correct $np$ random erasures for a $p$ satisfying 
\begin{align*}
p \geq 1-R(1+\e).
\end{align*}

\item A linear code $C$ of rate $R=o(1)$ is $\e$-close to achieving capacity for errors or for the BSC if it can correct $np$ random errors for a $p$ that satisfies 
\begin{align*}
h(p) \geq 1-R(1+\e),
\end{align*} 
where $h(p)=- p \log_2(p)- (1-p) \log_2(1-p)$ is the entropy function. 

\item A linear code $C$ of rate $R=1-o(1)$ is $\e$-close to achieving capacity for erasures or for the BEC if it can correct $np$ random erasures for a $p$ that satisfies 
\begin{align*}
p \geq (1-R)(1-\e).
\end{align*}

\item A linear code $C$ of rate $R=1-o(1)$ is $\e$-close to achieving capacity for erasures or for the BSC if it can correct $np$ random erasures for a $p$ that satisfies 
\begin{align*}
h(p) \geq (1-R)(1-\e).
\end{align*}
\end{itemize}
\end{defin}
Note that previous definition leads to the same notion of capacity for the uniform and i.i.d.\ models in view of Lemma \ref{equiv-models}.

\subsection{Equivalent requirements for probabilistic erasures}\label{prob-erasure}

%\Enote{I changed all the entropies in the paper from $H$ to $h$ and kept $H$ for parity-check matrices. Also, we don't have a notation for the set of subsets of $[n]$ of size $s$. I had defined ${[n] \choose s}$, but Amir didn't like it, and I also find it confusing. A natural candidate would be ${[n] \choose s}$, but then it sticks out from the equations. I suggest the following for now.}

In this section, we show the following basic results: a code can correct $s$ random erasures with high probability (whp), if a random subset of $s$ columns in its parity-check matrix are linearly independent whp, or if a random subset of $n-s$ rows in its generator matrix have full-span whp.

First note the following algebraic equivalence, which simply states that a bad erasure pattern, one that can confuse different codewords, corresponds to a subset of rows of the generator matrix that is not full-span (i.e., not invertible). 
\begin{lemma}\label{equiv2}
For an $n \times k$ matrix\footnote{We assume that $G$ has full column-rank, i.e., the generator matrix is non-redundant.} $G$ and for $S \subseteq [n]$, let $G_{S,\cdot}$ denote the subset of rows of $G$ indexed by $S$. Then the set of bad erasure patterns is given by 
\begin{align*}
\left\{ S \in {[n] \choose s} :  \exists x,y \in \ker(H), x\neq y, x[S^c]=y[S^c] \right\} \equiv \{ D \in \partial B(n,s): \rank (G_{D^c,\cdot})<k \},
\end{align*}
where $\rank (G_{D^c,\cdot})<k$ means that the columns of $G_{D^c,\cdot}$ are linearly dependent (i.e., the rows have full span).  
\end{lemma}
%Note that we will later apply this lemma to $G=E(m,r)^t$. 
\begin{proof}
We have
\begin{align*}
&\left\{ S \in {[n] \choose s} :  \exists x,y \in \im(G), x\neq y, x[S^c]=y[S^c] \right\}\\
%&\equiv \{E \in{[n] \choose s} : \exists z \in \ker(H), \text{ s.t. } \supp(z) \subseteq E, z \neq 0\}\\
&\equiv \left\{S\in{[n] \choose s}: \exists v \in \im(G) \text{ s.t. } v[S^c]=0, v \neq 0 \right\}\\
%& \equiv\{E \in{[n] \choose s} : G[D^c,\cdot] \text{ is not full row-rank} \}\\
&\equiv \left\{S \in {[n] \choose s}: \rank (G_{S^c,\cdot})<k   \right\}.
\end{align*}
\end{proof}
\begin{corol}\label{erasure-G}
For an $n \times k$ matrix $G$ and $s \in [n]$, denote by $G_{s,\cdot}$ the random sub-matrix of $G$ obtained by selecting $s$ rows uniformly at random. Then, the code with generator matrix $G$ can correct $s$ erasures if and only if 
\begin{align*}
\Pr\{  \rank (G_{n-s,\cdot})=k  \} \to 1, \quad \text{as } n \to \infty. \quad 
%\Pr_D\{ \sum_{j \in D} H[j] =0 \} \to 0, \quad \text{as } n \to \infty, \quad 
\end{align*}
\end{corol}
We now switch to the parity-check matrix interpretation.  
The following lemma shows that a bad erasure pattern for a code $C=\ker(H)$ can be identified as a subset of linear dependent columns in the parity-check matrix. 
\begin{lemma}\label{equiv}
For a matrix $H$ with $n$ columns, for $S \subseteq [n]$, and $H[S]$ the subset of columns of $H$ indexed by $S$, then the set of bad erasure patterns is given by 
\begin{align*}
\left\{ S \in  {[n] \choose s} :  \exists x,y \in \ker(H), x\neq y, x[S^c]=y[S^c] \right\} \equiv \left\{ S \in {[n] \choose s} : \mathrm{rk} (H[S] )<s \right\},
\end{align*}
where $\mathrm{rk} (H[S] )<s$ simply means the columns of $H[S]$ are linearly dependent. 
\end{lemma}
\begin{proof}[Proof of Lemma \ref{equiv}.]
Let 
\begin{align*}
\mathrm{BadSet}&:=\left\{ D \in  {[n] \choose s} : \mathrm{rk} (H[D] )<s \right\},\\
\mathrm{BadEra}&:=\left\{ S \in {[n] \choose s} :  \exists x,y \in \ker(H), x\neq y, x[S^c]=y[S^c] \right\},
\end{align*}
denote respectively the set of bad sets for which the columns of $H$ do not have full rank and the set of bad erasure patterns that can confuse codewords in $\ker H$. 
Since the code is linear, 
\begin{align*}
\mathrm{BadEra}&=\left\{ S \in {[n] \choose s} :  \exists v \in \ker(H), v \neq 0, v[S^c]=0 \right\}.
%&=\{ D \in [n,1:s]:  \sum_{j \in D} H[j] =0  \}
\end{align*}
Hence for any $S \in \mathrm{BadEra}$, there exists $v \in \ker(H)$, such that $v \neq 0$ and $\supp(v) \subseteq S$, and the columns of $H$ indexed by $S$ are not full rank. Conversely, if $D \in \mathrm{BadSet}$, $D$ contains a subset $V$ such that $V \neq \emptyset$ and $\sum_{i \in V} H[j]=0$, hence $D \in \mathrm{BadEra}$ (taking $v$ as the indicator vector of $V$). 
%  and $\mathrm{BadEra}$ must have vanishing probability. 
%Hence, to any set $E \in \mathrm{BadErr}$, we can associate a vector $z \in \ker(H) \setminus \{0\}$ such that $Hz=0$, 
%hence we can associate a set $D \in [n,1:s]$ given by the support of $z$, such that $\sum_{j \in D} H[j] =0$. This shows $\mathrm{BadErr} \subseteq \mathrm{BadSum}$. Conversely, for any set $D \in [n,1:s]$, such that $\sum_{j \in D} H[j] =0$, if $z$ is the indicator vector of $D$, then $w(z)\leq s$ and $z \neq 0$, and 
%Hence $\mathrm{BadErr} = \mathrm{BadSum}$. 
\end{proof}
\begin{corol}\label{erasure-H}
For a matrix $H$ with $n$ columns and $s \in [n]$, denote by $H[s]$ the random sub-matrix of $H$ obtained by selecting $s$ columns uniformly at random. Then, the code $\ker(H)$ can correct $s$ random erasures if and only if 
\begin{align*}
\Pr\{ \mathrm{rk} (H[s])=s \} \to 1, \quad \text{as } n \to \infty. \quad 
%\Pr_D\{ \sum_{j \in D} H[j] =0 \} \to 0, \quad \text{as } n \to \infty, \quad 
\end{align*}
\end{corol}
In other words, correcting $s$ random erasures is equivalent to asking that a random subset of $s$ columns in the parity-check matrix $H$ is full rank whp.

%\subsection{Requirements for probabilistic errors}

%\Amnote{Think what to say in light of what we shall write about decoding}
%\Enote{Kind of a small section, could be merge with previous now.}

While the requirement to correct probabilistic erasures (Corollaries \ref{erasure-H} and \ref{erasure-G}) is similar to the requirement for the worst-case model but ``with high probability,'' the situation is more subtle in the case of errors. Note that for a code $C$ with parity-check matrix $H$, the set of bad error patterns are the ones which lead to a non-unique syndrome, i.e. 
\begin{align*}
%&\{z \in \partial B(n,s):  \exists x,x' \in C_n, e' \in \partial B(n,s) \text{ s.t. } x \neq x', x + z= x' +z'\}\\
&\{z \in \partial B(n,s):  \exists z' \in \partial B(n,s) \text{ s.t. } z \neq z',Hz=Hz' \}  \\
&\equiv \{z \in \partial B(n,s):  \exists z' \in \partial B(n,s) \text{ s.t. } z \neq z', z+z' \in C \} .
\end{align*}
In other words, the set of bad error patterns are obtained by taking the set of codewords and splitting the codewords into elements of weight $s$. It is of course enough to consider the codewords of weight at most $2s$. However, even if the probability of drawing a codeword of weight at most $2s$ is vanishing, it does not follow that the probability of having a bad vector of weight $s$ is also vanishing. There are multiple ways to split a codeword in vectors of weight $s$, %namely ${w \choose w/2} {n-w \choose s-w/2}$ where $w$ is the weight of the codeword, 
and these lead to overlapping sets of vectors. Hence, the probability of a bad error pattern depends on the structure of $H$ beyond the probability of having dependent columns.  

%%%%%%%%%%%%%%%%%%%%%%%%%%%%%%%%%%
%%%%%%%%%%%%%%%%%%%%%%%%%%%%%%%%%%
%%%%%%%%%%%%%%%%%%%%%%%%%%%%%%%%%%
%%%%%%%%%%%%%%%%%%%%%%%%%%%%%%%%%%

\subsection{Basic properties of Reed-Muller codes}\label{RM-background}
The goal of this section is to revise the duality property of RM codes, which we use in this paper. One of the simplest way to understand this property is via the recursive structure of RM codes, mentioned below. 
We start by repeating the formal definition of RM codes via polynomials. 
\begin{defin}
Let $m,r$ be two positive integers with $r \leq m$, and let $n=2^m$. The Reed-Muller code of parameters $m$ and $r$ is defined by the set of codewords
\begin{align*}
RM(m,r)=\{(f(a_0),\dots,f(a_{n-1})) :  f \in  \mathbb{P}(m,r) \},
\end{align*}
where $\mathbb{P}(m,r)$ is the set of $m$-variate polynomials of degree at most $r$ on $\F_2$, and $a_0, \dots, a_{n-1}$ are all the elements of $\F_2^m$. 
\end{defin}
In particular, the matrix $E(m,r)$ that contains only the evaluations of the monomials of degree at most $r$ clearly defines a generator matrix for $RM(m,r)$. Formally, one should take the transpose $E(m,r)^t$ to obtain a generator matrix of $RM(m,r)$ that has dimension $n \times k$ (and not $k \times n$), where $k$ is the dimension of the code (as usually assumed in coding theory and as in previous section). The duality property says that $E(\cdot,\cdot)$ can also be used to obtain a parity-check matrix of $RM$ codes, as follows. 
\begin{lemma}\label{duality}[Duality of RM codes]
$E(m,m-r-1)$ is a parity-check matrix for $RM(m,r)$, or equivalently, $E(m,r)$ is a parity-check matrix for $RM(m,m-r-1)$. 
\end{lemma}
%This is a linear code, since $\mathbb{P}(m,r)$ is closed under addition. To generate the code, it is enough to consider the evaluation of monomials. In particular, the dimension of the code is $k(m,r)=\sum_{i=0}^r {m \choose i}$. The distance of the code is $d(m,r)=2^{m-r}$, and is easily found by using the following recursive structure. 

To show this result, note that RM codes can be defined recursively as follows. 
Instead of displaying the rows of the generator matrix by increasing order of the monomial degrees, consider the lexicographic order. For example, $RM(3,3)$ is generated by :

\begin{align*}
	\begin{blockarray}{ccccccccc}
	%		&   & \BAmulticolumn{3}{c}{\overbrace{\rule{0.16\columnwidth}{0pt}}^\text{monomials of deg-1}}& & \BAmulticolumn{3}{c}{\overbrace{\rule{0.3\columnwidth}{0pt}}^\text{monomials of deg-r}}\\		
		  & \begin{pmatrix} 0\\ 0 \\0 \end{pmatrix} & \begin{pmatrix} 1\\ 0 \\0 \end{pmatrix} & \begin{pmatrix} 0\\ 1 \\0 \end{pmatrix} & \begin{pmatrix} 1\\ 1 \\0 \end{pmatrix} & \begin{pmatrix} 0\\ 0 \\1 \end{pmatrix} & \begin{pmatrix} 1\\ 0 \\1 \end{pmatrix} & \begin{pmatrix} 0\\ 1 \\1 \end{pmatrix} & \begin{pmatrix} 1\\ 1 \\1 \end{pmatrix} \\
			\begin{block}{c(cccccccc)}
			1 & 1 & 1 & 1 & 1 & 1 & 1 & 1 & 1 \\
			x_1 & 0 & 1 & 0 & 1 & 0 & 1 & 0 & 1 \\
			x_2 & 0 & 0 & 1 & 1 & 0 & 0 & 1 & 1 &\\  
			x_1 x_2 & 0 & 0 & 0& 1 & 0 & 0 & 0 & 1 \\
			x_3 & 0 & 0 & 0 & 0 & 1 & 1 & 1 & 1 &\\
			x_1x_3 & 0 & 0 & 0 & 0 & 0 & 1 & 0 & 1 \\		
			x_2x_3 & 0 & 0 & 0 & 0 & 0 & 0 & 1 & 1 \\			
			x_1x_2x_3 & 0 & 0 & 0 & 0& 0 & 0 & 0 & 1 \\		
			\end{block}	
	\end{blockarray}
\end{align*}
Notice the order of $x_3$ and $x_1x_2$ in the above. With that order, the matrix is the tensor product of $\begin{pmatrix} 1 & 1 \\ 0 & 1 \end{pmatrix}$ with itself 3 times. In fact, RM codes can equivalently be defined in terms of tensor products. 
\begin{defin}
%Let $m,r$ be two positive integers with $0 \leq r \leq m$, and let $n=2^m$. 
For an integer $m \geq 0$, define
\begin{align*}
G(m)=\begin{pmatrix} 1 & 1 \\ 0 & 1 \end{pmatrix}^{\otimes m},
\end{align*} 
with $G(0)=1$. For $0 \leq r \leq m$, define $G(m,r)$ as the sub-matrix of $G(m)$ obtained by keeping the rows with weight more or equal to $2^{m-r}$. 
%Define also $H(m,r)=G(m,m-r - 1)$, i.e., the sub-matrix of $G(m)$ obtained by keeping the rows with weight more or equal to $2^{r+1}$. 
\end{defin}
Note that $G(m,r)$ is simply a permutation of the rows of $E(m,r)$, hence it is also a generator matrix for $RM(m,r)$. Moreover, it can be constructed recursively as follows: 
\begin{align*}
G(m,r) &= \begin{pmatrix} G(m-1,r) & G(m-1,r)  \\ 0 & G(m-1,r-1)  \end{pmatrix}.
\end{align*}
%with $G[0,0]=1$, $G[\cdot,-1]=\emptyset$ and $G[m,m+1]=G[m]$.
The polynomial interpretation of this recursion is simply the fact that a $m$-variate polynomial $f$ of degree at most $r$ can be expressed as
\begin{align*}
f(x_1,\dots,x_m)= f_1(x_1,\dots,x_{m-1}) + x_m f_2(x_1,\dots,x_{m-1}),
\end{align*}
where $f_1$ and $f_2$ are $m$-variate polynomials of degrees at most $r$ and $r-1$ respectively. 
With this recession, the duality property (Lemma \ref{duality}) (as well as the fact that the distance of $RM(m,r)$ is $2^{m-r}$) are direclty proved by induction. We refer to \cite{sloane-book} for complete proofs.

%RM codes then have an interesting duality property, in the sense that $H(m,r)$ is a parity-check matrix of $RM(m,r)$.
%\begin{lemma}\label{duality2}%[Duality of RM codes]
%\begin{align}
%G(m,r)^\perp &= \langle G(m,m-r-1) \rangle,\\
%H(m,r)^\perp &= \langle H(m,m-r-1) \rangle.
%\end{align}
%%or equivalently,
%%\begin{align}
%%H(m,r)^\perp = \langle H(m,m-r-1) \rangle.
%%\end{align}
%\end{lemma}
%Note that $H(m,r)$ can also be constructed recursively as follows:
%\begin{align}
%H(m,r) &= \begin{pmatrix} H(m-1,r-1) & H(m-1,r-1)  \\ 0 & H(m-1,r)  \end{pmatrix}.
%\end{align}
%with $H[m,m]=H[m,m+1]=\emptyset$ and $H[m,m-1]=G[m]$.
%Finally, the distance of $RM(m,r)$ is equal the minimal weight of the rows of $G(m,r)$.
%\begin{lemma}[Distance of RM codes]
%The distance of $RM(m,r)$ is given by  
%\begin{align}
%d(m,r)= 2^{m-r}. 
%\end{align}
%\end{lemma}
%Both the duality and distance property of RM codes are easily proved using these recursion formulas. 
%\begin{corol}
%$RM(m,r)$ can correct $2^{m-r}-1$ erasures and $2^{m-r-1}-1$ errors (with probability one). 
%\end{corol}
%Note that this is a relative poor performance in the worst-case model of erasures/errors, compared to BCH codes for example. 
%We will next see that the performance is much improved for probabilistic models for erasures/errors.
%%%%%%%%%%%%%%%%%%%%%%%%%%%%%%%%%%%%%
%%%%%%%%%%%%%%%%%%%%%%%%%%%%%%%%%%%%%
%%%%%%%%%%%%%%%%%%%%%%%%%%%%%%%%%%%%%
%%%%%%%%%%%%%%%%%%%%%%%%%%%%%%%%%%%%%

\section{Weight distribution of Reed-Muller codes}\label{sec:wt-dist}

%\Enote{Make sure we cite latest bound everywhere}
%\Enote{Not sure that I understand the above comment (didn't write it), does ``cite'' refer to the bibliography or just the use of the simplified bound?}

In this section we study the weight distribution of Reed-Muller codes. Our analysis is based on the technique of Kaufman, Lovett and Porat \cite{KLP}.  We start with some high level intuition.
Naturally, one expects that most codewords of $RM(m,r)$ (or any linear code, for that matter) to have weight around $n/2 = (2^m)/2$. A trivial upper bound on the number of codewords having such weight (or larger) is the total number of codewords, i.e., $2^{m\choose \leq r}$. The question is thus how does this number changes when we consider smaller weights. Specifically, what is the number of codewords that have weight at most $2^{m-\ell}$ for some parameter $\ell$. If we denote this number with $2^{c(m,r,\ell)\cdot {m\choose \leq r}}$, then we are asking for the value of the term $c(m,r,\ell)$. A trivial lower bound on the number of such codewords is $2^{m\ell + {m-\ell \choose r-\ell}}$, which is obtained by counting all polynomials of degree-$r$ that are divisible by $\ell$ linear functions. If this was tight, then $c(m,r,\ell) \approx 1$, which suggests that the number of such polynomials grows roughly like the number of degree $r-\ell$ polynomials on $m-\ell$ variables. Kaufman et al. proved that indeed this number is essentially the right answer for constant $r$. More precisely, they proved that  $c(m,r,\ell) = O(r^2)$. Our contribution is replacing this estimate with roughly $c(m,r,\ell) = O(\ell^4)$. This change is most significant when $\ell$ is very small compared to $r$, e.g., when considering the number of words of weight e.g. roughly $n/4$ (so $\ell$ is 2) and when $r$ is large, e.g. $r=\Omega(m)$. This improvement turns out to be critical for two of our results on achieving capacity -- for erasure in low rates and errors in high rate. It remains open if one can prove that $c(m,r,\ell) = O( 1)$, namely is a constant independent of all parameters.

We start by giving the high level view of the proof of \cite{KLP} and then explain how to improve their analysis. We first introduce some notation.

For a function $f:\F_2^m\to \F_2$ (equivalently, a word $f\in \F_2^n$) we denote by $\wt(f)$ the relative (Hamming) weight of $f$, i.e., 
$$\wt(f) = \frac{1}{2^m}\left|\{ v\in\F_2^m \mid f(v)\neq 0\}\right|.$$

The cumulative weight distribution of $RM(m, r)$ at a relative weight $0 \leq \alpha \leq 1$, denoted $W_{m,r}(\alpha)$, is the number of codewords of $RM(m,r)$ whose relative weight is at most $\alpha$,
$$W_{m,r}(\alpha) \triangleq \left|\{f \in RM(m,r) \mid \wt(f) \leq \alpha \}\right|.$$

The main theorem of \cite{KLP} roughly states that the number of code words of $RM(m,r)$ of  relative weight at most $2^{-\ell}$ is roughly\footnote{We use $\exp(x)$ instead of $e^x$.} $\exp\left(r^2{m \choose \leq r}\cdot (\frac{r}{m-r})^\ell\right)$. I.e., the number of codewords of relative weight smaller than $1/2$ is significantly smaller than the number of words of relative weight $1/2$. We next give a slightly informal statement of the main theorem of \cite{KLP}.\footnote{Kaufmann et al. also give a lower bound on $W_{m,r}$ for small values of $r$, but we do not need it here.}

\begin{thm}[Theorem 3.1 of \cite{KLP}]\label{thm:KLP}
%Let $2^{-r} \leq \alpha < \frac{1}{2}$ be a parameter. 
Let $1\leq \ell \leq r-1$ and $0 < \e \leq 1/2$ be such that $2^{-r} \leq (1-\e)2^{-\ell} < \frac{1}{2}$. Then
$$W_{m,r}((1-\e)2^{-\ell}) \leq  (1/\e)^{O({r^2}{m \choose {\leq r-\ell}})}.$$
\end{thm}

We next explain the main lemma used to prove Theorem~\ref{thm:KLP}.
First we introduce the notion of a discrete partial derivative. 

The discrete derivative of $f:\F_2^m \to \F_2$ in direction $y\in \F_2^m$ at point $x$ is is $$\Delta_yf(x) \triangleq f(x+y)+f(x).$$
It is clear that $\Delta_yf(x) = \Delta_yf(x+y)$, so in particular, $\Delta_y\Delta_yf(x)=0$. Thus, the function $\Delta_y f(\cdot )$ %derivative in direction $y$ does not depend on $y$ and 
is determined by its value on the quotient space $\F_2^m/\langle y \rangle$, where for a set of vectors $V$,  $\langle V \rangle$ denotes the space spanned by the vectors in $V$. It is a straight forward observation that if $f$ is a polynomial of degree at most $r$ then $\Delta_yf(\cdot)$ is a polynomial of degree at most $r-1$.
Similarly, the $\ell$'th iterated derivative of $f$ in direction $Y=(y_1,\ldots,y_\ell)$ and point $x$ is
$$\Delta_Yf(x)  \triangleq \Delta_{y_1}\Delta_{y_2}\ldots\Delta_{y_\ell}f(x) = \sum_{I\subseteq [\ell]}f(x+\sum_{i\in I}y_i).$$
It is easy to show that $\Delta_Yf(x)$ does not depend on the order in which we take the derivatives. 

We are now ready to state the main lemma of \cite{KLP}.

\begin{lemma}[Lemma 2.1 of \cite{KLP}]\label{lem:KLP}
Let $f:\F_2^m$ be such that $\wt(f)\leq (1-\e)2^{-\ell}$, for $0<\e<1$. Let $\delta>0$ be an approximation parameter. There exists a universal algorithm $\A$ (which does not depend on $f$) with the following properties:
\begin{enumerate}
\item $\A$ has two inputs: $x\in \F_2^m$ and $(Y_1,\ldots, Y_t )\in \left(\F_2^m\right)^\ell$.
\item $\A$ has oracle access to the $\ell$'th derivatives $\Delta_{Y_1}f(\cdot),\ldots,\Delta_{Y_t}f(\cdot)$.
\end{enumerate}
Then, for $t = c\left(\log(1/\delta)\log(1/\e)+ \log^2(1/\delta)\right)$, for some absolute constant $c$ that does not depend on any of the parameters, there exists a setting for $Y_1,\ldots,Y_t$ such that 
$$\Pr_{x\in\F_2^m}\left[\A\left(x,(Y_1,\ldots,Y_t);\Delta_{Y_1}f(\cdot),\ldots,\Delta_{Y_t}f(\cdot)\right) = f(x)\right] \geq 1-\delta,$$
where $\Pr_{x\in\F_2^m}$ means that $x$ is uniformly drawn in $\F_2^m$.  
\end{lemma}

%\Enote{Above: Is $\log(1/\delta)^2 = 2 \log(1/\delta)$? $A$ has two inputs or $t+1$? The convention for evaluating $A(;,)$?  I added the def of $\Pr_{x \in \F_2^m}$}

In other words, what the lemma shows is that if $f$ has relatively low weight, then given an appropriate set of $ O\left(\log(1/\delta)\log(1/\e)+ \log^2(1/\delta)\right)$ many $\ell$-th derivatives of $f$, one can determine the value of $f$ on most inputs. When $f$ is a degree-$r$ polynomial, its derivatives are degree $r-\ell$ polynomials and thus the lemma lets us approximate $f$ well by lower degree polynomials. 

We now show how Kaufman et al.  deduced Theorem~\ref{thm:KLP} from Lemma~\ref{lem:KLP}. 
The first idea is to set $\delta = 2^{-r-1}$. The point is that there is at most one degree-$r$ polynomial $f$ at distance $\delta$ from the function $\A\left(x;Y_1,\ldots,Y_t,\Delta_{Y_1}f(\cdot),\ldots,\Delta_{Y_t}f(\cdot)\right)$. Indeed, by the triangle inequality, the distance between any two polynomials that are $\delta$-close to $\A\left(x;Y_1,\ldots,Y_t,\Delta_{Y_1}f(\cdot),\ldots,\Delta_{Y_t}f(\cdot)\right)$ is at most $2\delta<2^{-r}$, which is smaller than the minimum distance of $RM(m,r)$. Hence, to bound the number of polynomials $f\in RM(m,r)$ of relative weight at most $\wt(f)\leq (1-\e)2^{-\ell}$, it is enough to bound the possible number of functions of the form 
$\A\left(x;Y_1,\ldots,Y_t,\Delta_{Y_1}f(\cdot),\ldots,\Delta_{Y_t}f(\cdot)\right)$ for the appropriate $t$. 

The second step in the proof of Kaufmann et al. is to give an upper bound on the number of expressions of the form $\A\left(x;Y_1,\ldots,Y_t,\Delta_{Y_1}f(\cdot),\ldots,\Delta_{Y_t}f(\cdot)\right)$. Since $\A$ is fixed, they only have to bound the number of sets $Y_i$ and the number of polynomials of the form $\Delta_{Y_i}f$ and raise it to the power $t$. They now use the fact that $\Delta_{Y_i}f$ is a polynomial of degree at most $r-\ell$ so the number of such polynomials is $2^{m\choose \leq r-\ell}$.

Combining everything,and letting $\delta = 2^{-r-2}$ so that   
$$t = O\left(r\log(1/\e)+ r^2\right)$$
and 
\begin{equation}\label{eq:KLP-estimate}
W_{m,r}((1-\e)2^{-\ell}) \leq \left(2^{m \ell} \cdot 2^{{m\choose \leq r-\ell}} \right)^t =  \left( 2^{m\ell+ {m\choose \leq r-\ell}} \right)^{O\left(r\log(1/\e)+ r^2\right)}.
\end{equation}

%\Amnote{we should probably explain that the problem is the case $\ell=1$ and that we would like to save by replacing $r$ by $\ell$. I did not do so but rather just wrote down the argument}

One downside of the result of \cite{KLP} is that due to the dependence on $r$ of the constant in the big $O$, their estimate is tight only for constant $r$, and becomes trivial at $r = \tilde{O}(\sqrt{m})$. Indeed, the bound in the exponent goes down roughly like $(r/m)^\ell$. Hence, the maximum is obtained for small values of $\ell$, i.e., $\ell=1$ or $\ell=2$. For these values, the term $\log(1/\e)r+r^2$ in the exponent basically eliminates any saving that comes from $(r/m)^\ell$. Thus, to improve the bound on the weight distribution it is crucial to improve the bound for small values of $\ell$. Our result does exactly this, we are able to replace the power of $r$ in the exponent with a power of $\ell$, which gives the required saving for small values of $\ell$.

We now explain how we modify the arguments of \cite{KLP} in order to tighten the estimate given in Theorem~\ref{thm:KLP}, that hold for a broader range of parameters.

Our first observation is that one can relax the setting of $\delta$. We set $\delta = (1-\e)2^{-\ell -2}$, %for some constant $\cc$, 
instead of $\delta = 2^{-r-2}$. The effect is that now there can be many polynomials $g$ that are $\delta$-close to 
$\A\left(x;Y_1,\ldots,Y_t,\Delta_{Y_1}f(\cdot),\ldots,\Delta_{Y_t}f(\cdot)\right)$. Indeed, all we know is that the distance between any two such polynomials is at most $2\delta = (1-\e)2^{-\ell-1}$. The point is that the number of such polynomials can be bounded from above by $W_{m,r}(2^{-\ell-1})$ which is relatively small compared to $W_{m,r}(2^{-\ell})$ and so we can (almost) think of it as $1$. The effect on the expression 
\eqref{eq:KLP-estimate} is that in the expression for $t$, we can (almost) replace $r$ by $\ell$. As explained before, this gives a significant saving over the bound of \cite{KLP}.
%This is especially important for small values of $\ell$ (i.e., for functions of relatively high weight), which as explained above. 

Our second improvement comes from the simple observation that $\Delta_{Y_i}f$ can be defined by its value on the quotient space $\F_2^m/\langle Y_i \rangle$. As this is a space of dimension $m-\ell$, for a fixed $Y_i$, we can upper bound the number of polynomials of the form $\Delta_{Y_i}f$ by $2^{m-\ell \choose \leq r-\ell}$, instead of $2^{m \choose \leq r-\ell}$, which again yields a tighter estimate.

We now state our bound on the weight distribution of Reed-Muller codes.

\begin{thm}\label{thm:wt-dist}
%Let $2^{-r} \leq \alpha < \frac{1}{2}$ be a parameter. 
Let $1\leq \ell \leq r-1$ and $0 < \e \leq 1/2$. %Set $\alpha=(1-\e)2^{-\ell}$. 
Then, %for every integer $\cc>2$, 
if $r \leq m/4$, 
$$W_{m,r}((1-\e)2^{-\ell}) \leq  (1/\e)^{8c \ell^4 {m-\ell \choose \leq r-\ell} },$$
where $c$ is an absolute constant (same as in Lemma~\ref{lem:KLP}).
\end{thm}

\begin{proof}
We shall prove by induction a stronger statement, namely, $$W_{m,r}((1-\e)2^{-\ell}) \leq  (1/\e)^{2c\left(m(r+3)^3(r-\ell) +(\ell+3)^2{m-\ell \choose r-\ell} \right)}.$$
Set $\delta = (1-\e)2^{-\ell-2}$  and $t=c\cdot \left(\log(1/\delta)\log(1/\e)+ \log^2(1/\delta)\right) \leq c\cdot \log(1/\e) \cdot (\ell+3)^2$.
By Lemma~\ref{lem:KLP}, for any $f\in RM(m,r)$ of weight $\wt(f)\leq (1-\e)2^{-\ell}$, there is a choice of sets 
$Y_1,\ldots, Y_t \in \left(\F_2^m\right)^\ell$ such that 
$$\Pr_{x\in\F_2^m}\left[\A\left(x,(Y_1,\ldots,Y_t);\Delta_{Y_1}f(\cdot),\ldots,\Delta_{Y_t}f(\cdot)\right) = f(x)\right] \geq 1-\delta.$$
We next bound the number of functions $g$ of the form $$g=\A\left(x,(Y_1,\ldots,Y_t);\Delta_{Y_1}f(\cdot),\ldots,\Delta_{Y_t}f(\cdot)\right) .$$
We can upper bound the number of sets $Y\in\left(\F_2^m\right)^\ell$ with $2^{m\ell}$. For each such $Y$, since $\Delta_{Y}f$ is a polynomial of degree $\leq r-\ell$ that is defined by its values on the space $\F_2^m/\langle Y_i \rangle$, there are at most $2^{m-\ell \choose \leq r-\ell}$ polynomials of the form $\Delta_Yf$. Thus, the number of possible such functions $g$ is at most 
$$\left(2^{m\ell} 2^{m-\ell \choose \leq r-\ell}\right)^t = (1/\e)^{c(\ell+3)^2\left({m\ell} +{m-\ell \choose \leq r-\ell}\right)}.$$
Given any such $g$, the number of polynomials $g' \in RM(m,r)$ at distance at most $(1-\e)2^{-\ell-2}$ from $g$ is at most $W_{m,r}((1-\e)2^{-\ell-1})$. Indeed, fix some $f$ close to $g$. Then any other such polynomial $g'$ has distance at most $ 2(1-\e)2^{-\ell-2}$ from $f$, and so $\wt(f-g')\leq  (1-\e)2^{-\ell-1}$ and $f-g' \in RM(m,r)$.
Concluding we get 
$$W_{m,r}( (1-\e)2^{-\ell}) \leq(1/\e)^{c(\ell+3)^2\left({m\ell} +{m-\ell \choose \leq r-\ell}\right)} \cdot W_{m,r}((1-\e)2^{-\ell-1}).$$
Since $W_{m,r}((1-\e)2^{-r})=1$ (as only the $0$ polynomial has such small weight) and ${m-\ell \choose \leq r-\ell}\leq {m\choose\leq r}\cdot \left(\frac{r}{m}\right)^\ell$, we get by induction that
%$and algebraic manipulations the bound on $W_{m,r}$.\footnote{More accurately, induction will give a tighter bound on $W_{m,r}$ and the bound that we give is a simplification of that tighter, yet more cumbersome, bound.}
\begin{eqnarray*}
W_{m,r}( (1-\e)2^{-\ell})  &\leq &(1/\e)^{c(\ell+3)^2\left({m\ell} +{m-\ell \choose \leq r-\ell}\right)} \cdot W_{m,r}((1-\e)2^{-\ell-1})\\
&\leq &(1/\e)^{c(\ell+3)^2\left({m\ell} +{m-\ell \choose \leq r-\ell}\right)} \cdot (1/\e)^{2c\left(m(r+3)^3(r-\ell-1) + (\ell+4)^2{m-(\ell+1) \choose \leq r-(\ell+1)} \right)}\\
&\leq &(1/\e)^{c(\ell+3)^2\left({m\ell} +{m-\ell \choose \leq r-\ell}\right)} \cdot 
(1/\e)^{2c\left(m(r+3)^3(r-\ell-1) + (\ell+4)^2 {m-\ell \choose \leq r-\ell)}\cdot\left(\frac{r}{m}\right) \right)}\\
&\leq & (1/\e)^{2cm\left((r+3)^3 (r-\ell) \right)}
\cdot (1/\e)^{c{m-\ell \choose r-\ell}\left( (\ell+3)^2 + 2(\ell+4)^2\frac{r}{m}\right)}\\
&\leq^{*} & (1/\e)^{2c\left(m(r+3)^3(r-\ell) +(\ell+3)^2{m-\ell \choose r-\ell} \right)},
\end{eqnarray*}
where in inequality $(*)$ we use the fact that $r<m/4$.
The bound in the statement of the theorem follows by a simple manipulation.
This concludes the proof of the theorem.
\end{proof}

%As the bound given in Theorem~\ref{thm:wt-dist} is a bit hard to work with, let us give some rough estimates on its value. 
%
%\begin{corol}\label{cor:wt-dist-simple}
%%For a constant $0<\e<1/2$, and integer $\ell$, let $\alpha = (1-\e)2^{-\ell}$. 
%Under the conditions of Theorem~\ref{thm:wt-dist} we have that
%%\begin{enumerate}
%%\item For $\ell=1$, $W_{m,r}(\alpha)\leq 2^{O_{c,\e,\cc}\left( {m\choose \leq r}\frac{r}{m}\right)}$.
%$$W_{m,r}((1-\e)2^{-\ell})\leq (1/\e)^{8c \ell^4 {m-\ell \choose \leq r-\ell} }.$$
%%\end{enumerate}
%In particular, for  $r=o(m)$,
%$$W_{m,r}((1-\e)/2)\leq (1/\e)^{O(mr^4) +  o(1)\cdot {m\choose \leq r}}.$$
%\end{corol}

%
%As in the paper of \cite{KLP}, almost the exact same proof as our proof of Theorem~\ref{thm:wt-dist} yields a bound for list-decoding of Reed-Muller codes. 
%Following \cite{KLP} we denote:
%$$L_{m,r}(\alpha) = \max_{g:\{0,1\}^m\to \{0,1\}}\left| f \in RM(m,r) \mid \wt(f-g)\leq \alpha\}\right|.$$
%That is, $L_{m,r}(\alpha)$ denoted the maximal number of code words of $RM(m,r)$ in a hamming ball of radius $\alpha 2^m$.
%
%\begin{thm}\label{thm:list-dec}
%Let $1\leq \ell \leq r-1$ and $0 < \e \leq 1/2$. %Set $\alpha=(1-\e)2^{-\ell}$. 
%Then, 
%if $r \leq m/4$ then
%$$L_{m,r}((1-\e)2^{-\ell}) \leq  (1/\e)^{2c\left(m (r+3)^4 + (\ell+3)^2{m-\ell \choose \leq r-\ell} \right)},$$
%where $c$ is the absolute constant from Lemma~\ref{lem:KLP}.
%\end{thm}

%%%%%%%%%%%%%%%%%%%%%%%%%%%%%%%%%%
%%%%%%%%%%%%%%%%%%%%%%%%%%%%%%%%%%
%%%%%%%%%%%%%%%%%%%%%%%%%%%%%%%%%%
%%%%%%%%%%%%%%%%%%%%%%%%%%%%%%%%%%

%\Amnote{changed title to $E$}

\section{Random submatrices of $E(m,r)$} \label{sec:sampling}

%\Anote{Changed the questions - we don't care about the probability for given size of U, but rather at the best size of U for which the prob is high!}
As discussed in the introduction and Section \ref{prelim}, in order to understand the ability to decode from  erasures it is important to understand the following questions. Consider randomly chosen set  $U$ of a given parameter size $k$:
\begin{question}\label{question:columns}
What is the largest $s$ for which the submatrix $U^r$ has full column-rank with high probability?
\end{question}
\begin{question}\label{question:rows}
What is the smallest $s$ for which the submatrix $U^r$ has full row-rank with high probability?
\end{question}
%While the first question is important in view of Theorem~\ref{thm: lin ind eval vectors} and the connection to decoding errors in the BSC model, the second question determines the ability to correct erasures, i.e. to correct errors in the BEC model, see Lemma~\ref{???}. 

In this section we provide an answer to each of these questions.\footnote{Using tensoring to produce linearly independent vectors has also been studied recently in the context of real vectors \cite{smooth-tensor}.} Note that for any degree-$r$, the number of rows of $E(m,r)$, namely ${m \choose \leq r}$, is an upper bound on the value of $s$ for the first question and a lower bound for the second. For small $r$ we prove that we can approach this optimal bound asymptotically in both.
%\Amnote{Changed name of section}

Note that, interestingly, the duality property of RM codes allows to relate question \ref{question:columns} and \ref{question:rows} to each other but for different ranges of the parameters. Namely, the following holds. 
\begin{lemma}\label{lem:equiv2}
For a set $S \subseteq [n]$, denote by $E(m,d)[S]$ the sub-matrix of $E(m,d)$ obtained by selecting the columns indexed by $S$. For any $s \leq n$, 
\begin{align*}
\left\{S \in{[n] \choose s} : \mathrm{rk} (E(m,d)[S] )=s \right\} = \left\{ S\in {[n] \choose s} :  \mathrm{rk} (E(m,m-d-1)[S^c] )=n-{m\choose \leq d} \right\} .
\end{align*}
%\begin{align}
%\Pr\{ \mathrm{rk} (E(m,d)[s] )=s \} \to 1, \quad \text{as } n \to \infty,   \label{lim1}
%%\Pr_D\{ \sum_{j \in D} H[j] =0 \} \to 0, \quad \text{as } n \to \infty, \quad 
%\end{align}
%if and only if 
%\begin{align}
%\Pr\{ \mathrm{rk} (E(m,m-d-1)[n-s] )=n-{m\choose \leq d} \} \to 1, \quad \text{as } n \to \infty.   \label{lim2}
%%\Pr_D\{ \sum_{j \in D} H[j] =0 \} \to 0, \quad \text{as } n \to \infty, \quad 
%\end{align}
Note that $E(m,d)[S]=s$ means that $E(m,d)[S]$ has full column-rank and $E(m,m-d-1)[S^c]=n-{m\choose \leq d}$ means that $E(m,m-d-1)[S^c]$ has full row-rank.  
\end{lemma}
\begin{corol}\label{cor:equivalence}
For an integer $s \in [n]$, denote by $E(m,d)[s]$ the random matrix obtained by sampling $s$ columns uniformly at random in $E(m,d)$.
Then,
\begin{align*}
\Pr\{ \mathrm{rk} (E(m,d)[s] )=s \} = \Pr\left\{ \mathrm{rk} (E(m,m-d-1)[n-s] )=n-{m\choose \leq d} \right\},
\end{align*}
where both terms are the probability of drawing a uniform erasure pattern of size $s$ which can be corrected with the code $\ker{E(m,d)}$. 
%\begin{align*}
%\Pr\{ \mathrm{rk} (E(m,d)[s] )=s \} \to 1, \quad \text{as } n \to \infty,   \label{lim1}
%%\Pr_D\{ \sum_{j \in D} H[j] =0 \} \to 0, \quad \text{as } n \to \infty, \quad 
%\end{align*}
%if and only if 
%\begin{align*}
%\Pr\{ \mathrm{rk} (E(m,m-d-1)[n-s] )=n-{m\choose \leq d} \} \to 1, \quad \text{as } n \to \infty.   \label{lim2}
%%\Pr_D\{ \sum_{j \in D} H[j] =0 \} \to 0, \quad \text{as } n \to \infty, \quad 
%\end{align*}
%Note that $E(m,d)[s]=s$ means that $E(m,d)[s]$ is full column-rank and $E(m,m-d-1)[n-s]=n-{m\choose \leq d}$ means that $E(m,m-d-1)[n-s]$ is full row-rank. 
\end{corol}
This correspondence follows from Lemmas \ref{equiv}, \ref{equiv2} and \ref{duality}. We provide the proof below for convenience. 
\begin{proof}[Proof of Lemma \ref{lem:equiv2}]
Note that 
\begin{align*}
&\{S \in{[n] \choose s} : \mathrm{rk} (E(m,d)[S] )<s \}\\
&\equiv \{S \in{[n] \choose s} : \exists z \in \ker(E(m,d)), \text{ s.t. } \supp(z) \subseteq S, z \neq 0\}\\
&\equiv \{S \in{[n] \choose s} : \exists z \in \ker(E(m,d)) \text{ s.t. } z[S^c]=0, z \neq 0\},
\end{align*}
and using Lemma \ref{duality}, previous set is equal to 
\begin{align*}
&\{S \in{[n] \choose s} : \exists z \in \im(E(m,m-d-1)) \text{ s.t. } z[S^c]=0, z \neq 0\}\\
%&\equiv \{S \in{[n] \choose s} : \exists x,y \in \im(E(m,m-d-1)) \text{ s.t. } x[S^c]=y[S^c], x \neq y\}\\
& \equiv\{S \in{[n] \choose s} : E(m,m-d-1)[S^c] \text{ is not full row-rank} \}\\
&\equiv \{S \in{[n] \choose s} : \mathrm{rk} (E(m,m-d-1)[n-s] )<n-{m\choose \leq d}  \}.
\end{align*}
%From Lemma \ref{equiv}, $\{S \in{[n] \choose s} : \mathrm{rk} (E(m,d)[S] )=s \}$ is equal to the set of erasure patterns of size $s$ that can be corrected with the code $\ker{E(m,d)}$. 
%From Lemma \ref{duality}, this code is generated by the rows of the $(n-{m\choose \leq d}) \times n$ matrix $E(m,m-d-1)$. Hence the error patterns of size $s$ that can be corrected are given by $\{ S\in{[n] \choose s} :  \mathrm{rk} (E(m,m-d-1)[S^c] )=n-{m\choose \leq d} \}$, since erasing the columns indexed by $S$ in $E(m,m-d-1)$ must produce a set of columns which are full row-rank. 
\end{proof}
This equivalence property implies that it sufficient to answer each question in one of the two extremal regimes, which we next cover.  

%\subsection{Linear independence of columns of $E(m,r)$ for small degree-$r$}
\subsection{Random submatrices of $E(m,r)$, for small $r$, have full column-rank} \label{first-question}

%\begin{lemma}\label{lem: random r-tensors independent}
%Lemma about picking vectors at random and comparing their r-tensors.
%\end{lemma}

The following theorem addresses Question~\ref{question:columns} in the case of low degree-$r$.
\begin{thm}\label{thm:lin_ind_eval_vectors}
Let $\e>0$ and $k,m,r$ integers such that $s< {m- \log({m\choose \leq r})- \log(1/\e) \choose \leq r}$. Then, 
with probability larger than $1-\e$ if we pick $u_1,\ldots,u_s\in \F_2^m$ uniformly at random we get that the evaluation vectors, $u_1^{{r}},\ldots, u_s^{{r}}$ are linearly independent. 
\end{thm}

%\Anote{Added this comment}
Observe that for $r = o(\sqrt{m/\log m})$ the bound on $s$ is $(1-o(1)) {m \choose \leq r}$, which will give us a capacity-achieving result.

As discussed in Section~\ref{sec:techniques} (Theorem~\ref{thm:intro:high_deg_BEC}), to prove the theorem we have to understand the set of common zeroes of degree-$r$ polynomials. More accurately, we need to give an upper bound on the number of common zeroes of polynomials in some linear space. 

We start by introducing some notation and then discuss the reduction from Theorem~\ref{thm:lin_ind_eval_vectors} to the problem of determining the number of common zeroes of a space of polynomials.

%To prove this theorem, we first discuss the dual space to evaluation vectors, namely, the space of $m$-variate degree $r$ polynomials as it will be most useful in proving our main lemmas. 
Given a set of points $u_1,\ldots,u_s \in \F_2^m$ we define 
\begin{equation*}
\I(u_1,\ldots,u_s) = \{ f \in \mathbb{P}(m,r) \mid \forall i \;\; f(u_i)=0\}.
\end{equation*}
When $U$ is an $m\times s$ matrix we define $\I(U) = \I(u_1,\ldots,u_s)$, where  $u_i$ is the $i$th column of $U$. 
It is clear that $\I(U)$ is a vector space. 
Similarly, for a set of polynomials $F \subseteq  \mathbb{P}(m,r)$ we denote 
\begin{equation*}
\V(F) = \{ u \in \F_2^m \mid \forall f\in F \;\; f(u)=0\}.
\end{equation*}
In other words, $\V(F)$ is the set of common zeroes of $F$. 
From the definition it is clear that if $F_1 \subseteq F_2$ then $\V(F_2)\subseteq \V(F_1)$ and similarly, if $U_1 \subseteq U_2$ then $\I(U_2) \subseteq \I(U_1)$.

The next  lemmas explore the connection between the dual space of $U^r$, $\I(U)$ and $\V(\I(U))$. 
Hereafter we interpret a vector $f$ of length ${m\choose \leq r}$ as a polynomial in $\mathbb{P}(m,r)$, by viewing its coordinates as coefficients of the relevant monomials. We abuse notation and call this polynomial $f$ as well. 

\begin{lemma}\label{lem: I is dual to U}
Let $U$ be an $m\times s$ matrix. Then, a vector $f$ of length ${m\choose \leq r}$ satisfies $f\cdot U^r=0$ if and only if the corresponding polynomial $f(x_1,\ldots,x_m)$ is in $\I(U)$, namely, $f\in \I(U)$. 
\end{lemma}

\begin{proof}
The proof is immediate from the correspondence between vectors to polynomials and from the definition of $U^r$. Indeed, for a column $u_i$ we have that the coordinates of $u_i^{{r}}$ correspond to all evaluations of monomials of degree $\leq r$ on $u_i$. Similarly, the coordinates of the vector $f$ correspond to coefficients of the polynomial $f(x_1,\ldots,x_m)$. Thus, $f\cdot u_i$ is equal to $f(u_i)$. Hence, $f\cdot U =0$ if and only if $f(u_1)=\ldots,f(u_s)=0$, i.e. if and only if $f\in \I(U)$.
\end{proof}

\begin{lemma}\label{lem: new independent point}
Let $U$ be an $m\times s$ binary matrix. Then, for any $u \in \F_2^m$ we have that $u^{{r}}$ is in the linear span of the columns of $U^r$ if and only if 
$$ \I(U) = \I(U \cup \{u\}),$$
namely, every degree $\leq r$ polynomial that vanishes on the columns of $U$ also vanishes on $u$. 
\end{lemma}

\begin{proof}
It is clear that $u^{{r}}$  linearly depends on the columns of $U^r$ if and only if for every vector $f$ such that $f\cdot U^r=0$, it holds that $f\cdot u^{{r}}=0$, namely, that $f(u)=0$. By Lemma~\ref{lem: I is dual to U} this is equivalent to saying that $ \I(U) = \I(U \cup \{u\})$.
\end{proof}

%
%\begin{lemma}\label{lem: varieties and ideals}
%\begin{eqnarray*}
%F &\subseteq & \I(\V(F))\\
%\text{span}(\{\alpha_1,\ldots,\alpha_s\}) &= & \V\left(I\left(\text{span}\left(\{\alpha_1,\ldots,\alpha_s\}\right) \right)\right).
%\end{eqnarray*}
%\end{lemma}
%
%\begin{proof}
%The first inclusion is obvious so we only prove the second equality. Again, the inclusion $\subseteq$ is clear so we only prove the other direction. 
%Let $U$ be the $m\times s$ matrix whose columns are the $\alpha_i$. Let $\beta \in \F_2^m$ be linearly independent of the $\alpha_i$, namely, $\beta$ is not in the linear span of $U$'s columns.
%By Lemma~\ref{lem: new independent point} we get that $ \I(U \cup \{\beta\})\neq \I(U)$, i.e. that $ \I(U \cup \{\beta\})\subsetneq \I(U)$. In particular, there is a polynomial $f \in \I(U) \setminus  \I(U \cup \{\beta\})$. Thus, $f(\beta)\neq 0$, but $f$ vanishes on all of $U$. Hence, $\beta \not \in \V(f)$, yet $f \in \I(U)$. Thus, $\beta \not \in \V(\I(U))$.
%
%\end{proof}

Similarly, we get an equivalence when consider the common zeros of the polynomials that vanish on the columns of $U$. 

\begin{lemma}\label{lem:common_zero_in_U^r}
We have that $u \in \V(\I(U))$ if and only if $u^{{r}}$ is spanned by the columns of $U^r$.
\end{lemma}

\begin{proof}
If $u^{{r}}$ is spanned by the columns of $U^r$ then by Lemma~\ref{lem: new independent point} $ \I(U) = \I(U \cup \{u\})$. Thus, $u \in \V(\I(U \cup \{u\})) = \V(I(U))$. Conversely, if $u \in \V(\I(U))$ then $I(U) \subseteq I(U\cup\{u\})$. As $I(U\cup\{u\}) \subseteq I(U)$ we get $I(U\cup\{u\}) = I(U)$ and by Lemma~\ref{lem: new independent point} it follows that $u^{{r}}$ is spanned by the columns of $U^r$.
\end{proof}

Finally, we make the following simple observation.
\begin{lemma}\label{lem:IVI=I}
$\I(\V(\I(U))) = \I(U)$. 
\end{lemma}
\begin{proof}
Denote $\V=\V(\I(U))$.
It is clear that $U \subseteq \V$ and hence $\I(\V) \subseteq \I(U)$. 
On the other hand, let $f\in \I(U)$ and $v \in \V$. Lemmas~\ref{lem: new independent point} and \ref{lem:common_zero_in_U^r} imply that $\I(U) = \I(U\cup \{v\})$. Thus, $f(v)=0$ and hence $I(U) \subseteq \I(\V)$.
\end{proof}

%Lemma~\ref{lem: common zero in U^r} implies that in order to understand the probability that the evaluation vector of a randomly chosen $u\in \F_2^m$ is linearly independent of the columns of $U^r$ we need to understand the probability that $u$ belongs to $\V:=\V(\I(U))$. 

Going back to our original problem, assume that we picked $s$ columns at random and got linearly independent evaluation vectors. Now we have to understand the probability that  a randomly chosen $u\in \F_2^m$  will give an independent evaluation vector. By Lemma~\ref{lem:common_zero_in_U^r}, this amounts to understanding the probability that $u$ belongs to $\V:=\V(\I(U))$. By linear algebra arguments we have the following identity 
\begin{equation*}
\dim(\I(U)) = {m\choose \leq r} - \rank(U^r).
\end{equation*}
Thus, our goal is understanding how many common zeroes can the polynomials in an $\left({m\choose \leq r} - s\right)$-dimensional space have.

The way to prove that a given set of polynomials does not have too many common zeros is to show that any large set of points (in our case, $\V$) has many linearly independent degree-$r$ polynomials that are defined over it. That is, we only consider the restriction of polynomials of degree-$r$ to the points in $\V$ and we identify two polynomials if they are equal when restricted to $\V$. Notice that this is the same as showing that the rank of $E(m,r)[\V]$ is large, i.e., that there are many linearly independent columns that are indexed by elements of $\V$. Thus, two such polynomials $f,g$ are identified if and only if $f-g\in \I(\V)$.
Stated differently, we wish to show that the dimension of the quotient space $\mathbb{P}(m,r)/\I(\V)= \mathbb{P}(m,r)/\I(U)$ (by Lemma~\ref{lem:IVI=I}) is large. Indeed, if we can lower bound this dimension in terms of $\V$ then, since 
\begin{eqnarray*}
\rank(E(m,r)[\V]) = \dim(\mathbb{P}(m,r)/\I(\V)) =\dim(\mathbb{P}(m,r)/\I(U)) &=& \dim(\mathbb{P}(m,r)) - \dim(I(U)) \\&=& {m\choose \leq r} - \left({m\choose \leq r}-s\right)=s,
\end{eqnarray*} 
we will get that some function of $|\V|$ is upper bounded by $s$. Thus, unless $s$ is large, $|\V|$ is small and hence the probability that a randomly chosen $u$ in independent of $U^r$ is high. We state our main lemma next.

%We are now ready to prove .

%\begin{corol}\label{cor:GHW}
\begin{lemma}\label{lem:number_of_ind_polynomials}
Let $\V\subseteq \F_2^m$ such that $|\V| > 2^{m-t}$. Then there are more than ${m-t \choose \leq r}$ linearly independent degree $\leq r$ polynomials defined on $\V$, i.e., $\dim(E(m,r)[\V])>{m-t \choose \leq r}$.
\end{lemma}
%\end{corol}
%\Enote{did we define: ``defined on $\V$''}

We give two different proofs of this fact. The first uses a hashing argument;  if $\V$ is large, then, after some linear transformation, its projection on a set of roughly $\log(|\V|)$ many coordinates is full. Thus, it supports at least ${\log(|\V|) \choose \leq r}$ many linearly independent degree-$r$ {\em monomials}. The second proof relies on a somewhat tighter bound that was obtained by Wei \cite{Wei}, who studied the {\em generalized Hamming weight} of Reed-Muller codes. As Wei's result gives slightly tighter bounds compared to the hashing argument (although both lead to a capacity-achieving result), this is the proof that we give in the main body of the paper. For completeness, and as the hashing argument is more self contained we give it in Appendix~\ref{app:hash-proof}.\\

%To get the bound in the lemma we shall need the  notion of {\em generalized Hamming weight}. A more self-contained proof of a slightly worse bound that uses hashing is given in Appendix~\ref{app:hash-proof}.

We start by discussing the notion of generalized Hamming weight. Let $C\subseteq \F_2^n$ be a linear code and $D\subseteq C$ a linear subcode. We denote
$$\supp(D) = \{ i : \exists y \in D, \text{ such that } y_i \neq 0\}.$$
In other words, the support of $D$ is the union of the supports of all codewords in $D$.

\begin{defin}[Generalized Hamming weight]\label{dfn:GHW_1}
For a code $C$ of length $n$ and an integer $a$ we define 
$$d_a(C) = \min \{ \supp(D) \mid D\subseteq C \text{ is a linear subcode with } \dim(D) = a\}.$$ 
\end{defin}

Thus, $d_a(C)$ is the minimal size of a set of coordinates $S$, such that there exists a subcode $D$, of dimension $\dim(D)=a$, that is supported on  $S$. 
The reason for this definition is that for any code $C$ if we let $a=\dim(C)$ then $d_a(C)=n-d$, where $d$ is the minimal distance of $C$.
By considering the complement set $S^c$ the next lemma gives an equivalent definition of $d_a(C)$.

\begin{lemma}\label{lem:GHW_2}
For a code $C$ of length $n$ and an integer $a$ we have that 
$$d_a(C) = \max \{ b \mid \forall |S|< b \text{ we have that }\dim(C[S^c]) > \dim(C)-a\}.$$ 
\end{lemma}

\begin{proof}
The proof follows immediately from a simple linear algebra argument. If $D$ is a subcode of $C$ that is supported on a set of coordinates $S$ then $\dim(C) = \dim(D) + \dim(C[S^c])$. 
\end{proof}

The alternative definition given in Lemma~\ref{lem:GHW_2} is very close to what we need. We wish to show that for any large $\V$, there are many linearly independent degree $\leq r$ polynomials that are defined on $\V$. In other words, we wish to prove that $$d_{o(1){m\choose \leq r}}(RM(m,r)) \geq 2^m -  \e \cdot 2^m/{m\choose \leq r}.$$ Indeed, this will imply that for any $|\V|\geq \e \cdot 2^m/{m\choose \leq r}$ there are at least $(1-o(1)){m\choose \leq r}$ linearly independent degree $\leq r$ monomials defined on $\V$ ($\V$ plays the role of $S^c$ in Lemma~\ref{lem:GHW_2}).

The next theorem of Wei \cite{Wei}  computes exactly the generalized Hamming weight of Reed-Muller codes. For stating the theorem we  need the following technical claim.

\begin{lemma}[Lemma 2 of \cite{Wei}]\label{lem_rep_Wei}
For every $0\leq a\leq {m\choose \leq r}$ there is a unique way of expressing $a$ as $a = \sum_{i=1}^{\ell}{m_i \choose \leq r_i}$, where $m_i-r_i = m-r -i+1$.
\end{lemma}

\begin{thm}[\cite{Wei}]\label{thm:GHW}
Let $0\leq a\leq {m\choose \leq r}$ be an integer. Then, 
$d_a(RM(m,r)) = \sum_{i=1}^{\ell} 2^{m_i}$, where $a = \sum_{i=1}^{\ell}{m_i \choose \leq r_i}$ is the unique representation of $a$ according to Lemma~\ref{lem_rep_Wei}.
\end{thm}

We are now ready to prove Lemma~\ref{lem:number_of_ind_polynomials}.

%\begin{corol}\label{cor:GHW}
%Let $\V\subseteq \F_2^m$ such that $|\V| > 2^{m-t}$. Then there are more than ${m-t \choose \leq r}$ linearly independent degree $\leq r$ polynomials defined on $\V$.
%\end{corol}

\begin{proof}[Proof of Lemma~\ref{lem:number_of_ind_polynomials}]
\sloppy
%Denote $\V = \V(\I(U))$ and let $t=\log({m\choose \leq r}/\e)$.
For $a = \sum_{i=1}^{t}{m-i\choose \leq r-1}$, Theorem~\ref{thm:GHW} implies that $d_a(RM(m,r)) = \sum_{i=1}^{t} 2^{m-i} = 2^m - 2^{m-t}$.
Thus, if $|\V| > 2^m - d_a(RM(m,r))  = 2^{m -t}$ then there are more than ${m\choose \leq r} -  a = {m\choose \leq r} -  \sum_{i=1}^{t}{m-i \choose \leq r-1}$ many linearly independent degree-$r$ polynomials defined on $\V$. 
To make sense of parameters we shall need the following simple calculation. We give the straightforward proof in Section~\ref{app:missing-proofs}.
\begin{claim}\label{cla:binomial-difference}
${m\choose \leq r} -  \sum_{i=1}^{t}{m-i \choose \leq r-1} = {m-t \choose \leq r}$.
\end{claim}
Thus, if $|\V| > 2^{m-t}$ then there are more than ${m-t \choose \leq r}$ linearly independent degree $\leq r$ polynomials defined on $\V$.
\end{proof}

We can now bound the number of common zeroes of all polynomials that vanish on a given set $U$.

\begin{lemma}\label{lem:number_of_common_zeroes}
Let $\e>0$ be a constant and $s$ an integer such that $s < {m- \lceil \log({m\choose \leq r})+ \log(1/\e)\rceil \choose \leq r}$.
Let $U$ be an $m\times s$ matrix such that $\rank(U^r)=s$, namely, the columns of $U^r$ are linearly independent. Then, $|\V(\I(U))| \leq \e\cdot 2^m/{m\choose \leq r}$.
\end{lemma}

%We are now ready to give the proof of Lemma~\ref{lem: number of common zeroes} .

\begin{proof}%[Proof of Lemma~\ref{lem: number of common zeroes} ]
By the discussion above we know that $\dim(\mathbb{P}(m,r)/\I(U))=s$. I.e., the dimension of the space of degree $\leq r$ polynomials that are defined on $\V=\V(\I(U))$ is $s$. 

Let $t$ be the minimal integer such that $2^{m-t} \leq \e\cdot \frac{2^m}{{m\choose \leq r}}$, i.e., $t = \lceil \log\left({m\choose \leq r}\right)+\log (1/\e) \rceil $. 
Assume towards a contradiction that $|\V| > \e\cdot \frac{2^m}{{m\choose \leq r}} \geq 2^{m-t}$. Lemma~\ref{lem:number_of_ind_polynomials} implies that there are more than ${m-t \choose \leq r}$ linearly independent polynomials defined on $\V$.

As there are at most $s$ polynomial defined on $\V$ we must have 
$$s > {m-t \choose \leq r} = {m-\lceil \log({m\choose \leq r})+ \log(1/\e)\rceil \choose \leq  r},$$
in contradiction to the assumption of the lemma.
\end{proof}

We can now derive Theorem~\ref{thm:lin_ind_eval_vectors}, the main result of this subsection.  

\begin{proof}[Proof of Theorem~\ref{thm:lin_ind_eval_vectors}]
We start by picking random points one after the other. Assume that we picked $\ell< s$ points and they gave rise to linearly independent evaluation vectors. 
By Lemmas~\ref{lem:common_zero_in_U^r} and \ref{lem:number_of_common_zeroes}, as long as $s\leq {m-\lceil \log({m\choose \leq r})+ \log(1/\e)\rceil \choose \leq r}$, the probability that the next random point will not generate an independent evaluation vector is at most $\e / {m\choose \leq r}$. Repeating this argument ${m-\lceil \log({m\choose \leq r})+ \log(1/\e)\rceil \choose \leq r} < {m\choose \leq r}$ times, the probability that we do not get a set of independent evaluation vectors is at most $\e$.
%This is exactly the statement of Theorem~\ref{thm: lin ind eval vectors}. 
%To ease the reading we repeat the statement of the theorem here. 
\end{proof}

%\begin{theoremNoNum}[Theorem~\ref{thm: lin ind eval vectors}]\sloppy 
%Let $\e>0$ be a real number and $k,m,r$ integers such that $k<{m- \log({m\choose \leq r})- \log(1/\e) \choose \leq r}$. Then, with probability larger than $1-\e$ if we pick $\alpha_1,\ldots,\alpha_k\in \F_2^m$ uniformly at random we get that the evaluation vectors, $\alpha_1^{{r}},\ldots, \alpha_k^{{r}}$ are linearly independent. 
%\end{theoremNoNum}

%%%%%%%%%%%%%%%%%%%%%%%%%%%%
%%%%%%%%%%%%%%%%%%%%%%%%%%%%
%%%%%%%%%%%%%%%%%%%%%%%%%%%%

%\Amnote{Changed name of section}
\subsection{Random submatrices of $E(m,r)$, for small $r$, have full row-rank} \label{second-question}
%\Enote{This seems to be a bad title, as literally, the rows of $E(m,r)$ are linearly independent..}

%\Anote{In Thm 5.15 and the ensuing proof, I am not sure why we call the vectors $\alpha_i$ as opposed to $u_i$ as per our notation for m-vectors. Also, in the Thm statement I would remove "is real" (also $\eta$ is real)}

%\Anote{It may be good to provide a bit of intuition before calculation, namely say that in each dyadic interval the probability of missing all nonzeros in k experiments balances perfectly with the  small number of polys of that weight. Also, that the argument is separate for polys of weight below and above $(1-\e)/2$.}

%\Anote{I would change the very last numerical expression estimating the prob in the very last paragraph, so that we don't have simultaneously ``$exp(-...)$” and “$2^{-...}$” added up - it looks weird and we should choose one or the other.}

In this section we study Question~\ref{question:rows} for the case of low degree. Thus, our goal is proving that, with high probability, a random set $U$ of columns of $E(m,r)$, of the appropriate size, has full row-rank. As we showed in Corollary~\ref{cor:equivalence}, this is equivalent to studying Question~\ref{question:columns} in the case of high degree. We prove the following theorem.

\begin{thm}\label{thm:random_span_using_KL}
%There exists a constant $\eta>0$ such that the following holds. 
Let $0<\delta<1/3$, $\eta = O(1/\log(1/\delta))$ and $m,r$ integers such that $r\leq\eta{m}$.  Let $s=\lceil(1+\delta){m\choose \leq r}\rceil$. Then, except of probability not larger than $\exp\left(-\Omega\left( \min(\delta,2^{-r}) \cdot {m\choose \leq r}\right)\right)$, if we pick $u_1,\ldots,u_s\in \F_2^m$ uniformly at random we get that the evaluation vectors, $u_1^{{r}},\ldots, u_s^{{r}}$ span the entire vector space $\F_2^{{m\choose \leq r}}$. 
\end{thm}
%\Enote{Please check theorem and proof}

As in the proof of Theorem~\ref{thm:lin_ind_eval_vectors} we will consider the dual space to evaluation vectors - the space of degree-$r$ polynomials. Our proof strategy is to show that for every possible polynomial of degree-$r$, with high probability, there is at least one point in our set on which the polynomial does not vanish. Thus, by the discussion in Section~\ref{first-question}, this means that the dual space contains only the zero polynomial, which is exactly what we wish to prove. 
To apply this strategy we will need to know the number of polynomials that have a certain number of nonzeroes. Such an estimate was given in Theorem~\ref{thm:wt-dist}, 
%and Corollary~\ref{cor:wt-dist-simple} 
following \cite{KLP}.

%[{\bf EA:} should it be $\Omega$ in the error probability below?]
%\Amnote{Restated the theorem to include dependence of $\eta$ on $\delta$}

%The idea of the proof is to show that the dual to the span is empty, namely, that $\I(U)=\{0\}$. To obtain this we show that no polynomial belongs to $\I(U)$ by proving that with high probability we pick a nonzero of each polynomial.

\begin{proof}

Set $\e = \delta /4$.
For an integer $1\leq \ell \leq r$ we denote with $P_\ell$ the set of degree $\leq r$ polynomials whose fraction of nonzeros is between $(1-\e)2^{-\ell-1}$ and $(1-\e)2^{-\ell}$. By Theorem~\ref{thm:wt-dist}, 
$$|P_\ell| \leq W_{m,r}((1-\e)2^{-\ell})\leq (1/\e)^{8c \ell^4 {m-\ell \choose \leq r-\ell} }.$$

Let $f$ be some polynomial in $P_\ell$. When picking $s$ points at random, the probability that $f$ vanishes on all of them is at most $(1 - (1-\e)2^{-\ell-1})^s$. Thus, the probability that any $f\in P_\ell$ vanishes on all $s$ points is at most $|P_\ell| \cdot (1 - (1-\e)2^{-\ell-1})^s$. We would like to show that this probability is small, when ranging over all $\ell$. We first study the case that $\ell>0$, i.e., of the contribution from the polynomials that have at most $(1-\e)/2$ nonzeros. Although in this case the probability of hitting a nonzero gets smaller and smaller as $\ell$ grows, the size of $P_\ell$ goes down at a faster rate (by Theorem~\ref{thm:intro:wt-dist}), and therefore we get an exponentially small probability of missing some polynomial. When $\ell=0$, although the number of polynomials of such high weight is huge, the probability of missing any of them is tiny, and so we are ok in this case as well. We now do the formal calculation.

By the above, the probability that some polynomial whose fraction of nonzeros is at most $(1-\e)/2$ vanishes on all $s$ points is at most 
\begin{equation}\label{eq:sum-over-dual}
\sum_{\ell=1}^{r-1}|P_\ell| \cdot (1 - (1-\e)2^{-\ell-1})^s.
\end{equation} 
Indeed, any degree-$r$ polynomial is nonzero with probability at least $2^{-r}$ and hence the above summation goes over all such polynomials.

Let us estimate a typical summand,
\begin{eqnarray*}
|P_\ell| \cdot \left(1 - (1-\e)2^{-\ell-1}\right)^s &\leq &(1/\e)^{8c \ell^4 {m-\ell \choose \leq r-\ell} }\cdot (1 - (1-\e)2^{-\ell-1})^s \\ &\leq &(1/\e)^{8c \ell^4 {m-\ell \choose \leq r-\ell} }\cdot \exp( - (1-\e)2^{-\ell-1}s) \\&\leq&
(1/\e)^{8c \ell^4 {m-\ell \choose \leq r-\ell} }\cdot \exp\left(- (1+\delta)(1-\e){m\choose \leq r}2^{-\ell-1}\right).
\end{eqnarray*}
%Thus, if $r\leq \eta m$, for a small enough constant $\eta>0$ (e.g., $\eta = O(1/\log(1/\e))$ suffices) then, 
Let $\eta = O(1/\log(1/\e))= O(1/\log(1/\delta))$.
From the fact that 
$${m-\ell \choose \leq r-\ell} \leq {m\choose \leq r}\cdot \left(\frac{r}{m}\right)^\ell,$$
we get by simple manipulations that if $r\leq \eta m$ then
%Continuing we get
\begin{eqnarray*}
|P_\ell| \cdot \left(1 - (1-\e)2^{-\ell-1}\right)^s &\leq & \exp\left({m\choose \leq r}\cdot \left( 3^{-\ell} - (1+\delta)(1-\e)2^{-\ell-1} \right) \right)\\ 
%&\leq & \exp \left(c \log(1/\e) \left(mr^4 + 10{\ell^2}\cdot  {m\choose \leq r}\left(\frac{r}{m}\right)^\ell \right) - (1+\delta)(1-\e)2^{-\ell-1}     \right)\\
%&\leq^{(\dagger)} & \exp \left( - {m\choose \leq r} 2^{-\ell-2}  \right),
&\leq &  \exp\left(-\Omega\left({m\choose \leq r}\cdot  2^{-\ell}\right)\right).
\end{eqnarray*}
%where inequality $(\dagger)$ holds for $r \leq \frac{\e}{5\sqrt{c}} \sqrt{m} = \frac{\delta}{20\sqrt{c}} \sqrt{m}$. This can be seen, e.g., by verifying the case $\ell=1$ and noting that every other value of $\ell$ also satisfies this inequality.

Going back to Equation~\eqref{eq:sum-over-dual} we see that 
$$\sum_{\ell=1}^{r-1}|P_\ell| \cdot (1 - (1-\e)2^{-\ell-1})^s \leq \sum_{\ell=1}^{r-1} \exp\left(-\Omega\left({m\choose \leq r}\cdot  2^{-\ell}\right)\right) \leq  \exp\left(- \Omega\left(2^{-r} {m\choose \leq r}\right)\right).$$ 
All we have to do now is bound from above the probability that there exists a polynomial that has more than $(1-\e)/2$ fraction of nonzeros, but that vanishes on all the  $s$ chosen points. As there are at most $2^{{m\choose\leq  r}}$ such polynomials, this probability can be bounded from above by %[{\bf EA:} it should be $\leq r$?]
\begin{eqnarray*}
2^{{m\choose \leq r}}\cdot \left( 1- \frac{1-\e}{2}\right)^s & =  & 2^{{m\choose \leq r}} \cdot \left( \frac{1+\e}{2}\right)^s \\ 
&\leq & 2^{{m\choose \leq r}} \cdot \left( \frac{1+\e}{2}\right)^{(1+\delta){m\choose \leq r}}\\
&=& \left((1+\e)^{1/\delta}\cdot \left( \frac{1+\e}{2}\right)  \right)^{\delta{m\choose \leq r}}\\
&\leq^{(*)}& \left(e^{1/4}\cdot \left( \frac{1+\e}{2}\right)  \right)^{\delta{m\choose \leq r}}\\
&\leq^{(\dagger)}& \exp\left(-\frac{1}{3}\delta{m\choose \leq r}\right),
\end{eqnarray*}
where inequality $(*)$ holds since $\e=\delta/4$ and inequality $(\dagger)$ follows by our choice $\delta < 1/3$.
%&\leq & 2^{{m\choose \leq r}} \cdot 2^{-(1-\log(e) \cdot \e)s}\\ &\leq& 2^{{m\choose \leq r} \cdot \left( 1 - (1-\log(e) \cdot \e)(1+\delta)  \right) }.
%By our choice of parameters it holds that 
%$$ 1 - (1-\log(e) \cdot \e)(1+\delta) < -\delta/2. $$
%Hence, it follows that 
%\begin{eqnarray*}
%2^{{m\choose \leq r}} \cdot \left( 1- \frac{1-\e}{2}\right)^s \leq 2^{-\frac{\delta}{2}{m\choose \leq r}}.
%\end{eqnarray*}
Concluding, when picking $s$ points at random, the probability that there will be some polynomial of degree $\leq r$ that does not vanish on any of the chosen points is at most $$\exp\left(- \Omega\left(2^{-r} {m\choose \leq r}\right)\right)+ \exp\left(-\frac{1}{3}\delta{m\choose \leq r}\right)=\exp\left(- \Omega\left(\min(\delta,2^{-r}) {m\choose \leq r}\right)\right).$$  
%\leq \exp\left(-\frac{1}{4}\delta{m\choose \leq r}\right)$.
By simple linear algebra and Lemma~\ref{lem: I is dual to U} it follows that if no such polynomial vanishes on all points that we picked then their evaluation vectors span the entire space. This concludes the proof of the theorem. 
\end{proof}

%%%%%%%%%%%%%%%%%%%%%%%%%%%%%%%%%%%%%%%%%%
%%%%%%%%%%%%%%%%%%%%%%%%%%%%%%%%%%%%%%%%%%
%%%%%%%%%%%%%%%%%%%%%%%%%%%%%%%%%%%%%%%%%%
%%%%%%%%%%%%%%%%%%%%%%%%%%%%%%%%%%%%%%%%%% 
\section{Reed-Muller code for erasures}\label{sec:RM-erasures}

\subsection{Low-rate regime}

%\Anote{I think we can remove this remark. We need to mention it higher, with the bound on $r$ that works}
%(EA: remove this now?) Corollary \ref{corol:RM_for_BEC} shows that the $RM(m,m-r)$ code can be used to communicate reliably over a BEC of erasure rate $\e=k/n$, where $k= [m- \log({m\choose \leq r})- \log(1/\e)  -2]^r$. If $r$ is small enough such that $\e=m^r/n + o(m^r/n)$, the rate of $RM(m,m-r)$ is given by $1-\e -o(\e)$, hence, in that sense, RM codes are capacity achieving at high-SNR over the BEC. (EA: remove this?)

%\Amnote{Need to add result for BEC for low rate}

%\Amnote{I wrote the corollary for the case $r=o(m)$. We can write something more detailed that give the dependence between $\eta$ and $\delta$. What do you think?}

Recall that from Corollary~\ref{erasure-G}, if $G$ is an $n \times k$ generator matrix of a code, then the code can correct $s$ random erasures if erasing a random subset of $n-s$ rows in $G$ (which gives the random matrix $G_{n-s,\cdot}$) has full span, i.e., 
\begin{align*}
\Pr\{  \rank (G_{n-s,\cdot})=k  \} \to 1, \quad \text{as } n \to \infty. \quad 
%\Pr_D\{ \sum_{j \in D} H[j] =0 \} \to 0, \quad \text{as } n \to \infty, \quad 
\end{align*}
Hence, since the transpose of $E(m,r)$ is an $n \times k$ generator matrix for $RM(m,r)$, we get our result for low-rate RM codes on the BEC as a direct consequence of Theorem \ref{thm:random_span_using_KL}.

%Theorem \ref{thm:random_span_using_KL} and Corollary~\ref{erasure-G} imply our second result for the performance of low rate RM codes over the BEC. 

\begin{corol}\label{corol_low_rate_BEC}
%There exists a constant $\eta>0$ such that the following holds. 
Let $0<\delta<1/3$, $\eta = O(1/\log(1/\delta))$ and $m,r$ integers such that $r\leq\eta{m}$.  
Then, $RM(m,r)$ can correct $2^m-(1+\delta){m\choose \leq r}$ random erasures, i.e., it is $\delta$-close to capacity-achieving.
Moreover, if $r=o(m)$, then $RM(m,r)$ is capacity-achieving on the BEC. 
\end{corol}

%\Enote{Made the above more precise. I think there was also a sign typo in the previous statement, see if you agree with the $+\delta$.}

%\begin{proof}
%Let $E$ denotes the location of the $k$ erasures, uniformly drawn in $[n,k]=\{ S \subseteq [n]: |S|=k\}$. Note that the set of bad erasure patterns are the ones that can confuse codewords in $RM(m,m-r)$, i.e., 
%\begin{align*}
%\mathcal{B}(k,m,r)&:=\{E \in [n,k]: \exists x,y \in RM(m,m-r), x \neq y, x[E^c]=y[E^c]\}\\
%&= \{E \in [n,k]: \exists z \in RM(m,m-r), z \neq 0, z[E^c]=0\}\\
%&= \{E \in [n,k]: \exists z \in \ker H(m,r), z \neq 0, z[E^c]=0\}, \label{dist1}
%%&= \{z \in B(n,k): \exists z \in \ker H(m,r), z \neq 0\},
%\end{align*}
%where we used the fact that the code is linear in the second equality. From \eqref{dist1}, a bad error pattern $E$ is one containing a non-empty subset $S$ such that $\sum_{j \in S} H(m,r)[\cdot,j]=0$. Hence, for $E \in \mathcal{B}(k,m,r)$, it must be that $H(m,r)[E]$ has linearly dependent columns, and conversely, if $E \in [n,k]$ is such that $H(m,r)[E]$ has linearly dependent columns, it must be that $E \in  \mathcal{B}(k,m,r)$.
%Therefore,
%\begin{align*}
%\mathcal{B}(k,m,r)&= \{E \in [n,k]: H(m,r)[E] \text{ has linearly dependent columns} \}, 
%\end{align*}
%which has probability at most $\e$ for $k=\lfloor {m-3 \log({m\choose \leq r})+ \log(\e) \choose r}\rfloor -1$ from Theorem \ref{thm: lin ind eval vectors}.  
%\end{proof}

\subsection{High-rate regime}

%To use Corollary~\ref{erasure-G}, as in the case of low degree, we first  show that due to the duality of RM codes, an answer to Question~\ref{question:columns} implies an answer to Question~\ref{question:rows} in the high-rate regime.
%these two questions can be related to each other, however by studying different regimes of the parameters.  

%\Amnote{We should think what to write here in light of the results in section 3.4}

We now use the parity-check matrix interpretation of correcting errors, namely from Corollary~\ref{erasure-H}, a code with parity-check matrix $H$ can correct $s$ errors if the random set of $s$ columns (which we denote $H[s]$) are linearly independent with high probability, i.e., 
\begin{align*}
\Pr\{ \mathrm{rk} (H[s])=s \} \to 1, \quad \text{as } n \to \infty. \quad 
%\Pr_D\{ \sum_{j \in D} H[j] =0 \} \to 0, \quad \text{as } n \to \infty, \quad 
\end{align*}
Using now Theorem~\ref{thm:lin_ind_eval_vectors} and the fact that $E(m,r)$ is a parity-check matrix for $RM(m,m-r-1)$ (see Lemma \ref{duality}), we obtain our result for the performance of high-rate RM codes on the BEC. 
\begin{corol}\label{corol:RM_for_BEC}
Let $\e>0$, $r \leq m$ be two positive integers and $s= \lfloor{m- \log({m\choose \leq r}) - \log(1/\e) \choose \leq r}\rfloor$.
Then, $RM(m,m-r)$ can correct $s$ random erasures with probability larger than $1-\e$.  
In particular, if $m-r = o(\sqrt{m/\log m})$, then $RM(m,r)$ is capacity-achieving on the BEC. 
\end{corol}
%\Enote{Changed the above.}

The following calculation gives a better sense of the parameters (the proof is in Section \ref{app:missing-proofs}).
\begin{claim}\label{cla: estimation r small}
For $r < \sqrt{\frac{\delta m}{4\log(m)}}$ and $\e > m^{-r/2}$ we have that ${[m- \log({m\choose \leq r})- \log(1/\e) \choose \leq r} > (1-\delta){m \choose \leq r}$.
\end{claim}

%%%%%%%%%%%%%%%%%%%%%%%%
%%%%%%%%%%%%%%%%%%%%%%%%
%%%%%%%%%%%%%%%%%%%%%%%%

\section{Reed-Muller code for errors}\label{sec:BSC-proofs}
We present next results for errors at both low and high rate. The results at low-rate rely on the weight distribution results of Section~\ref{sec:wt-dist}, whereas the high-rate results rely on a novel relation between decoding from errors and decoding from erasures.

\subsection{Low-rate regime}

%\Enote{use current bound on the weight distribution}
%\Enote{Perhaps also write explicitly the dependence of $\eta$ on $\delta$ in the statement of the theorem as I did in Theorem~\ref{thm:random_span_using_KL}}

\begin{thm}\label{thm:errors_from_weights}
%There exists a constant $\eta>0$ such that the following holds. 
Let $\delta>0$. There exists\footnote{The exact value of $\eta$ is given in \eqref{eq:eta7}.} $\eta = O(1/\log(1/\delta))$ such that the following holds. 
For any two integers $r$ and $m$ satisfying $r/m\leq\eta$, and any $p$ satisfying  
\begin{align*}
1-h(p)=(1+\delta)R, \quad \text{where } R=\frac{{m\choose \leq r}}{n},
\end{align*}
$RM(m,r)$ can correct $pn$ random errors with probability at least $\exp\left(-\Omega\left( \min(\delta,2^{-r}) \cdot {m\choose \leq r}\right)\right)$
%\begin{align*}
%1- \exp\left(-\min\left(\Omega \left(2^{-r} {m\choose \leq r} \right),\Omega\left(\delta {m\choose \leq r}\right)\right)\right) .
%\end{align*}
In particular, for $r=o(m)$, $RM(m,r)$ is capacity-achieving on the BSC.  
\end{thm}
%Note that for $r/m=\eta$, the above is $1-\exp\left( -\min(\Omega( n^{h(\eta)-\eta} ),\Omega( \delta n^{h(\eta)} ))\right)$, which tends to 1 since $h(\eta)> \eta$ for small $\eta$.  

Before giving the proof we make a small calculation to get a better sense of what parameters we should expect. Since $R$ is small we can expect to correct a fraction of errors approaching $1/2$. Let us denote $p= (1-\xi)/2$. We now wish to figure how small should $\xi$ be. At corruption rate close to $1/2$ we have that 
\begin{align}
h(p)=h(1/2 - \xi/2) = 1- \xi^2/(2 \ln(2)) + \Theta(\xi^4). \label{ent-approx}
 \end{align}
Thus, if we wish to have $(1+\delta)R=1-h(p)$ then $\xi$ should satisfy
\begin{align}
(1+\delta)R= 1-h(p)= \xi^2/(2 \ln(2)) - \Theta(\xi^4). \label{eq:xi}
 \end{align}
 Hence, $\xi^2 =\Theta(R)$. 

We now give the proof following the outline described in Section~\ref{sec:techniques}.

\begin{proof}
Let $s$ be the number of errors, i.e., $s=pn$ and $p=1/2 - \xi/2$. 
%We can assume that $2s \geq d=2^{m-r}$, i.e., the number of errors is more than half the minimum distance of the code. 
%Since  
%\begin{align}
%H(1/2 - \xi/2) = 1- \xi^2/(2 \ln(2)) + O(\xi^4), \label{ent-approx}
% \end{align}
%we expect to recover a fraction of errors scaling as $p=1/2-\xi(n)/2$ where $\xi(n)$ tends to zero at fastest as 
%\begin{align}
%\frac{n\xi(n)^2}{2 \ln(2)} = (1+\delta) {m\choose \leq r},
%\end{align}
%for $\delta >0$.  
A bad error pattern $z \in \F_2^n$ is one for which there exists another error pattern $z' \in \F_2^n$, of weight $s$, such that $z+z'$ is a codeword in $RM(m,r)$. 
We concentrate on the case that $w(z')=s$ as this is the most interesting case. 

%\Enote{Removed 'We shall concentrate on the case that $w(z')=s$ as this is the most interesting case' since the model is for fixed amounts of errors}

Note that since both $z$ and $z'$ are different and have the same weight, the weight of $z+z'$ must be even and in $\{d,\dots,2s\}$. As both $z+z'$ and the all $1$ vector  are codewords, we also have that the weight of $z+z'$ is at most $n-d$, hence $w(z+z')\in \{d,\dots,n-d\}$.
% since $\xi$ is taken small.  
 Therefore, counting the number of bad error patterns is equivalent to counting the number of weight $s$ vectors that can be obtained by ``splitting'' codewords of even weight in $\{d,\dots,n-d\}$. 
Note that for a codeword $y$ of weight $w(y)=w$, there are $w/2$ choices for the support of $z$ inside $\supp(y)$ and $s-w/2$ choices outside the codeword's support. Indeed, $z$ and $z'$ must cancel each other outside the support of $y$ and hence they have the same weight inside $\supp(y)$.\footnote{When $z$ and $z'$ do not have the same weight they still have to cancel each other outside $y$. Thus, $\supp(z)$ must have intersection $w/2 + (w(z)-w(z'))/2$ with $\supp(y)$.} It follows, that for a fixed $y$ there are  
 \begin{align*} 
 {w \choose w/2}{n-w \choose s-w/2}
 \end{align*}
possibilities to pick a bad error pattern with intersection $w/2$ with $\supp(y)$. Denoting by $\mathcal{B}$ the set of bad error patterns and $N_{m,r}(w)$ the number of codewords of weight $w$ in $RM(m,r)$, a union bound gives
 \begin{align*}
 \Pr\{\mathcal{B}\} &\leq \sum_{w \in \{d,\dots,n-d\}} N_{m,r}(w) \frac{{w \choose w/2}{n-w \choose s-w/2}}{{n \choose s}} . \end{align*}
 We are now going to prove that $\Pr[{\mathcal B}]$ is exponentially small, for our setting of parameters, which will imply the theorem.
Since for\footnote{The binomial coefficient should be defined for the rounding of $\alpha n$ with either ceiling or floor functions.}
 $\alpha \in (0,1)$, 
\begin{align*}
 2^{n h(\alpha) - O(\log(n))} \leq {n \choose \alpha n} \leq  2^{n h(\alpha)} ,
\end{align*}
and recalling the entropy approximation of \eqref{ent-approx}, we have, by defining $\beta=w/n$,     
 \begin{align*}
 \Pr\{\mathcal{B}\} &\leq \sum_{\beta \in \{d/n,\dots,1-d/n\}} N_{m,r}(\beta n) \frac{2^{\beta n} 2^{n \left(1-\beta  - \frac{\xi^2}{(1-\beta )2 \ln(2)} + O\left(\frac{\xi^4}{(1-\beta )^2}\right) \right)} }{2^{n\left(1-\frac{\xi^2}{2 \ln(2)}\right) - O(\log n)}} \\ 
 &=\sum_{\beta  \in \{d/n,\dots,1-d/n\}} N_{m,r}(\beta n) 2^{-\frac{n \xi^2 }{2 \ln(2)}  \frac{\beta }{(1-\beta )} +nO\left(\frac{\xi^4}{(1-\beta )^2}\right) + O(\log n)}. 
 \end{align*}
Let $\e= \delta/3$.  
We next upper bound the above summation by grouping codewords of weights between $(1-\e)2^{-\ell-1}$ and  $(1-\e) 2^{-\ell}$, with $\ell \in \{1,2,\dots,r-1\}$. 
%For $\ell \in \{1,2,\dots, r-1 \}$, we use Theorem \ref{}. 
For codewords of weight close to $1/2$, we use the fact that there are $2^{{m \choose \leq r}}$ codewords in $RM(m,r)$.
Since the function $\beta \to \beta /(1-\beta )$ is increasing, we obtain the following bound where $|P_\ell|$ is, as before, the number of codewords having weights between between $(1-\e)2^{-\ell-1}$ and  $(1-\e) 2^{-\ell}$:
 \begin{align}
 \Pr\{\mathcal{B}\}  &\leq \sum_{\ell \in \{1,2,\dots,r-1\}} |P_{\ell}| 2^{-\frac{n \xi^2}{2 \ln(2)}  \frac{(1-\e) 2^{-\ell-1}}{(1-(1-\e) 2^{-\ell-1})}+nO\left(\frac{\xi^4}{(1-(1-\e) 2^{-\ell-1})^2}\right) + O(\log n)} \label{half1} \\
 & \hspace{0.8in} + 2^{{m \choose \leq r}} 2^{-\frac{n \xi^2}{2 \ln(2)}  \frac{(1-\e)/2}{(1-(1-\e)/2)}+ nO\left(\xi^4/(1+\e)^2\right)+ O(\log n)} . \label{half2}
 \end{align}
Using the inequality of Theorem~\ref{thm:wt-dist}:   
\begin{align*}
|P_\ell| \leq  (1/\e)^{8c \ell^4 {m-\ell \choose \leq r-\ell} },
\end{align*}
and noting that by \eqref{eq:xi},  $\frac{\xi(n)^2}{2 \ln(2)} = (1+\delta)R + O(\xi^4) = (1+\delta)R+\Theta(R^2)$, 
the term in \eqref{half2} is bounded as
 \begin{align*}
 2^{-{m \choose \leq r}\left((1+\delta) \frac{(1-\e)/2}{1-(1-\e)/2} -1+\Theta(R)\right)+ O(m)}= 2^{-{m \choose \leq r}\left((1+\delta) \frac{1-\e}{1+\e} -1+\Theta(R)\right)+ O(m)} ,
 \end{align*}
which vanishes since $\e <\delta/2$ (recalling that $R=o(1)$) and is overall $2^{-\Omega\left(\delta {m\choose \leq r}\right)}$. 

For the summands in \eqref{half1}, they are bounded as 
 \begin{align*}
& (1/\e)^{8c \ell^4 {m-\ell \choose \leq r-\ell} } 2^{ -  (1+\delta) {m\choose \leq r}  \frac{(1-\e) 2^{-\ell-1}}{1-(1-\e) 2^{-\ell-1}} + \Theta(R) +nO(\frac{\xi^4}{(1-(1-\e) 2^{-\ell-1})^2})+ O(m)} 
 \end{align*}
 where the above exponent is upper bounded
\begin{align}
-{m\choose \leq r} \left(   (1+\delta) \frac{(1-\e) 2^{-\ell-1}}{1-(1-\e) 2^{-\ell-1}} +\Theta(R) - 8c \log(1/\e){\ell^4}\cdot  \left(\frac{r}{m}\right)^\ell  \right)+ O(m). \label{factor}
 \end{align}
Note that the second bracket in this exponent is given by  
 \begin{align}
%\begin{equation} \label{expo}
&(1+\delta) \frac{(1-\e) 2^{-\ell-1}}{(1-(1-\e) 2^{-\ell-1})} - 8c \log(1/\e){\ell^4}\cdot  \left(\frac{r}{m}\right)^\ell \\
%\end{equation}
 &= (1-\e) 2^{-\ell-1} \left(\frac{1+\delta}{1-(1-\e) 2^{-\ell-1}} -  16c  \frac{\log(1/\e)}{1-\e}{\ell^4}\cdot  \left(\frac{2r}{m}\right)^\ell \right), \label{expo}
 \end{align}
which is positive for $r/m < \frac{1+\delta}{8c\frac{3 +\e}{1-\e}\log(1/\e)}$, and since $\e < \delta/2$, it is equal to $\Omega\left( (1-\e) 2^{-\ell-1}\right)$ for 
\begin{equation}\label{eq:eta7}
r/m \leq \eta=\frac{1+\delta}{9c\frac{3 +\delta/2}{1-\delta/2}\log(2/\delta)}.
\end{equation}
Note that this condition on $r/m$ is obtained from the extremal case $\ell=1$, which minimizes the term $\left(\frac{1+\delta}{1-(1-\e) 2^{-\ell-1}} -  16c  \frac{\log(1/\e)}{1-\e}{\ell^4}\cdot  \left(\frac{2r}{m}\right)^\ell\right)$ in \eqref{expo}.
Thus,  $$\eqref{factor} =-\Omega\left( (1-\e) 2^{-\ell-1}{m\choose \leq r}\right) \leq -\Omega\left(2^{-r}{m\choose \leq r}\right).$$
%For $\ell=r-1$, the term $20c{\ell^2}\cdot  \left(\frac{2r}{m}\right)^\ell \frac{\log(1/\e)}{1-\e}$ vanishes in \eqref{expo}, so that \eqref{expo} becomes $O(2^{-r})$, and \eqref{factor} becomes $\left(-\Omega({m\choose \leq r} 2^{-r})\right)$. Finally, since there are  $r-1$ terms in the summation \eqref{half1}, the overall sum is $2^{-\Omega({m\choose \leq r} 2^{-r})}$, and 
Concluding, the overall probability $\Pr\{\mathcal{B}\}$ is bounded as
\begin{align*}
\Pr\{\mathcal{B}\} &\leq 2^{-\Omega\left({m\choose \leq r} 2^{-r}\right)} + 2^{-\Omega\left(\delta {m\choose \leq r}\right)} = 2^{-\Omega\left(\min(\delta,2^{-r}) {m\choose \leq r}\right)}.  
\end{align*}
\end{proof}

\subsection{High-rate regime}

In this section we prove our main result for the BSC in the high rate regime.

\begin{thm}\label{thm:main_for_BSC}
Let $\e>0$, $r \leq m$ two positive integers and $s=\lfloor {m- \log({m\choose \leq r})+ \log(\e) \choose \leq r}\rfloor -1$.
Then $RM(m,m-(2r+2))$ can correct a random error pattern of weight $s$ with probability larger than $1-\e$.  
\end{thm}

Using Claim \ref{cla: estimation r small}, Theorem~\ref{thm:main_for_BSC} gives the following corollary.

%\Amnote{Verify correctness after section 6 is done}

\begin{corol}\label{cor:main_for_BSC}
Let $\e,\delta>0$, $r \leq m$ two positive integers such that $r < \sqrt{\frac{\delta m}{4\log(m)}}$ and $\e >  m^{-r/2}$. Then $RM(m,m-(2r+2))$ can correct a random error pattern of weight $(1-\delta){m \choose \leq r}$ with probability larger than $1-\e$.
\end{corol}

%\Enote{Technically, it sufficient to say in previous two theorems that "RM can correct x random errors", in view of our formal definitions. But it s also fine to be more explicit for the theorems.}

%We next prove Theorem \ref{thm:main_for_BSC}.

The rest of this section is organized as follows. We first give a combinatorial view of the syndrome of an error pattern under $E(m,r)$ (Section~\ref{sec:patterns_parity}). We then study the case of $E(m,3)$, which corresponds to the case $r=1$ in Theorem~\ref{thm:main_for_BSC} (as $E(m,3) = H(m,m-4)$), in Section~\ref{sec:deg_3}. The case of general degree-$r$ is handled in Section~\ref{sec:deg-r}. Then, in Section~\ref{sec:general_deg_3} we extend the case $r=1$ to hold for arbitrary linear codes of high degree and in Section~\ref{sec:counterex} we prove that our results for the case $r=1$ are tight, in some sense.

\subsubsection{Parity check matrix and parity of patterns}\label{sec:patterns_parity}

In this section we give a combinatorial interpretation of the syndrome of an error pattern. 
Consider the code $RM(m,m-r-1)$. Its parity check matrix is $H(m,m-r-1)=E(m,r)$.

Let $U\subseteq \F_2^m$ be a set of size $s$. We associate with $U$ the error pattern $\1_U \in \F_2^n$. Clearly 
$w(\1_U)=|U|=s$. We denote with $u_j$ the $j$'th element of $U$. We shall also think of $U$ as an $m\times s$ 
matrix whose $j$'th column is $u_j$. As before we denote with $U^r$ the submatrix of $E(m,r)$ whose columns are 
indexed by $U$. Alternatively, this is the set of all evaluation vectors of $U$'s columns. 
We shall use the same convention for another subset $V \subseteq \F_2^m$.

The following definition captures a combinatorial property that we will later show its relation to syndromes under $E(m,r)$.

\begin{defin}
For two matrices $A,B$ of same dimension $n_1\times n_2$, we denote $A \sim_r B$ if any pattern of size at most $r$ in the columns of $A$ appears with the same parity in the columns of $B$. I.e., for every subset $I \subset [n_1]$ of size $r$ and every $z\in \F_2^r$ the number of columns  
in $A_{I,\cdot}$ that equal $z$ is equal, modulo $2$, to the number of columns in $B_{I,\cdot}$ that equal $z$.
\end{defin}

For example, the matrices
\begin{equation}\label{eq:matrix-patterns}
A = \left( \begin{array}{cccccc}
1 & 0 & 0 & 0 & 0&0\\
0 & 1 & 0 & 0 & 0&0\\
0 & 0 & 1 & 0 & 0&0\\
0 & 0 & 0 & 1 & 0&0\\
0 & 0 & 0 & 0 & 1&0\\
0 & 0 & 0 & 0 & 0&1
\end{array} \right)
\quad \text { and } \quad
B = \left( \begin{array}{cccccc}
1 & 1 & 1 & 0 & 0& 0\\
1 & 1 & 0 & 1 & 0& 0\\
1 & 1 & 0 & 0 & 1& 0\\
1 & 1 & 0 & 0 & 0& 1\\
1 & 0 & 1 & 1 & 1& 1\\
0 & 1 & 1 & 1 & 1& 1
 \end{array} \right)
\end{equation}
satisfy $A\sim_2 B$ but $A\not\sim_3 B$.
To see that $A\sim_2 B$ one can observe that: the number of $1$'s in row $i$ in $A$ is equal (modulo $2$) to that number in $B$; the inner product (modulo $2$) between rows $i$ and $j$ in $A$ is the same as in $B$. Indeed, that inner product between rows $i$ and $j$ counts the number of columns that have $1$ in both rows. Together with the information about the number of $1$'s in row $i$ and in row $j$ we are guaranteed that any pattern on rows $i$ and $j$ has the same parity in both matrices.
On the other hand, the pattern $(1,1,*,*,1,*)$, which stands for $1$ in the first, second and fifth rows (in the terminology of the definition, $I=\{1,2,5\}$ and $z=(1,1,1)$), appears once in $B$ but it does not appear in $A$.

%Indeed, the pattern $(1,*,1)$, i.e. $1$ in the first row and $1$ in the third row appears once in $B$ (in the second column) but it doesn’t appear in $A$. On the other hand, the parity of any pattern of size $1$ is the same in $A$ and $B$. 
%the parity of the pattern $(1,0)$ in the first two rows is $0$ as it appears in two columns. The parity of the pattern $(1,*,1)$, i.e., $1$ in the first row and $1$ in the third row, is $1$ as it only appears in the second column.

The next lemma shows that two error patterns $\1_V$ and $\1_U$ have the same syndrome under $E(m,r)$ if and only if the two matrices $U$ and $V$ satisfy $U \sim_r V$. We denote with $ \mathbb{M}(m,r)$ the set of all $m$-variate monomials of degree at most $r$.

\begin{lemma}[Parity of patterns]\label{lem:pattern_parity}
For two sets $U,V \subseteq \F_2^m$ of size $s$ it hold that 
\begin{align}
E(m,r) \cdot \1_U = E(m,r)\cdot  \1_V \quad &\Longleftrightarrow \quad \sum_{i=1}^s f(u_i) = \sum_{i=1}^s f(v_i), \quad \forall 
f \in \mathbb{M}(m,r) \label{eq:monomial-pattern} \\
&\Longleftrightarrow \quad \sum_{i=1}^s f(u_i) = \sum_{i=1}^s f(v_i), \quad \forall 
f \in \mathbb{P}(m,r), \label{eq:polynomial-pattern}\\
&\Longleftrightarrow \quad U \sim_r V. \label{eq:same-pattern}
\end{align}
\end{lemma}
\begin{proof}
The first equivalence is by definition and the second one is clear. The equivalence between \eqref{eq:polynomial-pattern} and \eqref{eq:same-pattern} is best explained by the following example. Consider the polynomial $f(x_1,\dots,x_m)=x_1(1+x_m)$. In order for the equivalence to hold, it must be that the number of $u_i$'s (which are $m$-bit vectors themselves) that have $1$ in the first coordinate and $0$ in the last coordinate, is equal (modulo $2$) to the number of $v_i$'s that have the same structure. In other word, the equation 
$\sum_{i=1}^s f(u_i) = \sum_{i=1}^s f(v_i)$ makes sure that the number of columns that have $1$ in the first row and $0$ in the last row is the same (modulo $2$) in $U$ and in $V$.
By choosing different polynomials of degree at most $r$, the same must hold for any pattern of size at most $r$, and hence, by definition, $U \sim_r V$. The reverse direction is proved in a similar way. We note that the formal proof follows by induction on $r$. We leave the exact details of the proof to the reader.
\end{proof}

Our goal is to understand, for given values of $m$ and $r$, how many vectors $\1_U$ are bad, in the sense that they admit a bad companion $\1_V$ such that $E(m,r) \cdot \1_U = E(m,r)\cdot  \1_V$. Thus, by Lemma~\ref{lem:pattern_parity} this is equivalent to studying pairs $U,V$ such that  $U \sim_r V$. The next lemma will allow us to apply linear transformations to $U$ in order to make it ``nicer'' without losing generality.

%We will make a small abuse of notation and use $u \sim_d y$ as well as $u \sim_d v$ in addition to $U \sim_d V$. 
\begin{lemma}[Affine invariance]\label{lem:affine inv}
In the notation of Lemma~\ref{lem:pattern_parity}, if $U \sim_r V$ then $(AU) \sim_r (AV)$ for any linear transformation $A:\F_2^m\to\F_2^m$. Furthermore, if $A$ is invertible then $(AU) \sim_r (AV)$ implies that $U \sim_r V$.
\end{lemma}
% It is not true that we can shift by B. E.g. consider U=(01,10), V=(10,01) and B = (10,00). U+B =(11,10) and V+B=(00,01).

\begin{proof}
By Lemma~\ref{lem:pattern_parity}, $AU \sim_r AV$ iff $\sum_{i=1}^s f(u_i) = \sum_{i=1}^s f(v_i)$, $\forall f \in \mathbb{P}(m,r)$. For a polynomial $f$ denote $f_A(x) \triangleq f(Ax)$. It is clear that $\deg(f_A)\leq \deg(f)$. It thus follows that 
\begin{eqnarray*}
U \sim_r V &\Longleftrightarrow& \quad \sum_{i=1}^s f(u_i) = \sum_{i=1}^s f(v_i)\quad \forall 
f \in \mathbb{P}(m,r),\\
&\Longrightarrow& \quad \sum_{i=1}^s f_A(u_i) = \sum_{i=1}^s f_A(v_i), \quad \forall 
f \in \mathbb{P}(m,r),\\
&\Longleftrightarrow& \quad \sum_{i=1}^s f(Au_i) = \sum_{i=1}^s f(Av_i)\quad \forall 
f \in \mathbb{P}(m,r),\\
&\Longleftrightarrow& \quad AU \sim_r AV.
\end{eqnarray*}
To see the furthermore part, we note that if $A$ is invertible then $f_{A^{-1}}(AV) = f(V)$. Hence, all the implications above can be made ``if and only if''.
\end{proof}

\subsubsection{The case $r=1$} \label{sec:deg_3}

To prove Theorem~\ref{thm:main_for_BSC} we first study the case where $r=1$ as the proof for the general case will use ideas similar to the proof of this case. Note that this case corresponds to studying syndrome of error patterns under $H(m,m-(2r+2))=H(m,m-4)=E(m,3)$, namely evaluations by all degree-3 monomials.

We will prove first a deterministic result:  if $U$ is {\em any} set of  linearly independent columns,  then for any $V\neq U$, we have that  $U\not\sim_3 V$. Thus, any set of errors that is supported on linearly independent coordinates (when viewed as vectors in $\F_2^m$) can be uniquely corrected. This immediately gives an average-case result. If we have $m-\log(m/\e)$ random errors, then with probability at least $1-\e$ their locations correspond to linearly independent $m$-bit vectors and therefore we can correct such amount of errors with high probability.\footnote{To eliminate possible confusion we repeat: an error pattern is an $n$-bit vector, whose coordinates are indexed by $m$-bit vectors.} Notice that this is already highly nontrivial, as $R(m,m-4)$ has (absolute) distance $16$, so in the worst case one cannot correct more than 8 worst-case errors!

\begin{lemma}\label{lem:deg_3}
Let $U\subseteq\F_2^m$ be a set of linearly independent vectors, such that $|U|=s$. 
Then, for any $V\neq U$, such that $|V|\leq s$, we have that $V\not\sim_3 U$.
\end{lemma}

In particular this means that we can correct the error pattern $\1_U$ in $RM(m,m-4)$. 

\begin{proof}
By multiplying $U$ with an invertible $A$ (changing the basis $\F_2^m$) we can assume, w.l.o.g., that the columns of $U$ are the elementary basis vectors, $e_1,\ldots,e_s$.\footnote{This is not really necessary, but it makes the argument simpler to explain.} Indeed, since  $A$ is invertible it follows from  Lemma~\ref{lem:affine inv}  that it is enough to prove the claim for $AU$. 

Let $V\subseteq\F_2^m$ be such that $|V|=s$ and  $V\sim_3 U$. Our task is to show that $V=U$. This will be shown in two steps. First, we'll show that span($V$) = span($U$), which in particular implies that $V$ is linearly independent as well. Proving linear independence requires only that $V\sim_2 U$, namely evaluations by degree-2 monomials. Using $V\sim_3 U$, we'll prove that they actually have the same span, and from that derive that $V=U$. 

Let us first argue linear independence of $V$. We'll think of $U$ and $V$ not only as sets of vectors, but also as $m\times s$ matrices, and denote by $U'$ the transpose of $U$. Note that, as the columns of $U$ are unit vectors, we have $U'U = I_s$. Now since diagonal elements of this product capture the value of degree-1 monomials of the syndrome, and off-diagonal elements of the of the product correspond to inner products of rows, namely (as in the example of the previous section), to degree-2 monomials of the syndrome. As $V \sim_2 U$ we also have that $V'V = I_s$ and so the dimension of $V$ is $s$ as well. We will later show~\ref{sec:counterex} that this linear independence is the only thing we can infer from $V\sim_2 U$.

We now actually prove the stronger statement that in fact $U$ and $V$ span the same subspace. This will require $V\sim_3 U$. It will be sufficient to prove that $V$ spans the vector $e_1$, as for other vectors in $U$ the proof is identical.
Consider the pattern $(1,0)$ in the first two rows of $U$. That is, consider all columns of $U$ that have $1$ in their first coordinate and $0$ in the second. It is clear that this pattern only appears in $e_1$ and hence its parity in $U$ is $1$. Thus, there must be an odd number of columns in $V$ whose first two rows equal $(1,0)$. The main observation is that if we add up the columns then we obtain the vector $e_1$. 

\begin{claim}\label{claim: V spans e1}
Under the conditions of the lemma, the sum of all columns in $V$ whose first two coordinates equal $(1,0)$ is $e_1$. More generally, for $i\in [s]$, if we consider the pattern that has $1$ in the $i$'th coordinate and $0$ in some $j\neq i$ coordinate, then the sum of all columns in $V$ that have this pattern is equal to $e_i$.
\end{claim}

\begin{proof}
Assume that this is not the case, namely, the sum is a vector $w\neq e_1$. We first note that the first two coordinates of $w$ equal $(1,0)$. Indeed, this holds as we summed an odd number of vectors that has these values. Hence, there must exist a coordinate $i>2$ such that $w_i=1$. Thus, the number of vectors in $V$ with the pattern (1,0,1) in rows (1,2,$i$) is odd. But this is not the case in $U$, contradicting $V\sim_3 U$.

The proof of the general case is similar. 
\end{proof}

%As an immediate corollary we obtain that  the columns of $V$ are linearly independent.

%\begin{claim}\label{claim:V independent}
%Under the conditions of the lemma, it must be that the columns in $V$ are linearly independent and $|V|=s$. 
%\end{claim}
%
%\begin{proof}
%Claim~\ref{claim: V spans e1} implies that $e_i$ is in the columns span of $V$ for all $i\in [s]$. As $|V|\leq s$ the claim follows.  
%\end{proof}
We now use the fact that $U$ and $V$ have the same span to conclude that $U=V$. 
Denote  $J_{10}\subset [s]$ the indices of columns in $V$ that have $(1,0)$ as their first two coordinates. By Claim~\ref{claim: V spans e1} we have that $\sum_{i\in J_{10}}v_i = e_1$.
Next, consider the pattern $(1,*,0)$, namely, the pattern that has $1$ in the first row and $0$ in the third row. Denote the corresponding set of column indices with $J_{1*0}$. Again, Claim~\ref{claim: V spans e1} implies that $\sum_{i\in J_{1*0}}v_i = e_1$. However, since the columns in $V$ are linearly independent, there is only one way to represent $e_1$ as a linear combination of the columns of $V$ and therefore it must be the case that  $J_{10}=J_{1*0}$.

Continuing in this fashion we get that $J_{10} = J_{1*0} = J_{1**0}=\ldots = J_{1*\ldots*0}$, where the last set corresponds to the columns that have $1$ in the first coordinate and $0$ in the last coordinate. As the size of $J_{10}$ is odd we know that it is not empty. In particular, all the vectors in $J_{10}$ must satisfy that they have $1$ in the first coordinate and $0$ in the remaining coordinates. Indeed each such $0$ can be justified by one of those $J$ sets. In particular $e_1$ is a column in $V$. Repeating this process for all $e_i$, $i\in [s]$, we get that all these $e_i$'s are columns in $V$. Since $V$ has exactly $s$ columns it must have the same set of  columns as $U$. In particular, $V=U$. This completes the proof of Lemma~\ref{lem:deg_3}.
\end{proof}

Lemma~\ref{lem:deg_3} shows that if the columns of  $U$ are linearly independent then $E(m,3)\1_U \neq E(m,3)\1_V$ for any other $V$ of the same size. As randomly picking $m-\log(m/\e)$ vectors in $\F_2^m$ we are likely to get linearly independent vectors we obtain our main result for the case $r=1$.

\begin{claim}\label{claim: random linear vectors}
Let $\e>0$ and $t=m-\log(m/\e)$. Pick $t$ vectors uniformly at random from $\F_2^m$, $u_1,\ldots,u_t\in \F_2^m$. Then, with probability at least $1-\e$, the $u_i$ are linearly independent.  
\end{claim}

\begin{proof}
Set $u_0=\vec{0}$. 
For each $1\leq i\leq t$, the probability that $u_i$ belongs to the span of $u_{0},\ldots,u_{i-1}$ is at most $2^{i}/2^m \leq \frac{\e}{m}$. The claim now follows from the union bound.
\end{proof}

Combining Lemma~\ref{lem:deg_3} with Claim~\ref{claim: random linear vectors} we obtain the following corollary which is a special case of our main theorem.

\begin{corol}\label{cor: RM deg 3}
We can correct a random set of $m-\log(m/\e)$ errors in $RM(m,m-4)$ with high probability. 
\end{corol}

\subsubsection{The degree-$r$ case}\label{sec:deg-r}

The proof of the degree-$r$ case proceeds along the same lines as the proof of the $r=1$ case. However, in order to correct errors for $RM(m,m-(2r+2))$ we will require that the matrix $U^r$ corresponding to the error pattern $\1_U$ has linearly independent columns. Note that when $r=1$ this amounts to requiring that $U$ has linearly independent columns, in order to correct $\1_U$ in $RM(m,m-4)$, which is exactly what we proved in Section~\ref{sec:deg_3}. This is also a shortcoming of our result, it is clear that with this condition we cannot expect to correct more than ${m\choose \leq r}$ errors. On the one hand this is much better than the minimum-distnace based result of $2^{2r+2}$, but on the other hand one may hope to be able to decode from $O\left({m  \choose \leq 2r+1}\right)$ many errors. 

Our main lemma is an analog of Lemma~\ref{lem:deg_3}. 

\begin{lemma}\label{lem:deg_r_errasures_to_errors}
Let $U\subseteq \F_2^m$ be such that $|U|=s$ and the columns of $U^r$ are linearly independent. Then, for any $V\subseteq \F_2^m$ satisfying $|V|=s'\leq s$ and  $V\neq U$, we have that $V\not\sim_{2r+1} U$.
\end{lemma}

Notice that Lemma~\ref{lem:deg_3} is obtained by setting $r=1$ in Lemma~\ref{lem:deg_r_errasures_to_errors}.

\begin{proof}

Denote with $u_i$ the $i$'th column of $U$. Recall that the $i$'th column of $U^r$ corresponds to all evaluations of monomials of degree at most $r$ at the point $u_i$, i.e., it is equal to $u_i^{{r}}$. This interpretation will be helpful throughout the proof. 
Assume that $V\subseteq\F_2^m)$ is such that $|V|=s' \leq s$ and $V\sim_{2r+1} U$. Similarly, we denote the columns of $V$ with $v_1,\ldots,v_{s'}\in \F_2^m$ and note that the $i$'th column of $V^r$ is $v_i^{{r}}$. 

As the columns of $U^r$ are linearly independent, there exist vectors $f_i$ so that $f_i\cdot U^r = e_i\in \F_2^s$. As the coordinates of each $f_i$ are indexed by monomials of degree $\leq r$, we can interpret $f_i$ as a degree-$r$ polynomial $f_i(x_1,\ldots,x_m)$ and rewrite $f_i\cdot U^r$ as $f_i\cdot U^r = (f_i(u_1),\ldots,f_i(u_s))$. In other words, $f_i(u_j) = \delta_{i,j}$.\footnote{Intuitively, the polynomials $f_i$ correspond to the rows of the matrix $A$ that were used in the proof of Lemma~\ref{lem:deg_3} to make the columns of $U$ equal the elementary unit vectors there.}

Our next goal is proving that if $U\sim_{2r+1} V$ then $u_1\in V$. This will clearly imply the lemma as we can prove the same for any other $u_i$. Our main handle will be the polynomial $f_1$ that separates $u_1$ from the other $u_i$'s. 
%{\bf EA: Is it $U_i$'s ?}

Let us assume wlog that $(u_1)_1=1$, i.e., that the first coordinate of $u_1$ equals $1$.\footnote{If it equals $0$ then we consider the polynomial $(1+x_1)f_1$ in what follows.}
Consider the polynomial $x_1\cdot f_1$. This is a polynomial of degree $\leq r+1$ and by the definition of $f_1$ and the assumption on $(u_1)_1$ we have that $\sum_{i=1}^{s}(x_1f_1)(u_i)=1$. 
As $U\sim_{2r+1} V$, it must hold that $\sum_{i=1}^{s'}(x_1f_1)(v_i)=1 \mod 2$.  
%{\bf EA: This is a degree ``at most" $r+1$.. and the sums are from $i=1$ until $s$?}

Following the footsteps of the proof of Lemma~\ref{lem:deg_3}, we denote $J_{f_1,1} = \{i \mid  (x_1f_1)(v_i)=1\}$. The next claim is analogous to Claim~\ref{claim: V spans e1}.

\begin{claim}\label{cal: V^r span U^r}
%the first $m$ co-ordinates of the vector 
$$\sum_{i\in J_{f_1,1}}v_i^{{r}} = u_1^{{r}}.$$ 
%equal $U_1$. 
\end{claim}

\begin{proof}
Let $M$ be some monomial of degree $\leq r$. To ease the notation assume that $M(u_1)=1$ and consider the polynomial $M\cdot x_1\cdot f_1$.\footnote{If $M(U_1)=0$ then we  consider the polynomial $(1+M)\cdot x_1\cdot f_1$ instead.} It is clear that $(M\cdot x_1\cdot f_1)(u_1)=1$. Since $V\sim_{2r+1} U$ and $\deg(M\cdot x_1\cdot f_1)\leq 2r+1$, it follows that $\sum_{i=1}^{s'}(Mx_1f_1)(v_i)=1$. From definition of $J_{f_1,1}$ we have that $$\sum_{i=1}^{s'}(Mx_1f_1)(v_i)=\sum_{i\in J_{f_1,1}}(Mx_1f_1)(v_i)=1.$$ Indeed, for every $i\not\in  J_{f_1,1}$ we have that $(x_1f_1)(v_i)=0$. 
We thus conclude that there is an odd number of vectors $v_i$, $i\in J_{f_1,1}$, such that $(Mx_1f_1)(v_i)=1$. In particular, the $M$'th coordinate in the sum $\sum_{i\in J_{f_1,1}}v_i^{{r}}$ equals $1$, i.e. it is equal to $M(u_1)$. As $M$ was arbitrary we conclude that  $\sum_{i\in J_{f_1,1}}v_i^{{r}} = u_1^{{r}}$, as required. 
\end{proof}

As we can prove an analogous lemma for every $u_i$, we conclude the the columns of $U^r$ belong to the span of the columns of $V^r$. In particular, the columns of $V^r$ are linearly independent and $|V|=s$ (earlier we called this observation Claim~\ref{claim:V independent}). Our next step is proving that, up to permutation of columns, 
$U^r=V^r$. This will imply that $u_i=v_i$ as we wanted. 

To show this, for every $\ell \in [m]$,  we denote $$J_{f_1,\ell} = \{i\mid (v_i)_\ell = (u_1)_\ell  \text{ and } f_1(v_i)=1\}.$$
We note that there is an alternative way to define $J_{f_1,\ell}$ by considering either the polynomial $x_\ell \cdot f_1$ or the polynomial $(1+x_\ell)\cdot f_1$. We thus have that for every $\ell\in [m]$,
$$\sum_{i\in J_{f_1,\ell}}v_i^{{r}} = u_1^{{r}}.$$
However, as the columns of $V^r$ are linearly independent and 
$$\sum_{i\in J_{f_1,\ell}}v_i^{{r}} = u_1^{{r}} = \sum_{i\in J_{f_1,1}}v_i^{{r}}$$
we get that $J_{f_1,1} = J_{f_1,2} = \ldots = J_{f_1,m}$.
Hence, $$J_{f_1,1} = \cap_{\ell=1}^{m}J_{f_1,\ell} =  \{i\mid \forall \ell \in [m] \; (v_i)_\ell = (u_1)_\ell  \text{ and } f_1(v_i)=1\}.$$
Thus, for every $i\in J_{f_1,1}$ we have that $v_i = u_1$. In particular, since $J_{f_1,1}\neq \emptyset$, it follows that there is some $i\in [s]$ such that $v_i = u_1$. As we can prove this for every $u_j$, we conclude that $U=V$ as claimed. This concludes the proof of Lemma~\ref{lem:deg_r_errasures_to_errors}.
\end{proof}

We thus proved that if an error pattern $\1_U$ is such that its coordinates $u_i$ satisfy that $u_i^{{r}}$ are linearly independent, then we can correct that error pattern in $RM(m,m-(2r+2))$. We summarise this in the following theorem.

\begin{thm}\label{thm:erasures_to_errors}
If a set of columns $U$ are linearly independent in $E(m,r)$ (namely, $RM(m,m-r-1)$ can correct the {\em erasure} pattern $\1_U$), then the {\em error} pattern $\1_U$ can be corrected  in $RM(m,m-(2r+2))$.
\end{thm}

\begin{proof}
Lemma~\ref{lem:deg_r_errasures_to_errors} tells us that if the columns indexed by $U$ are linearly independent in $E(m,r)$, then there is no other $V\subseteq \F_2^m$ of size $\leq s$ such that  $V\sim_{2r+1} U$. Lemma~\ref{lem:pattern_parity} now implies that $$E(m,2r+1)\cdot \1_U \neq E(m,2r+1)\cdot \1_V,$$ for any $U\neq V\subseteq \F_2^m$ of size $\leq s$.
As $E(m,2r+1) = H(m,m-(2r+2))$, it follows that the syndrome of $\1_U$ is unique and hence $\1_U$ is uniquely decodable in $RM(m,m-(2r+2)$.\end{proof}

 The proof of Theorem \ref{thm:main_for_BSC} immediately follows  from Theorems~\ref{thm:lin_ind_eval_vectors} and \ref{thm:erasures_to_errors}.
%Lemmas~\ref{lem:pattern_parity}
% and \ref{lem:deg_r_errasures_to_errors}.

\begin{proof}[Proof of Theorem \ref{thm:main_for_BSC}]
Theorem~\ref{thm:lin_ind_eval_vectors} guarantees that a set $U\subseteq \F_2^m$ of  $s=\lfloor {m- \log({m\choose \leq r})+ \log(\e) \choose \leq r}\rfloor -1$ randomly chosen vectors, satisfy that the columns of $U^r$ are linearly independent. By Lemma~\ref{lem:deg_r_errasures_to_errors} we learn that there is not other $V\subseteq \F_2^m$ of size $\leq s$ such that  $V\sim_{2r+1} U$. Lemma~\ref{lem:pattern_parity} implies that for any such $V$, $$E(m,2r+1)\cdot \1_U \neq E(m,2r+1)\cdot \1_V.$$
As $E(m,2r+1) = H(m,m-(2r+2))$, it follows that the syndrome of $\1_U$ is unique and hence $\1_U$ is uniquely decodable in $RM(m,m-(2r+2)$.
\end{proof}

%As a corollary of the above we get the following interesting theorem.

%\Anote{Changed title}
\subsubsection{A general reduction from decoding from errors to decoding from erasures} \label{sec:general_deg_3}

In this section we show that the results proved in Section~\ref{sec:deg_3} are in fact more general and apply to any degree three tensoring of a linear code with itself. We first set up the required definitions.

\begin{defin}\label{def: hadamard product}
The Hadamard product of two vectors $y,z\in \F_2^n$ is the vector $w=y \circ z$ obtained from the coordinate wise product $w_i = y_i\cdot z_i$.
\end{defin}

\begin{defin}\label{def: tensoring}
Let $H$ be a $k\times n$ matrix. For every natural number $\ell$, $H^{\otimes \ell}$ is a ${k \choose \leq \ell} \times n$ matrix that is defined as follows. Rows of $H^{\otimes \ell}$ are indexed by tuples $i_1 < i_2<\ldots <i_j$, $1\leq j\leq \ell$, where the corresponding row in $H^{\otimes \ell}$ is equal to the Hadamard product of rows $i_1, i_2,\ldots,i_j$. 
\end{defin}

In other words, if we think of $H$ as the set of its column vectors then, using our usual notation, $H^{\otimes \ell}=H^\ell$.
In particular,  the parity check matrix $H(m,m-r-1)$ of the code $RM(m,m-r-1)$ is equal to $H(m,1)^{{r}}$. Indeed the row indexed by $i_1 < i_2<\ldots <i_j$ corresponds to the evaluations of the monomial $\prod_{t=1}^{j}x_{i_t}$. 

It is also clear that for any integers $m,n$ and any $m\times n$ matrix $H$, the set of columns of $H$ is contained in the set of columns of  the Hadamard matrix of rank $m$, i.e., $E(m,1)$, namely, every column of $H$ appears in $E(m,1)$. We thus obtain the following two corollaries that follow from the proof technique of the previous section.\footnote{The proofs are completely identical and are thus omitted.}

\begin{corol}\label{cor:inde columns general}
Let $m,n$ be integers and $H$ an $m\times n$ matrix. Let $S\subseteq [n]$ be such that the columns indexed by $S$ in $H$ are linearly independent. Then, in the code whose parity check matrix is $H^{\otimes 3}$, we can correct the error pattern $S$.
\end{corol}

Using the relation between correcting erasures and independence (Lemma~\ref{equiv}) in the parity check matrix we obtain the following corollary.

\begin{thm}\label{thm:general_erasures_tensored_to_errors}
Let $C\subseteq\F_2^n$ be a linear code with parity check matrix $H$. For any subset $S\subseteq [n]$ the following hold: If we can recover codewords in $C$ from erasures in the coordinates $S$ then in the code whose parity check matrix is $H^{\otimes 3}$, we can correct the error pattern $\mathbb{1}_S$.
\end{thm}

We note that Corollary~\ref{cor: RM deg 3} is a special case of Theorem~\ref{thm:general_erasures_tensored_to_errors}.

\subsubsection{The degree-$2$ counterexample}\label{sec:counterex}

In this section we prove that $RM(m,m-3)$ does not achieve capacity for the BSC. 
%\Amnote{Please check that the degree is correct}
For this Reed-Muller code to achieve capacity we must have that, w.h.p., a random error pattern of weight $O(m)$ has a unique syndrome. We next show that very few patterns of this weight have unique syndromes. In fact, we show that very few patterns of weight $\sqrt{m}$ have unique syndromes.

Let $s<\sqrt{m}$ be an even integer. Let $B$ be the following $s\times s$ matrix. For $1\leq i\leq s-2$, the $i$'th row of $B$ has $1$ in coordinates $1,2$ and $i+2$. E.g., the first tow of $B$ begins with three $1$'s followed by zeros. The $s-1$'th row of $B$ equals $(1,0,1,\ldots,1)$ and the last row of $B$ is $(0,1,\ldots,1)$. Note that $B\cdot B^t = I$. For example, the matrix $B$ in \eqref{eq:matrix-patterns} is what we get if we set $s=6$.

Let $u_1,\ldots,u_s$ be any set of vectors in $\F_2^m$. Let $U$ be the matrix whose $i$'th column is $u_i$. Define $V  = U \cdot B$. %, and let $v_i$ be the $i$'th column of $V$. 
It is not hard to verify that $E(m,2)\cdot \1_U = E(m,2)\cdot \1_V$. Thus, if $V\neq U$ then $\1_U$ does not have a unique syndrome. 

Finally, we note that picking $u_1,\ldots,u_s$ at random is equivalent to picking the matrix $U$ at random. If  $U$ and $V$ have the same set of columns then there must exist an $s\times s$ permutation matrix $\Pi$ such that $U\Pi=V$. Thus, $U(B-\Pi)=0$. Fix $\Pi$. The probability that all rows of $U$ are in the kernel of $B-\Pi$ is at most $2^{-m}$. Indeed, for every permutation matrix $\Pi$, $\rank(B-\Pi)\geq 1$. As there are $s!$ permutation matrices, the probability that $U$ is unique is at most $s!/2^m < 2^{-m/2}$.

%%%%%%%%%%%%%%%%%%%%%%%%%%%%%%%%%%
%%%%%%%%%%%%%%%%%%%%%%%%%%%%%%%%%%
%%%%%%%%%%%%%%%%%%%%%%%%%%%%%%%%%%
%%%%%%%%%%%%%%%%%%%%%%%%%%%%%%%%%%

%\Amnote{achieve capacity for all parameters, questions about evaluation vectors, decoding algorithms, extending weight enumerator}

\section{Future directions and open problems}\label{sec:open}

We believe that our work renews hope for progress on some classical questions, and suggests some new concrete directions and open problems.

The most obvious of all is the question of whether Reed-Muller codes achieve capacity for all ranges of parameters, either for random erasures or for random errors. We only handle here the extreme cases of very high or very low rates, whereas most interest is traditionally focused on constant rate codes. We believe that the techniques for each of our four bounds can be improved to a larger set of parameters (see below), but feel that they fall short of reaching constant rate, and possibly new techniques are needed.

One way to improve our bounds in both Theorem~\ref{thm:intro:low_deg_BEC} (low-rate BEC) and Theorem~\ref{thm:intro-low-BSC} (low rate BSC) is through tighter bounds on the weight enumeration of Reed-Muller codes, as well as tighter bounds on the probability of error for Theorem~\ref{thm:intro-low-BSC}.
We believe that %the right bound 
in Theorem~\ref{thm:intro:wt-dist} one can eliminate the factor $\ell^4$ in the exponent, resulting in a bound that is  a fixed polynomial (independent of $m,r,\ell$) of the lower bound in \cite{KLP}. While such a tight result would not get us (in either Theorem~\ref{thm:intro:low_deg_BEC} and \ref{thm:intro-low-BSC}) to the constant rate regime, this question of weight enumeration is of course basic in its own right. Moreover, both in ~\cite{KLP} and our paper, it also implies similar bounds for list-decoding, which is another basic question.

Theorem~\ref{thm:intro:RM_for_BEC_high_degree} (high rate BEC) is quantitatively much weaker than Theorems~\ref{thm:intro:low_deg_BEC} and \ref{thm:intro-low-BSC}, in that the latter two can handle polynomials of degree-$r$ which is linear in $m$, whereas the former only reaches degrees $r$ which are about $\sqrt{m}$. The bottleneck in the argument, which probably prevents it from reaching a linear degree, is the use of the union bound. We upper bound the probability that, when adding a subsequent random vector $u$ to our set $U$, its evaluation $u^r$ will be linearly independent of the evaluations of all previously chosen points. This current proof does not use at all that previous points were chosen randomly, as we don't know how to take advantage of this. 

For high-rate BSC (Theorem~\ref{thm:intro:main_for_BSC}), while we are able to correct many more errors than previously known, we are not even able to achieve capacity. Here we feel  that one important bottleneck is our inability to argue directly about corruption patters (sets $U$) which are linearly {\em dependent}. Our unique decoding proof, even for $r=1$ (on which we focus now), showing that a set $U \in \F_2^m$ is uniquely determined by its syndrome under evaluations by degree-$3$ monomials i.e., by $E(m,3)\cdot \1_U$, is especially tailored to linearly independent sets $U$. The gap between our lower bound (namely that  $E(m,2)\cdot \1_U$ does not suffice) and the above upper bound (that  $E(m,3)\cdot \1_U$ suffices) is intriguing, and we believe we can find a subset of quadratically many monomials of degree at most $3$ which guarantee unique decoding - such a result is information theoretically optimal; number of error patterns $U$ which are linearly independent is about $\exp(m^2)$, and thus $O(m^{2})$ bits are needed in any unique encoding.

Another burning question regarding this result is its inefficiency. While unique decoding is guaranteed, the best way we know to identify the set $U$ is brute force, requiring $\exp(m^2)$ steps for independent sets $U$ of size $m$.  We feel that a good starting place is (perhaps using our uniqueness proof) which recovers $U$ in $\exp(m) = \poly(n)$ steps from its evaluation on all degree-3 monomials (or even degree-10 monomials). Of course, it is quite possible that a $\poly(m)$ algorithm exists. In particular, recursive algorithms (that exploit the recursive nature of RM codes) could be used to that effect\footnote{Practically, one can decode RM codes on the BSC by using a recursive decoder (e.g., like for polar codes) for each missing component in the syndrome, and by growing a list of possible codewords each time the decoder has doubts, or pruning down the tree each time the decoder can check the validity of a path (from the available components of the syndrome).}. 

Yet another research direction related to our result is of course exploring the connections between recovering from erasures and from errors. Our general reduction between the two uses tensor powers and hence loses in efficiency (which here is best captured by the co-dimension of the code, which is cubed). Is there a reduction which looses less? We do not know how to rule out a reduction that increases the co-dimension only by a constant factor. There is no reason to restrict attention to Reed-Muller codes and our tensor construction - such a result would be of use anywhere, as erasures are so much simpler to handle than errors.

Finally, we believe that a better understanding of the relation between Reed-Muller codes and Polar codes is needed, and perhaps more generally an understanding of which subspaces of polynomials generated by subsets of monomials give rise to good, efficient codes. In particular, it would also be interesting to investigate the scaling of the blocklengh in terms of the gap to capacity for RM codes. It was proved recently in \cite{guru-polar} that for polar codes, the blocklength scales polynomially with the inverse of the gap to capacity, with a precise characterization given in \cite{hamed_thesis}. While this scaling does not match the optimal scaling of random codes \cite{strassen}, it is in contrast to the exponential scaling obtained with concatenated codes \cite{forney-thesis} (see \cite{guru-polar} for a discussion on this). It would be interesting to investigate such finer questions for RM codes in view of the results obtained in this paper, which already provide partial information about these scalings. 

\section*{Acknowledgements}
We thank Venkatesan Guruswami for bringing \cite{Wei} to our attention. 
The second author would like to thank the organizers of Dagstuhl meeting ``Algebra in Computational Complexity,'' where he discussed the results with Venkatesan Guruswami.
%Venkat corresponded with us during a Dagstuhl meeting, where we first presented our result. 
%
%

\bibliographystyle{amsalpha}
\bibliography{RMcode}

\appendix

\section{Proofs of Claim~\ref{cla:binomial-difference} and Claim \ref{cla: estimation r small}}\label{app:missing-proofs}

\begin{proof}[Proof of Claim~\ref{cla:binomial-difference}]

We first need the following estimate.

\begin{claim}\label{cla:sum_of_binomials}
For integers $0\leq a\leq c$ and $b\leq c-a$ we have that $\sum_{i=1}^{a}{c-i\choose b} =  {c \choose b+1} -  {c-a\choose b+1}$.
\end{claim}

\begin{proof}%[Proof of Claim~\ref{cla:sum_of_binomials}]
\begin{eqnarray*}
\sum_{i=1}^{a}{c-i \choose b} &=& \sum_{i=1}^{a}{c-i\choose b} + {c-a \choose b+1} -  {c-a\choose b+1}\\ &=& \sum_{i=1}^{a-1}{c-i\choose b} + {c-a+1 \choose b+1} -  {c-a\choose b+1} \\ &=& \ldots \\ &=& {c \choose b+1} -  {c-a\choose b+1}.
\end{eqnarray*}
\end{proof}

\noindent Using the claim we get that 
%Using Claim~\ref{cla:sum_of_binomials} we get
\begin{eqnarray*}
{m\choose \leq r} -  \sum_{i=1}^{t}{m-i \choose \leq r-1} &=&{m\choose \leq r} - \sum_{i=1}^{t}\sum_{j=0}^{r-1}{m-i \choose j} \\
&=&{m\choose \leq r} - \sum_{j=0}^{r-1}\sum_{i=1}^{t}{m-i \choose j} \\
&=& {m\choose \leq r} -  \sum_{j=0}^{r-1}\left( {m \choose j+1} - {m-t \choose j+1} \right)\\
&=& {m-t \choose \leq r}.
\end{eqnarray*}

\end{proof}

%\Amnote{Make sure what the exact statement is and the read the proof}

\begin{proof}[Proof of Claim \ref{cla: estimation r small}]
We shall need the following two simple inequalities that hold for every $C>A/4>B$: 
$$\frac{{A \choose B}}{{C \choose B}} > \left(\frac{A-B}{C}\right)^B \quad\quad \text{and}\quad\quad {A \choose \leq B} \leq 
%\left(\frac{eA}{B}\right)^B.$$
2A^B.$$
Thus, for any $0\leq r'\leq r$,
\begin{eqnarray*}
\frac{{m-3 \log({m\choose \leq r})+ \log(\e) \choose r'}}{{m \choose r'}} &>& \left( \frac{m-3 \log({m\choose \leq r})+ \log(\e)-r'}{m}   \right)^{r'} \\ &=& \left(1 - \frac{3 \log({m\choose \leq r})- \log(\e)+r'}{m} \right)^{r'}\\&\geq &
\left(1 - \frac{3 r\log(m)+3- \log(\e)+r'}{m} \right)^{r'} \\ &\geq & \left(1 - \frac{4 r\log(m)}{m} \right)^{r'} 
\\& \geq & \left(1 - \frac{4 r'\cdot r\log(m)}{m} \right) \\ & \geq & (1-\delta).
\end{eqnarray*}
Hence, 
$${{m-3 \log({m\choose \leq r})+ \log(\e) \choose \leq r}} = \sum_{r'=0}^{r} {{m-3 \log({m\choose \leq r})+ \log(\e) \choose r'}} > \sum_{r'=0}^{r}(1-\delta){{m \choose r'}} = (1-\delta){{m \choose \leq r}},$$
as claimed.
\end{proof}

%%%%%%%%%%%%%%%%%%%%

\section{A proof of Lemma~\ref{lem:number_of_ind_polynomials} using hashing}\label{app:hash-proof}

%\Amnote{I inserted back the old proof using hashing}

In this section we prove the following slightly weaker version of Lemma~\ref{lem:number_of_ind_polynomials}.

\begin{lemma}\label{lem:number_of_ind_polynomials-weak}
Let $\V\subseteq \F_2^m$ such that $|\V| > 2^{m-t}$. Then there are more than ${m-t - 2\lceil\log({m-t \choose \leq r})\rceil\choose \leq r}$ linearly independent polynomials of degree $\leq r$ that are defined on $\V$.
\end{lemma}

Notice that the only difference between Lemma~\ref{lem:number_of_ind_polynomials} and Lemma~\ref{lem:number_of_ind_polynomials-weak} is the constant $2$ in the lower bound.

The main idea in the proof is showing that there exists a linear transformation $T$ such that the projection of the set $T(\V)$ onto (roughly) the first $\log(|\V|)$ coordinates contains a ball of radius $r$ around some point. 
Since restricting monomials, of degree $\leq r$, to a ball of radius $r$ yields linearly independent functions, the claim follows. To prove that a random transformation has a large projection onto the first coordinates we use the {\em leftover hash lemma} of Impagliazzo et al. \cite{ImpagliazzoLL89}. This is where we lose compared to Lemma~\ref{lem:number_of_ind_polynomials}. The lemma of \cite{ImpagliazzoLL89} gives more information than just a large projection (i.e., that the distribution on the projection is close to uniform)  and so it does not get the same parameters that we can get using the result of Wei  (Theorem~\ref{thm:GHW}).

\begin{proof}
%By the discussion preceding Lemma~\ref{lem:number_of_ind_polynomials} we have that $\dim(\mathbb{P}(m,r)/\I(U)=s$. I.e., the dimension of the space of degree $\leq r$ polynomials that are defined on $\V$ is $s$. 

%Next we show that if $|\V|$ is large then there are many linearly independent polynomials defined on $\V$. For this end we will need the following corollary of the leftover hash lemma of Impagliazzo et. al.\footnote{We obtain the corollary by picking the best seed in the Lemma of \cite{ImpagliazzoLL89}.} 

We start by proving that, after a suitable linear transformation, the projection of $\V$ onto the first coordinates contains a large ball.

%In the next  lemma we change notation and talk of a set $Y$ instead of $\V$ as we want to speak in more general terms. 

\begin{lemma}\label{lem:A_contains_ball}
Let $Y\subseteq\F_2^{m}$  a set of size $2^{a}$. Then, there is a linear transformation $T$ such that for some $z \in \F_2^{a - 2\lceil\log({a \choose \leq r})\rceil}$, the ball $B(z,r)$ is contained in the projection of $T(A)$ onto the first $a - 2\lceil\log({a \choose \leq r})\rceil$ coordinates. 
\end{lemma}

\begin{proof}
For the proof we need the following leftover hash lemma of Impagliazzo et al. \cite{ImpagliazzoLL89}.
\begin{lemma}[\cite{ImpagliazzoLL89}]\label{lem:LOHL}
Let $\ell\leq a \leq m$ be integers and $Y\subseteq\F_2^{m}$  a set of size $2^{a}$. Then, there exists an invertible $m\times m$ matrix $T$ such that the projection of the set $T(Y)$ onto the first $a-\ell$ coordinates yields a set of size larger than $2^{a-\ell}(1-2^{-\ell/2})$.
\end{lemma} 
Given Lemma~\ref{lem:LOHL} the proof of Lemma~\ref{lem:A_contains_ball} is by a simple averaging argument.
%The next lemma is a simple averaging argument that shows that if a set $Y$ is large then, after a suitably linear transformation, its projection on the first coordinates contains some ball of radius $r$. 
%The exact parameters are determined by Lemma~\ref{lemma: LOHL}.
For the proof we shall denote with $\pi_{b}(\cdot)$ the map that projects $m$-bit vectors on their first $b$ coordinates. 

Denote $\hat{a} = a - 2\lceil\log({a\choose \leq r})\rceil$.
Apply Lemma~\ref{lem:LOHL} with $\ell =a-\hat{a},a,m$ on the set $Y$. We get that, for a suitable linear transformation $T$, the projection $\pi_{ \hat{a}}(T(Y))$ is a set of size larger than $2^{\hat{a}}(1 - 2^{-\lceil \log({a \choose \leq r})\rceil}) \geq 2^{\hat{a}}(1-\frac{1}{{a \choose \leq r}})$. 

By linearity of expectation, there is a point $z \in \F_2^{\hat{a}}$ so that the fraction of points in $B(z, \hat{a}) \cap \pi_{\hat{a}}(T(Y))$ is larger than $|B(z, \hat{a})|\cdot (1-\frac{1}{{a \choose \leq r}})$. As $|B(z, \hat{a})| = {a \choose \leq r}$, it follows that the size of the intersection is larger than ${a \choose \leq r}(1-\frac{1}{{a\choose \leq r}})={a\choose \leq r}-1$. Since the intersection size is an integer it must equal ${a \choose \leq r}$. In other words, $\pi_{\hat{a}}(T(Y))$ contains $B(z,r)$. 
\end{proof}

We continue with the proof of Lemma~\ref{lem:number_of_ind_polynomials-weak}.
The point of the last two lemmas is that for such a set $Y$, the set of polynomials that are defined on it is isomorphic to the set of polynomials defined on $T(Y)$ (also when considering degree $\leq r$ polynomials for both sets). Let us focus on $T(Y)$ and consider only polynomials in the variables $x_1,\ldots, x_{\hat{a}}$. As $\pi_{\hat{a}}(T(Y))$ contains a ball of radius $r$, we get that all monomials of degree $\leq r$ in $x_1,\ldots,x_{\hat{a}}$ are linearly independent on $\pi_{\hat{a}}(T(Y))$. However, the value of any such monomial on a point in $T(Y)$ is the same as its value on its projection. Thus, there are at least ${\hat{a} \choose \leq r}$ many linearly independent polynomials, of degree $\leq r$, that are defined on $T(Y)$. Hence, there are at least ${\hat{a}  \choose \leq r}$ many linearly independent polynomials, of degree $\leq r$, that are defined on $Y$.

In our case $|\V|>2^{m-t}$. Thus, in the notation above, $a=m-t$ and $\hat{a} =  a - 2\lceil\log({a \choose \leq r})\rceil = m-t - 2\lceil\log({m-t \choose \leq r})\rceil$. Thus, more than ${m-t - 2\lceil\log({m-t \choose \leq r})\rceil\choose \leq r}$ linearly independent polynomials of degree $\leq r$ that are defined on $\V$.

% assume that $|\V| = \e 2^m/{m \choose \leq r}$. Denote $t = \lfloor m - \log({m \choose \leq r}) + \log(\e)\rfloor$. By the discussion above we get that for $\hat{t} =  t - 2\lceil\log({t \choose \leq r})\rceil$, there are at least ${\hat{t} \choose \leq r}$ many linearly independent polynomials, of degree $\leq r$ that are defined on $|\V|$. We thus get that 
%${\hat{t} \choose \leq r} \leq s$. Thus,
%\begin{eqnarray*}
%s\geq {\hat{t} \choose \leq r} %&\geq & {\hat{t} \choose r} \\ 
%&=& {\lfloor m - \log({m \choose \leq r}) + \log(\e)\rfloor - 2\log \left({\lfloor m - \log({m \choose \leq r}) + \log(\e)\rfloor \choose \leq r} \right)\choose \leq r}\\
%&\geq & {m-3\lceil \log({m \choose \leq r}) + \log(\e)\rceil \choose \leq r}.
%\end{eqnarray*}
%In other words, if $s < {m-3\lceil \log{m \choose \leq r})+ \log(\e)\rceil \choose \leq r}$ then $|\V(\I(U)) < \e 2^m/{m \choose \leq r}$. This completes the proof of Lemma~\ref{lem: number of common zeroes-weak}.
\end{proof}

\end{document}